\ifdefined\TITWRAPPER
\else
\documentclass[11pt,a4paper]{article}

\usepackage[T1]{fontenc}
\usepackage[utf8]{inputenc}
\usepackage{lmodern}
\input{glyphtounicode}
\usepackage{amsmath,amssymb,amsthm,mathtools}
\usepackage{booktabs,array,tabularx}
\usepackage{placeins}
\usepackage{enumitem}
\usepackage[protrusion=false]{microtype}
\usepackage[a4paper,left=25.5mm,right=25.5mm,top=1.05in,bottom=1.05in,includefoot]{geometry}
\usepackage[hidelinks]{hyperref}

\allowdisplaybreaks
\makeatletter
\renewcommand\section{\@startsection{section}{1}{\z@}%
  {-3.5ex \@plus -1ex \@minus -.2ex}{2.3ex \@plus.2ex}%
  {\normalfont\Large\bfseries\raggedright}}
\makeatother

\hypersetup{
  pdftitle={Systematic Codes Correcting Two Deletions with Three Redundant Symbols over Large Alphabets},
  pdfauthor={Valery (Binyamin) Grishin; Aryeh Lev Zabokritskiy (Yohananov)},
  pdfsubject={Systematic deletion codes whose redundant symbols may also be deleted},
  pdflang={en-US},
  pdfdisplaydoctitle=true,
  pdfkeywords={deletion codes, systematic codes, nonbinary codes,
               synchronization errors}
}

\newtheorem{theorem}{Theorem}[section]
\newtheorem{lemma}[theorem]{Lemma}
\newtheorem{proposition}[theorem]{Proposition}

\theoremstyle{definition}
\newtheorem{definition}[theorem]{Definition}

\theoremstyle{remark}
\newtheorem{remark}[theorem]{Remark}

\newcommand{\C}{\mathcal C}
\newcommand{\F}{\mathbb F}
\newcommand{\calG}{\mathcal G}
\newcommand{\concat}{\mathbin\Vert}

\title{Systematic Codes Correcting Two Deletions\\
with Three Redundant Symbols over Large Alphabets}
\author{\normalsize Valery (Binyamin) Grishin\\[0.4em]
\small Department of Computer Science, Tel-Hai University of Kiryat Shmona\\
\small and the Galilee, Kiryat Shmona, Israel\\
\small\texttt{valerig.tech@gmail.com}\\[1em]
\and
\normalsize Aryeh Lev Zabokritskiy (Yohananov)\\[0.4em]
\small Department of Computer Science, MIGAL--Galilee Research Institute,\\
\small Kiryat Shmona, Israel; and Department of Computer Science, Tel-Hai\\
\small University of Kiryat Shmona and the Galilee, Kiryat Shmona, Israel\\
\small\texttt{yuhanalev@telhai.ac.il}\\[0.4em]
\small Corresponding author. ORCID: 0000-0003-3151-6192.}
\date{}
\begin{document}

\maketitle

\begin{abstract}
We construct systematic \(q\)-ary codes correcting two deletions with
three redundant symbols over polynomial-size alphabets.  For every
information length \(n\ge2\), our first construction has length \(n+3\)
over an alphabet of size \(q=2^s<2^{32}n^7\) for some integer \(s\), with polynomial-time
encoding and decoding.  A cardinality bound of Liu and Xing implies
that three redundant symbols are necessary, so the construction is
optimal in this measure.  We also establish the existence of systematic
codes with the same redundancy over alphabets of size \(O(n^4)\),
leaving efficient encoding and decoding at that alphabet scale open.  For a single
burst of at most two adjacent deletions, we give a polynomial-time
systematic construction with three redundant symbols over a power-of-two
alphabet of size \(q<(n+2)^4\).  A converse bound for adjacent deletions shows that three redundant
symbols are again optimal for \(n\ge2\).
\end{abstract}
\fi

\section{Introduction}
\label{sec:introduction}

Deletion-correcting codes recover information when symbols are removed
at unknown positions.  Since the classical work of Levenshtein and
Tenengolts~\cite{Levenshtein1966,Tenengolts1984}, a central objective has
been to construct such codes with small redundancy and efficient
encoding and decoding.  The case of two deletions has received
considerable attention, with binary and \(q\)-ary constructions achieving
redundancy logarithmic in the word length
\cite{GuruswamiHastad2021,SongCai2023,SimaGabrysBruck2020Systematic,
SimaGabrysBruck2020Qary,LiuTjuawinataXing2024}.
This literature includes systematic codes and several different
tradeoffs between redundancy and computational complexity~%
\mbox{\cite{SimaGabrysBruck2020Systematic,YeSunYuGeElishco2026,EnGad2026}}.

Most redundancy bounds are expressed in bits, often with the alphabet
held fixed, as for the \(q\)-ary two-deletion codes of~\cite{SongCai2023,LiuTjuawinataXing2024}.  When the alphabet grows with the word length, one can
instead minimize the number of redundant symbols.  This is the
appropriate measure when each transmitted symbol is a unit of storage
or transmission, for example a packet or a storage block, whose size
may grow with the word length.  Existing
large-alphabet constructions provide useful benchmarks.  For \(n\)
information symbols and \(q\ge n\), Sima, Gabrys, and
Bruck~\cite{SimaGabrysBruck2020Qary} give a polynomial-time construction
for two deletions whose redundancy is a constant number of alphabet
symbols.  Their construction is not systematic.  Schaller, Toesca, and
Vu~\cite{SchallerToescaVu2025} obtain strong bit-redundancy bounds for
large-alphabet codes whose codewords have pairwise distinct symbols;
efficient encoding remains open in their work.

We consider systematic encoders that append redundant symbols to an
arbitrary information word of length \(n\) over an alphabet of size \(q\).
Our aim is to minimize the number of appended symbols while retaining
polynomial-time encoding and decoding over an alphabet of polynomial~size.

For two arbitrary deletions, we attain the minimum of three redundant
symbols.  Specifically, for every \(n\ge2\), we construct a systematic
\(q\)-ary code of length \(n+3\), with polynomial-time encoding and
decoding, where
\(q=2^s\) and \(s=7\lceil\log_2 n\rceil+25\); in particular,
\(q<2^{32}n^7\).  Optimality follows from a bound of Liu and
Xing~\cite{LiuXing2023}: for \(n\ge2\), every two-deletion-correcting
code of size \(q^n\) has length at least \(n+3\).  Using Haxell's
independent-transversal theorem~\cite{Haxell1995,Haxell2001}, we also
prove that systematic
codes with three redundant symbols exist over alphabets of size
\(O(n^4)\); efficient encoding and decoding at this smaller alphabet scale
remain open.

At the alphabet sizes used by our efficient construction, we compare
with Sima--Gabrys--Bruck at the same information length and alphabet:
our construction uses three redundant symbols, fewer than their
guarantee, and it retains the information word as a
prefix.  Bit redundancy leads to a different
comparison.  At the same transmitted length and our alphabet scale,
the bit redundancy that Schaller, Toesca, and Vu guarantee for two
deletions is smaller than that of our construction.  Our guarantees instead include an unrestricted
information prefix and polynomial-time encoding and~decoding.

The related problem of correcting a burst of adjacent deletions also
has a substantial literature, beginning with Levenshtein's binary
construction~\cite{Levenshtein1967Adjacent}.  Later work treats binary
and \(q\)-ary alphabets, fixed and variable burst lengths, and
systematic encoding
\cite{SchoenyWachterZehGabrysYaakobi2017,WangTangSimaGabrysFarnoud2024,
NguyenCaiSiegel2024,SongCaiQuek2025,SunLuZhangGe2025}.
In particular, Nguyen, Cai, and Siegel give a systematic construction
for a burst whose location is known to lie within a window of
positions~\cite[Construction~8]{NguyenCaiSiegel2024}.

For a single burst of at most two adjacent deletions, we construct a
systematic code with three redundant symbols over a power-of-two
alphabet of size \(q<(n+2)^4\), with polynomial-time encoding and
decoding.  The burst may occur anywhere in the codeword, without
supplied information about its location.  Our construction for two arbitrary deletions also corrects such a burst, but over a larger alphabet.  Building on a known cardinality bound~\cite{WangKongYaakobiDuman2026,WangTangSimaGabrysFarnoud2024}, we prove that three redundant
symbols are optimal for \(n\ge2\) among all codes of size~\(q^n\).

\section{Model and benchmarks}
\label{sec:model}

This section fixes the definitions and proves four benchmarks: the lower
bound of three redundant symbols, the exclusion of
linear codes, the existence argument based on Haxell's theorem, and a
lower bound on the alphabet size.  Together they show that, for two arbitrary deletions, the number of
redundant symbols is settled by an existence argument, and that the issue addressed
here is polynomial-time encoding and decoding.  Table~\ref{tab:notation}
collects the main notation used in more than one section;
symbols that appear only inside one proof are defined there and are not
listed.

\subsection{Deletion descendants and systematic encoders}

Let \(\mathcal A\) be an alphabet of cardinality \(q\ge2\).  For a word
\(w\in\mathcal A^M\) and an integer \(0\le k\le M\), write
\[
 \mathsf D_k(w)=
 \bigl\{w_{i_1}\cdots w_{i_{M-k}}:
       1\le i_1<\cdots<i_{M-k}\le M\bigr\}
\]
for the set of its \emph{distinct} exact-\(k\)-deletion descendants; we also
call \(\mathsf D_k(w)\) the \emph{\(k\)-deletion ball} of \(w\).  A code
\(\C\subseteq\mathcal A^M\) corrects exactly \(k\) deletions when
\[
 \mathsf D_k(x)\cap\mathsf D_k(y)=\varnothing
 \qquad(x,y\in\C,\ x\ne y).
\]
It corrects up to \(k\) deletions when the same condition holds at every
deletion count from zero through~\(k\).

\begin{lemma}[Exact deletion correction is monotone]
\label{lem:exact-monotone}
If a fixed-length code corrects exactly \(k\) deletions, where \(k\) is at
most the block length, then it corrects up to \(k\) deletions.
\end{lemma}

\begin{proof}
Suppose distinct codewords had a common descendant after \(a<k\) deletions.
Deleting any further \(k-a\) coordinates from that common word would produce
a common exact-\(k\)-deletion descendant, a contradiction.
\end{proof}

Accordingly, an unqualified statement that a fixed-length code ``corrects
\(k\) deletions'' will mean correction up to \(k\) deletions; by
Lemma~\ref{lem:exact-monotone}, this is equivalent here to correction of
exactly \(k\) deletions.

\begin{definition}[Systematic encoder]
An encoder with \(n\) information symbols and \(r_{\mathrm{red}}\) redundant
symbols is \emph{systematic} if it has the form
\[
 E:\mathcal A^n\longrightarrow\mathcal A^{n+r_{\mathrm{red}}},
 \qquad E(x)=x\concat p(x),
\]
where \(p:\mathcal A^n\to\mathcal A^{r_{\mathrm{red}}}\).  Thus the
information symbols occur unchanged as a prefix.  We call \(p(x)\) the
\emph{redundancy block}; all of its symbols pass through the deletion channel
together with the~prefix.
\end{definition}

Some texts call a code systematic when the information symbols appear
unchanged in some fixed set of coordinates.  For deletion channels the
positions of those coordinates matter, because a deletion shifts every later
symbol.  We therefore always use the prefix form above; it is the form used
by all constructions and by every systematic converse in this paper
(Proposition~\ref{prop:appended-alphabet-lower},
Theorem~\ref{thm:adjacent-parity-prefix-converse} and the threshold of
Theorem~\ref{thm:adjacent-n4-exact-threshold}).

\emph{Computational convention.}
Running time is measured in bit operations.
Polynomial-time encoding and decoding for a family, possibly over alphabets
\(\mathcal A_n\) that vary with \(n\), means that one fixed pair of
deterministic algorithms works for every \(n\).  Given \(n\) and the
relevant input word, these algorithms construct all auxiliary data
determined by \(n\) and evaluate the encoder or decoder, both in time
polynomial in \(n\) and \(\log|\mathcal A_n|\) (for the families constructed in
this paper, \(|\mathcal A_n|\) is itself polynomial in \(n\), so this is time
polynomial in \(n\)).  In particular, every
auxiliary object used by the encoder (the representation of the finite
field, the total order on the alphabet, the maps that read fixed groups of bits of a symbol, and the
encoding of colors as bit strings) must itself be computed by the algorithm
from \(n\), rather than supplied as length-dependent advice of any size;
for example, no precomputed table of size exponential in \(n\) is allowed.
Data of this kind, determined by \(n\) alone and
known to both encoder and decoder, are called \emph{public}.

\emph{Canonical fields.}
Whenever a finite field of characteristic two is used, from
Section~\ref{sec:values} on, it is realized canonically as
follows.  For \(r\ge1\), let \(f_r(X)\) be the first monic irreducible
polynomial of degree \(r\) in \(\F_2[X]\), where polynomials are ordered by
the integer \(f(2)\), that is, by their coefficient strings read from the
highest degree down.  Realize
\[
 \F_{2^r}=\F_2[X]/(f_r(X)),
\]
and write \(\theta\) for the residue class of \(X\).  The polynomial basis
\(1,\theta,\ldots,\theta^{r-1}\) identifies \(\F_{2^r}\) with \(\F_2^r\): the
element \(\sum_{j=0}^{r-1}b_j\theta^j\) has the \emph{coordinate word}
\(\operatorname{coord}\bigl(\sum_jb_j\theta^j\bigr)=(b_0,\ldots,b_{r-1})
\in\{0,1\}^r\), with coefficients in increasing degree order, and the
\emph{binary label}
\(\operatorname{lab}\bigl(\sum_jb_j\theta^j\bigr)=\sum_jb_j2^j
\in\{0,\ldots,2^r-1\}\).  Field addition is bitwise XOR of coordinate
words.  The coordinate word of a word over \(\F_{2^r}\) is the concatenation, in order, of the coordinate words of its symbols.  Deterministic irreducibility testing is polynomial in \(r\), so an
exhaustive search through the at most \(2^r\) monic candidates finds
\(f_r\) in time polynomial in \(2^r\); every field used by the
encoders of Sections~\ref{sec:construction} and~\ref{sec:adjacent}, and
every auxiliary field of Appendix~\ref{app:two-round-coloring}, has order
polynomial in \(n\), so these realizations are public.
Which public total order on the field elements is used is stated where the
order is fixed.

Without the efficiency requirement, choosing a systematic code amounts to
choosing, for each message \(x\), one word from the set of all words that
begin with \(x\):
\[
 P_x=\{x\concat z:z\in\mathcal A^{r_{\mathrm{red}}}\},
 \qquad x\in\mathcal A^n.
\]
This choice of one word from each part \(P_x\) is used only to establish the Haxell benchmark
(Theorem~\ref{thm:haxell-systematic}) below.

\subsection{A known lower bound on the number of redundant symbols}

The lower bound stated in the Introduction is a direct specialization of
Theorem~3.6(ii) in the improved nonlinear Singleton bound of Liu and Xing
\cite{LiuXing2023}.  We include a short self-contained proof; it also shows
why the hypothesis \(n\ge2\) is needed.

\begin{theorem}[Known lower bound on the number of redundant symbols]
\label{thm:lower-bound}
Let \(q\ge2\), \(k\ge1\), and \(n\ge2\). A \(q\)-ary code of
cardinality \(q^n\) that corrects \(k\) deletions has block length at
least \(n+k+1\). In particular, every code of cardinality \(q^n\)
correcting two deletions uses at least three redundant symbols, whether or
not it is systematic.
\end{theorem}

\begin{proof}[Proof of Theorem~\ref{thm:lower-bound}]
Let \(\C\) be such a code, and let \(M\) be its block length.  Every
\(k\)-deletion ball is nonempty and lies in \(\mathcal A^{M-k}\).  Pairwise disjointness gives
\[
 q^n=|\C|\le q^{M-k},
\]
so \(M\ge n+k\).

Assume for a contradiction that \(M=n+k\).  The \(q^n\) nonempty, pairwise
disjoint sets \(\mathsf D_k(c)\), for \(c\in\C\), lie in the ambient set
\(\mathcal A^n\), which also has size \(q^n\).  Hence every
\(\mathsf D_k(c)\) is a singleton.

Fix \(c=c_1\cdots c_{n+k}\in\C\).  For every \(0\le i\le k\), the consecutive
window
\[
 c_{i+1}\cdots c_{i+n}
\]
is an exact-\(k\)-deletion descendant: delete the first \(i\) and last
\(k-i\) coordinates.  All these windows must be equal.  Equality of the
windows for \(i\) and \(i+1\) forces
\[
 c_{i+1}=c_{i+2}=\cdots=c_{i+n+1}.
\]
As \(i\) ranges from zero to \(k-1\), these chains cover every adjacent pair
of coordinates of \(c\).  Thus \(c\) is constant.  There are only \(q\)
constant words, contradicting \(|\C|=q^n>q\) for \(n\ge2\).
\end{proof}

The restriction \(n\ge2\) is needed: for \(n=1\), the \(q\) repetition
words \(\{a^{k+1}:a\in\mathcal A\}\) correct \(k\) deletions with only \(k\)
redundant symbols.

\subsection{A linear-code obstruction}

For \(n\ge3\), the target parameters cannot be obtained from a linear code,
regardless of whether that code is displayed systematically.

\begin{proposition}[Half-Singleton specialization]
\label{prop:linear-obstruction}
For every \(n\ge3\), no linear \([n+3,n]_q\) code corrects two deletions.
In particular, no Reed--Solomon code has the target parameters.
\end{proposition}

\begin{proof}
The half-Singleton bound of Cheng, Guruswami, Haeupler, and Li
\cite[arXiv version, Corollary~5.2 and its direct proof]{ChengGuruswamiHaeuplerLi2023}
is an asymptotic rate bound.  Its direct proof gives the following
finite-length form for \(1\le K\le(N+1)/2\): a linear \([N,K]_q\)
code corrects at most \(N-2K+1\) insertions and deletions in total;
equivalently, some two distinct codewords have a common subsequence of
length at least \(2K-2\).  This is the form restated in
\cite[arXiv version, Theorem~4]{BeelenEtAl2026}, used here with the
stated range restriction.  In particular, two-deletion correction requires
\[
  2\le N-2K+1.
\]
For \(N=6\) and \(K=3\) the right side is \(1\), so no linear \([6,3]_q\)
code corrects two deletions.  Direct substitution of \(N=n+3\) and \(K=n\)
in the quoted formulation falls outside its stated range when \(n\ge5\):
it would give the negative bound \(4-n\), which is false, since every code
corrects zero errors.  We therefore use the length-six instance to cover
every \(n\ge3\).  Now let
\(\C\) be a linear \([n+3,n]_q\) code with \(n\ge3\).  Its codewords whose
first \(n-3\) coordinates are zero form a subspace of dimension at least
three, on which deleting these coordinates is injective.  Applying this
deletion to a three-dimensional subspace of it gives a linear \([6,3]_q\)
code \(\C'\).  If \(\C\) corrected two deletions, so would \(\C'\): a common
two-deletion descendant \(w\) of two distinct codewords of \(\C'\) would
give the common two-deletion descendant \(0^{n-3}w\) of the corresponding
codewords of \(\C\).  Hence \(\C\) does not correct two deletions.
\end{proof}

Linear checks such as Reed--Solomon parities can still be useful once the
deletion positions are known, as for erasures; the construction for one
adjacent burst in Section~\ref{sec:adjacent} uses its two parity symbols in
this way (Lemma~\ref{lem:adjacent-parity-inversion}).  An encoder with three
redundant symbols, however, cannot be \(\F_q\)-linear from \(n=3\) onward;
at \(n=2\), where \(N=5\) and \(K=2\), the quoted bound \(N-2K+1\) equals
\(2\) and does not exclude two-deletion correction.

\subsection{An existence benchmark from independent transversals}

For total length \(M\), define the exact-\(k\)-deletion confusion graph
\(G_{M,k,q}\) on \(\mathcal A^M\): two distinct words are adjacent when their
\(k\)-deletion balls intersect.  Put
\[
 B=\binom{M}{k}.
\]
Each vertex has at most \(B\) distinct deletion outputs.  A fixed output of
length \(M-k\) has at most \(Bq^k\) length-\(M\) supersequences, because one may
choose the insertion coordinates and their values.  Therefore
\begin{equation}
 \Delta(G_{M,k,q})\le B(Bq^k-1).
 \label{eq:crude-conflict-degree}
\end{equation}

Haxell's independent-transversal theorem states that if a graph of maximum
degree at most \(d\) has its vertex set partitioned into parts of size at
least \(2d\), then it has an independent set meeting every part
once~\cite{Haxell1995,Haxell2001}; we use only the weaker statement for
parts of size strictly greater than \(2d\).
We use Haxell's theorem as a black box; the graph to which we apply it, and
the degree bound \eqref{eq:crude-conflict-degree}, are specific to deletion
codes.

\begin{theorem}[Existential systematic benchmark]
\label{thm:haxell-systematic}
Let \(k\ge1\), \(n\ge1\), and \(M=n+k+1\).  If
\[
 q\ge2\binom{M}{k}^{2},
\]
then there is a systematic code
\(E:\mathcal A^n\to\mathcal A^M\) that corrects up to \(k\) deletions and
uses exactly \(k+1\) redundant symbols.
\end{theorem}

\begin{proof}
Partition \(\mathcal A^M\) into the prefix parts
\[
 P_x=\{x\concat z:z\in\mathcal A^{k+1}\},
 \qquad x\in\mathcal A^n.
\]
Every part has size \(q^{k+1}\).  The assumed inequality and
\eqref{eq:crude-conflict-degree} give
\[
 q^{k+1}
 \ge2B^2q^k
 >2B(Bq^k-1)
 \ge2\Delta(G_{M,k,q}).
\]
Haxell's theorem supplies an independent transversal.  Its representatives
form a systematic code correcting exactly \(k\) deletions, and
Lemma~\ref{lem:exact-monotone} gives correction up to \(k\).
\end{proof}

For \(k=2\), this theorem gives three redundant symbols whenever
\[
 q\ge2\binom{n+3}{2}^2=O(n^4).
\]
Requiring \(q\) to be a power of two changes this estimate by a factor of
less than two: take the least such power above the displayed threshold.
This existence argument uses a graph with \(q^{n+3}\) vertices and chooses
the redundant symbols of all messages jointly; it supplies no
polynomial-time algorithm for computing \(p(x)\) at a single message.
The present work additionally requires polynomial-time encoding and decoding
under the computational convention above.

\subsection{A lower bound on the alphabet size of systematic codes}

The preceding benchmarks concern the number of redundant symbols.  A
counting argument also bounds from below the alphabet of every systematic
code with \(k+1\) redundant symbols; it is the comparison point for the
alphabet sizes of Theorems~\ref{thm:haxell-systematic} and~\ref{thm:main}.
A fixed word of length \(M-k\) has exactly
\begin{equation}
 \mathsf I_k(M,q)=\sum_{j=0}^k\binom Mj(q-1)^j
 \label{eq:insertion-sphere}
\end{equation}
distinct length-\(M\) supersequences.  This standard insertion-ball count
follows by the usual last-symbol recurrence; its binary ancestor goes back
to Levenshtein~\cite{Levenshtein1966}.

\begin{proposition}[Necessary alphabet size for systematic codes with \(k+1\) redundant symbols]
\label{prop:appended-alphabet-lower}
Let \(n\ge k\ge1\).  If
\[
 E(x)=x\concat p(x),
 \qquad p:\mathcal A^n\longrightarrow\mathcal A^{k+1},
\]
corrects \(k\) deletions, then
\[
 q^{k+1}\ge\mathsf I_k(n,q).
\]
\end{proposition}

\begin{proof}
Fix \(y\in\mathcal A^{n-k}\).  If distinct length-\(n\) supersequences \(x,x'\)
of \(y\) had \(p(x)=p(x')\), deleting \(k\) information symbols from both codewords
would produce the common descendant \(y\concat p(x)\).  Thus \(p\) is
injective on the set of length-\(n\) supersequences of \(y\), which has exactly
\(\mathsf I_k(n,q)\) elements by \eqref{eq:insertion-sphere}, while \(p\)
takes at most \(q^{k+1}\) values.
\end{proof}

At \(k=2\), the proposition gives
\[
 q^3\ge1+n(q-1)+\binom n2(q-1)^2,
\]
so \(q^3\ge\binom n2(q-1)^2\).  Since
\((q-1)^2\ge q^2/4\) for \(q\ge2\), it follows that
\(q\ge\tfrac14\binom n2=\Omega(n^2)\).  Thus, for systematic codes with
three redundant symbols, the alphabet must grow at least quadratically;
Haxell's theorem gives existence at a quartic size, and
Theorem~\ref{thm:main} gives a construction with polynomial-time encoding
and decoding at a larger polynomial size.  Only the number of redundant
symbols is settled optimally here.

\begin{table}[tp]
\caption{Notation used across sections.}
\label{tab:notation}
\small
\renewcommand{\arraystretch}{1.1}
\begin{tabularx}{\textwidth}{@{}>{\raggedright\arraybackslash}p{0.24\textwidth}
  >{\raggedright\arraybackslash}X
  >{\raggedright\arraybackslash}p{0.14\textwidth}@{}}
\toprule
symbol & meaning & introduced \\
\midrule
\(n\), \(q\), \(\mathcal A\) & number of information symbols; alphabet size;
  alphabet & Sec.~\ref{sec:model} \\
\(r_{\mathrm{red}}\) & number of redundant symbols & Sec.~\ref{sec:model} \\
\(\mathsf D_k(w)\) & set of words obtained from \(w\) by exactly \(k\)
  deletions (the \(k\)-deletion ball) & Sec.~\ref{sec:model} \\
\(x\concat p(x)\) & systematic codeword: information prefix \(x\), redundancy
  block \(p(x)\) & Sec.~\ref{sec:model} \\
\(\Gamma(x)\), \(\pi_{\mathrm{tag}}\), \(\pi_{\mathrm{mark}}\) & label chosen for \(x\);
  maps reading the label bits and the position-marker bits of a symbol
  & Sec.~\ref{sec:boundary} \\
\(S(x)\), \(V(x)\) & sum of the symbols of \(x\); the second symmetric check
  \(\sigma_2(x)+S(x)^2\) & Sec.~\ref{sec:values} \\
\(\preceq\), \(\alpha(x)\), \(T_n(x)\) & public total order on
  \(\F_q\); signature; Tenengolts checksum & Sec.~\ref{sec:values} \\
\(u\), \(N\), \(s\) & \(u=\lceil\log_2n\rceil\), \(N=2^u\), and
  \(q=2^s\) & Secs.~\ref{sec:graph}, \ref{sec:construction} \\
\(h_S\), \(h_V\) & numbers of leading bits of \(S\) and \(V\) stored in
  redundant symbol~\(1\) & Secs.~\ref{sec:graph}, \ref{sec:construction} \\
\(d_{\mathrm{tag}}\), \(\delta\), \(\rho\), \(b_{\mathrm{col}}\) & label bits;
  label and marker bits; payload bits per symbol; bits for \(\ell_1\)
  & Sec.~\ref{sec:construction} \\
\(\ell_{\mathrm{name}}\), \(K_{\mathrm{col}}\) & length of vertex names;
  number of colors & Sec.~\ref{sec:graph} \\
\(S_{\mathrm{head}},S_{\mathrm{tail}},V_{\mathrm{head}},V_{\mathrm{tail}}\)
  & leading and trailing bit strings of \(S(x)\) and \(V(x)\)
  & Secs.~\ref{sec:graph}, \ref{sec:construction} \\
\(G_V\), \(G_H\), \(J_2^T\), \(G_\star\) & confusion graphs on \(\F_q^n\) and
  their union & Sec.~\ref{sec:graph} \\
\(\ell_0(x)\), \(\ell_1(x)\) & bit string of \(T_n(x)\); proper color of
  \(G_\star\) & Secs.~\ref{sec:graph}, \ref{sec:construction} \\
\(w_1,w_2,w_3\), \(p_1,p_2,p_3\) & the three payloads and the three
  redundant symbols & Sec.~\ref{sec:construction} \\
\(\mathsf B_{r,d}\) & deletion of \(r\) adjacent symbols starting at
  position \(d\) & Sec.~\ref{sec:adjacent} \\
\(A(x)\), \(B(x)\), \(W(x)\), \(H(x)\) & the two parity symbols, the parity
  prefix \((x,A(x),B(x))\), and the final symbol of the adjacent
  construction & Sec.~\ref{sec:adjacent} \\
\(P_{r,d}(x)\) & \(W(x)\) after deleting \(r\) adjacent symbols from
  position \(d\) & Sec.~\ref{sec:adjacent} \\
\(\Lambda(x)\), \(\iota\), \(R_n\) & triple of Tenengolts checksums;
  mixed-radix index; number of candidates for \(H(x)\) & Sec.~\ref{sec:adjacent} \\
\(U_x\), \(\kappa_n\) & values the final symbol must avoid (avoidance
  set); upper bound on \(|U_x|\) & Sec.~\ref{sec:adjacent} \\
\(W_{n,t}(x)\) & information symbols followed by one XOR parity per residue
  class modulo \(t\) & Sec.~\ref{sec:adjacent-converses} \\
\(\mathcal L_x\), \(\calG_{n,q}\), \(G_m\) & exact list of the final
  values that collide with a distinct message; graph joining messages
  whose parity prefixes have a common burst descendant;
  \(G_m=\calG_{4,2^m}\) & Secs.~\ref{sec:adjacent}, \ref{sec:adjacent-n4} \\
\bottomrule
\end{tabularx}
\end{table}

\section{Locating the end of the information part}
\label{sec:boundary}

A codeword \(x\concat p(x)\) consists of the information prefix \(x\)
followed by the redundancy block \(p(x)\).  After deletions, the received
word no longer shows where the prefix ends.  Two codewords could therefore
produce the same received word even though different numbers of their
information symbols were deleted.  This section shows how to exclude such
collisions with a few reserved bits in every redundant symbol, without using
any of the bits that later carry the deletion checks.  It then reduces the
correction of two deletions to three cases, according to how many
information symbols are deleted.  The decoder of
Section~\ref{subsec:algorithms} does not search for this boundary, and it
reads neither the labels nor the position markers introduced below: it
tries each possible number of deleted
information symbols and re-encodes the resulting candidates.  The
results of this section are used instead in the correction proof of
Section~\ref{sec:construction}, which shows that at most one candidate is
consistent with the received word.

If \(\mathcal X\) and \(\mathcal Y\) are sets of words, write
\(\mathcal X\mathcal Y=\{a\concat b:a\in\mathcal X,\ b\in\mathcal Y\}\).  We use the deletion-ball notation
from Section~\ref{sec:model}.  The first lemma allows a redundancy block
\(p\) of arbitrary length \(r_{\mathrm{red}}\); the construction later uses
\(r_{\mathrm{red}}=3\).

\begin{lemma}[Splitting deletions between prefix and redundancy block]
\label{lem:data-trailer-factorization}
Let \(x\in\mathcal A^n\), let \(p\in\mathcal A^{r_{\mathrm{red}}}\), and let
\(0\le k\le n+r_{\mathrm{red}}\).  Then
\[
 \mathsf D_k(x\concat p)=
 \bigcup_{\max\{0,k-r_{\mathrm{red}}\}\le a\le\min\{k,n\}}
       \mathsf D_a(x)\,\mathsf D_{k-a}(p).
\]
\end{lemma}

\begin{proof}
Every deletion pattern removes a uniquely determined number \(a\) of
information coordinates and \(k-a\) redundant coordinates.  Deletion
preserves order, and the whole prefix precedes the redundancy block, so the
resulting word lies in the displayed concatenation.  Conversely, any pair of
descendants in one term of the union is realized by applying the two
deletion patterns in their respective blocks.
\end{proof}

\subsection{A label absent from the last two information symbols}

Let \(C_{\mathrm{tag}}\) be a finite set of labels, and let
\[
 \pi_{\mathrm{tag}}:\mathcal A\longrightarrow C_{\mathrm{tag}}
\]
be a fixed public map assigning a label to each alphabet symbol.  In the
construction, \(\pi_{\mathrm{tag}}\) reads a fixed group of bits of a symbol.
Suppose each information word \(x\) is assigned a value
\(\Gamma(x)\in C_{\mathrm{tag}}\) different from the labels of its last two
symbols:
\begin{equation}
 \Gamma(x)\notin
 \{\pi_{\mathrm{tag}}(x_i):\max\{1,n-1\}\le i\le n\}.
 \label{eq:tag-absent}
\end{equation}
Every redundant symbol will be assigned this same label by the fixed map.

\begin{theorem}[A common label absent from the last two information symbols fixes the split]
\label{thm:absent-label-boundary}
Let \(p(x)=p_1(x)\cdots p_{r_{\mathrm{red}}}(x)\) satisfy
\[
 \pi_{\mathrm{tag}}(p_j(x))=\Gamma(x)
 \qquad(1\le j\le r_{\mathrm{red}})
\]
for every message \(x\), where \(\Gamma(x)\) satisfies \eqref{eq:tag-absent}.
Fix a deletion count \(k\le2\) and assume
\(r_{\mathrm{red}}\ge k+1\).  Suppose that, for messages \(x\) and \(y\) that need not be distinct, some pattern of \(k\) deletions
applied to \(x\concat p(x)\) and some pattern of \(k\) deletions applied to
\(y\concat p(y)\) produce the same word.  Then the two patterns delete the
same number of information symbols.
\end{theorem}

The case \(x=y\) is included so that a received word alone determines how
many information symbols were deleted, even when the two patterns act on
the same codeword; the decoder of Section~\ref{subsec:algorithms} reads no labels, but its analysis there uses this case, and the
proof of Theorem~\ref{thm:equal-split-reduction} uses only the case
\(x\ne y\).

\begin{proof}
Let the two patterns delete \(a\) and \(b\) information symbols,
respectively, and let \(w\) be the common received word.  Suppose for a
contradiction that \(a>b\), and put \(\delta_{\mathrm{split}}=a-b>0\).  In
\(w\), the surviving prefix on the \(x\)-side has length \(n-a\), whereas
that on the \(y\)-side has the larger length
\(n-b=n-a+\delta_{\mathrm{split}}\).  Split \(w\) as
\[
 w=\xi\concat\omega\concat v,
 \qquad |\xi|=n-a,\quad |\omega|=\delta_{\mathrm{split}},\quad
 |v|=r_{\mathrm{red}}-k+b.
\]
On the \(x\)-side, \(\xi\) is the surviving information part and
\(\omega\concat v\) is the surviving redundancy block.  On the \(y\)-side,
\(\xi\concat\omega\) is the surviving information part and \(v\) is the
surviving redundancy block.  Here
\[
 |v|=r_{\mathrm{red}}-k+b\ge r_{\mathrm{red}}-k\ge1.
\]
The symbols of \(v\) are redundant symbols on both sides, so projecting one
of them gives \(\Gamma(x)=\Gamma(y)\).  Every symbol of the nonempty word
\(\omega\), viewed on the \(x\)-side, is a redundant symbol and has label
\(\Gamma(x)\).
Viewed from the \(y\)-factorization, the same symbols are the last
\(\delta_{\mathrm{split}}\) surviving information symbols of \(y\).  Each of
them, say \(y_i\), is followed in \(y\) by at most
\(\delta_{\mathrm{split}}-1\) surviving and at most \(b\) deleted information
symbols, so \(n-i\le(\delta_{\mathrm{split}}-1)+b=a-1\), that is, \(i\ge n-a+1\ge n-1\), because \(a\le k\le2\).  Hence
\(\Gamma(y)=\Gamma(x)\) is the label of one of the last two information symbols
of \(y\), contradicting \eqref{eq:tag-absent}.  The case \(b>a\) is
symmetric.
\end{proof}

Such a label is easy to compute, and two bits of each redundant symbol suffice for it, since they give \(4>2\) distinct labels.  Give \(C_{\mathrm{tag}}\) a public
enumeration, and use the induced total order.  Whenever
\(|C_{\mathrm{tag}}|>2\), the deterministic choice
\begin{equation}
 \Gamma(x)=
 \min\Bigl(C_{\mathrm{tag}}\setminus
   \{\pi_{\mathrm{tag}}(x_i):\max\{1,n-1\}\le i\le n\}\Bigr)
 \label{eq:least-absent-tag}
\end{equation}
is well defined.  To compute it, mark the at most two labels of the last two
information symbols and scan the enumeration until its first unmarked
element; the first three entries suffice.

\subsection{Reduction to three cases}

By Theorem~\ref{thm:absent-label-boundary}, a collision between two
codewords can occur only when both deletion patterns remove the same number
\(a\) of information symbols.  We now formulate a condition that excludes
these remaining, equal-split collisions.  For \(n\ge2\), a map
\(p:\mathcal A^n\to\mathcal A^3\) is said to satisfy the
\emph{two-deletion equal-split condition} if, for every
\(a\in\{0,1,2\}\) and all distinct \(x,y\in\mathcal A^n\),
\begin{equation}
 \mathsf D_a(x)\cap\mathsf D_a(y)\ne\varnothing
 \quad\Longrightarrow\quad
 \mathsf D_{2-a}(p(x))\cap\mathsf D_{2-a}(p(y))
 =\varnothing.
 \label{eq:equal-split-condition}
\end{equation}
In words: if \(x\) and \(y\) can be confused after deleting \(a\) symbols
from each, then their redundancy blocks cannot be confused after deleting
the remaining \(2-a\) symbols from each.  (For \(a=0\) the hypothesis forces
\(x=y\), so the condition is vacuous.)

\begin{theorem}[Equal-split reduction for two deletions]
\label{thm:equal-split-reduction}
Assume that every redundant symbol has the common label specified in
Theorem~\ref{thm:absent-label-boundary}.  Then the systematic code
\[
 \{x\concat p(x):x\in\mathcal A^n\}
\]
corrects up to two deletions if and only if
\eqref{eq:equal-split-condition} holds.
\end{theorem}

\begin{proof}
For necessity, if distinct \(x,y\) share an exact-\(a\)-deletion
descendant and their redundancy blocks share an exact-\((2-a)\)-deletion
descendant, concatenating the two common words gives a common
exact-two-deletion descendant of the codewords.

For sufficiency, by Lemma~\ref{lem:exact-monotone} it is enough to exclude
a common exact-two-deletion descendant of distinct codewords.  Suppose two
distinct codewords had one.  Applying
Theorem~\ref{thm:absent-label-boundary} with deletion count two forces the
two patterns to delete the same number \(a\) of information symbols.  Since the surviving
prefixes then have the same length \(n-a\) on both sides,
Lemma~\ref{lem:data-trailer-factorization} gives common descendants in both
\(\mathsf D_a(x)\cap\mathsf D_a(y)\) and
\(\mathsf D_{2-a}(p(x))\cap\mathsf D_{2-a}(p(y))\), contradicting
\eqref{eq:equal-split-condition}.
\end{proof}

The redundant symbols will also carry fixed public position markers.
Formally, let \(\pi_{\mathrm{mark}}:\mathcal A\to\{0,1,2,3\}\) be public (in the
construction it reads two bits of a symbol; the three redundant symbols use only the values \(1\), \(2\) and \(3\), and information symbols may carry any value) and require
\begin{equation}
 \pi_{\mathrm{mark}}(p_j(x))=j,
 \qquad j=1,2,3.
 \label{eq:marker-convention}
\end{equation}
If one redundant symbol is deleted, the surviving ordered marker pair is
\((2,3)\), \((1,3)\), or \((1,2)\), according as symbol \(1\), \(2\), or
\(3\) was deleted.  The bits of a redundant symbol that are not used by the
label or the marker form its \emph{payload}; the payloads carry the deletion
checks.  The bit widths are fixed in Section~\ref{sec:construction}.  Thus the marker pair determines which two redundant symbols,
and hence which two payloads, survived.

If the redundancy blocks of two messages have a common one-deletion
descendant, its two symbols carry the same marker pair on both sides, so the
same redundant symbol is missing on both sides and the two surviving
payloads of one message equal those of the other.  Condition
\eqref{eq:equal-split-condition} therefore splits into the
following three cases, where \(a\) is the number of deleted information
symbols; in the last two rows the payloads alone are asked to separate the
messages, which suffices for \eqref{eq:equal-split-condition}, and any
separation by the labels is not used.

\begin{center}
\begin{tabularx}{\textwidth}{@{}c c >{\raggedright\arraybackslash}X@{}}
\toprule
deleted information symbols & deleted redundant symbols & what remains to be shown \\
\midrule
0 & 2 & nothing: the intact prefix identifies the message \\ \addlinespace
1 & 1 & the marker pair identifies which redundant symbol is missing;
        for each of the three possibilities, the two surviving payloads must
        distinguish any two messages sharing a one-deletion descendant \\ \addlinespace
2 & 0 & all three payloads survive and must distinguish any two messages
        sharing a two-deletion descendant \\
\bottomrule
\end{tabularx}
\end{center}

Splits with unequal numbers of deleted information symbols do not occur, by
Theorem~\ref{thm:absent-label-boundary}.  Sections~\ref{sec:values}
to~\ref{sec:construction} construct payloads that meet the obligations in
the second and third rows and pack them into the three marked redundant
symbols.

\section{Deleted-value recovery and reconstruction after one deletion}
\label{sec:values}

This section provides the checks used in the payloads.  Two symmetric
functions of the information symbols recover the values of up to two deleted
symbols (Lemmas~\ref{lem:one-deletion-nonlinear}
and~\ref{lem:two-deletion-nonlinear}).  A Tenengolts checksum, together
with the sum of the symbols and the received one-deletion descendant,
then determines the original word
(Proposition~\ref{prop:tenengolts-one-deletion}).  The remaining
ambiguities that these checks leave open are treated in
Section~\ref{sec:graph}.

Throughout this section let \(\F_q\) be a finite field of characteristic
two, so \(q\) is a power of two.  Characteristic two is used twice: squaring
is additive and bijective, which the recovery of deleted values below
exploits, and a power-of-two alphabet lets each symbol be split into bit
fields in Section~\ref{sec:construction}.  Fix a public \(\F_2\)-basis of
\(\F_q\) (the construction uses the polynomial basis of
Section~\ref{sec:model}), and give \(\F_q\) a fixed public total order
\(\preceq\).  The order
will be used only for comparisons and need not be compatible with the field
operations.

For every finite word \(z=(z_1,\ldots,z_m)\in\F_q^m\), including the empty
word when \(m=0\), define
\begin{equation}
 S(z)=\sum_{i=1}^m z_i,
 \qquad
 \sigma_2(z)=\sum_{1\le i<j\le m}z_iz_j,
 \qquad
 V(z)=\sigma_2(z)+S(z)^2.
 \label{eq:nonlinear-checks}
\end{equation}
Thus \(S=\sigma_1\), and \(\sigma_2\) is the second elementary symmetric
function.  Since \(\sigma_2=V+S^2\), the pairs \((S,V)\) and \((S,\sigma_2)\)
carry the same information.  We store \(V\) rather than \(\sigma_2\) for
the following reason.  In one case of the construction (redundant
symbol~\(2\) deleted), the surviving redundant symbols carry the full value of \(V\) but only
part of \(S\).  By the second claim of Lemma~\ref{lem:one-deletion-nonlinear} below,
\(V\) alone still leaves at most two candidates for a single deleted value.
The value of \(\sigma_2\) alone would not suffice: if the received word has
sum zero, then \(\sigma_2\) does not depend on the deleted value at all.

\subsection{Recovering deleted field values}

\begin{lemma}[One deleted value]
\label{lem:one-deletion-nonlinear}
Let \(w\in\mathsf D_1(x)\) and let \(a\) be the deleted
field value.  Then \(w\) and \(S(x)\) determine \(a\) uniquely.  Moreover,
\(V(x)\) depends only on \(w\) and \(a\), and for fixed \(w\) and any value \(v\), at most two values of \(a\) give \(V(x)=v\).
\end{lemma}

\begin{proof}
Put \(R=S(w)\).  The sum identity is
\[
 S(x)=R+a,
\]
so \(a=S(x)+R\), where subtraction equals addition in characteristic two.
Moreover,
\[
 \sigma_2(x)=\sigma_2(w)+Ra
\]
and \((R+a)^2=R^2+a^2\).  Hence
\begin{equation}
 V(x)=\sigma_2(w)+R^2+a^2+Ra.
 \label{eq:one-deletion-V}
\end{equation}
For fixed \(w\) and \(V(x)\), this prescribes the value of the
\(\F_2\)-linear map
\[
 L_R(a)=a^2+Ra.
\]
If \(R=0\), its kernel is \(\{0\}\).  If \(R\ne0\), the equation
\(a^2+Ra=a(a+R)=0\) shows that its kernel is \(\{0,R\}\).  Every nonempty
fiber of a linear map is a coset of its kernel, so
\eqref{eq:one-deletion-V} has at most two solutions.
\end{proof}

\begin{lemma}[Two deleted values]
\label{lem:two-deletion-nonlinear}
Let \(w\in\mathsf D_2(x)\) and let \(\{a,b\}\) be the unordered multiset of
deleted values.  Then \(w\), \(S(x)\), and \(V(x)\)
determine \(\{a,b\}\) uniquely, including when \(a=b\).
\end{lemma}

\begin{proof}
Put \(R=S(w)\).  First,
\begin{equation}
 \eta:=a+b=S(x)+R.
 \label{eq:two-deletion-sum}
\end{equation}
The pairs contributing to \(\sigma_2(x)\) split into pairs inside \(w\), pairs
joining one deleted value to one symbol of \(w\), and the pair \(ab\).
Consequently
\[
 \sigma_2(x)=\sigma_2(w)+\eta R+ab.
\]
Since \(\sigma_2(x)=V(x)+S(x)^2\) in characteristic two, the product is
\begin{equation}
 ab=V(x)+S(x)^2+\sigma_2(w)+\eta R.
 \label{eq:two-deletion-product}
\end{equation}
The deleted multiset is therefore the root multiset of
\[
 Z^2+\eta Z+ab.
\]
This also covers a repeated deleted value.  If \(a=b\), then \(\eta=0\) and
\(ab=a^2\); Frobenius squaring is a bijection of the finite field \(\F_q\),
so the repeated root is determined uniquely.
\end{proof}

The recovery in both lemmas is algorithmic in the public binary basis.  An
equation \(a^2+Ra=c\) is inverted directly when \(R=0\); when \(R\ne0\), the
substitution \(a=Ry\) reduces it to the \(\F_2\)-linear system
\(y^2+y=c/R^2\).  The quadratic in Lemma~\ref{lem:two-deletion-nonlinear}
has the same form after scaling when \(\eta\ne0\), while \(\eta=0\) is
inverted by the Frobenius automorphism.  Gaussian elimination over
\(\F_2\) therefore lists the relevant value or values in time
polynomial in \(\log q\).

\subsection{A Tenengolts checksum for one deletion}

Suppose two messages share a one-deletion descendant and have the same sum
\(S\).  Then the deleted value is the same on both sides, so both messages
are obtained by inserting this value into the common descendant.  To
distinguish the resulting distinct words, we record, for each pair of neighboring symbols,
whether the sequence weakly increases there.  This turns the problem into a
binary one, which the Varshamov--Tenengolts (VT) checksum solves:
Levenshtein~\cite{Levenshtein1966} showed that the VT codes correct one
deletion.

For every \(m\ge1\) and \(x\in\F_q^m\), define the signature
\(\alpha(x)\in\{0,1\}^{m-1}\) of \(x\) by
\[
 \alpha(x)_i=\mathbf 1[x_i\preceq x_{i+1}],
 \qquad 1\le i<m,
\]
and the Tenengolts checksum of \(x\) by
\begin{equation}
 T_m(x)=\sum_{i=1}^{m-1}i\,\alpha(x)_i\pmod m.
 \label{eq:tenengolts-checksum}
\end{equation}
The subscript records the length of the word, which is also the modulus.
When \(m=1\), the word \(\alpha(x)\) is empty and \(T_1(x)=0\).  Messages
\(x\in\F_q^n\) with \(n\ge2\) use \(T_n(x)\); Section~\ref{sec:adjacent}
applies the same checksum to words of other lengths.
The signature \(\alpha\) depends only on the total order \(\preceq\),
and \(S\) only on field addition; no compatibility between the two is
needed.  This argument adapts the nonbinary single-deletion checksum of
Tenengolts \cite{Tenengolts1984} directly to an ordered alphabet.  We include
the exact argument needed below.

\begin{lemma}[Binary VT separation; Levenshtein~\cite{Levenshtein1966}]
\label{lem:binary-vt-separation}
Let \(c,c'\in\{0,1\}^{\ell_{\mathrm{VT}}}\) be distinct words with a
common one-deletion descendant.  Then
\[
 \sum_{r=1}^{\ell_{\mathrm{VT}}} r c_r\not\equiv
 \sum_{r=1}^{\ell_{\mathrm{VT}}} r c'_r\pmod{\ell_{\mathrm{VT}}+1}.
\]
\end{lemma}

The standard proof is reproduced in Appendix~\ref{app:vt-proof} for
completeness.

\begin{lemma}[Deleting a symbol deletes one signature bit]
\label{lem:ascent-deletion}
Let \(m\ge2\), let \(x\in\F_q^m\), and let
\(z\in\mathsf D_1(x)\).  Then
\(\alpha(z)\in\mathsf D_1(\alpha(x))\).
\end{lemma}

\begin{proof}
Deleting an endpoint removes the corresponding endpoint comparison.  For an
internal deletion, write the three local symbols as \(r,a,s\).  The two old
comparison bits and the new one are
\[
 c=\mathbf 1[r\preceq a],\qquad
 d=\mathbf 1[a\preceq s],\qquad
 e=\mathbf 1[r\preceq s].
\]
If \(c=d\), transitivity of the total order gives \(e=c=d\).  If
\(c\ne d\), the binary bit \(e\) equals one of \(c,d\).  Deleting the other
old bit therefore gives the new signature.
\end{proof}

\begin{lemma}[The recovered symbol and signature determine the word]
\label{lem:ascent-insertion-unique}
Fix \(z\in\F_q^{n-1}\) and \(a\in\F_q\).  Distinct words obtained by inserting
\(a\) once into \(z\) have distinct signatures.
\end{lemma}

\begin{proof}
Two distinct words obtained by inserting \(a\) into \(z\) can be written as
\[
 u=P\,a\,v_1\cdots v_m\,Q,\qquad u'=P\,v_1\cdots v_m\,a\,Q
\]
for words \(P,Q\) and a nonempty block \(v_1\cdots v_m\).  Leading and
trailing copies of \(a\) in the block may be moved into \(P\) or \(Q\)
without changing \(u\) and \(u'\); since \(u\ne u'\), we may therefore
assume \(v_1\ne a\) and \(v_m\ne a\).

Let \(t=|P|\).  The signature bits of \(u\) at positions \(t+1,\ldots,t+m\)
compare consecutive pairs among \(a,v_1,\ldots,v_m\), and those of \(u'\)
compare consecutive pairs among \(v_1,\ldots,v_m,a\):
\[
\begin{split}
 u:\quad&\mathbf 1[a\preceq v_1],\mathbf 1[v_1\preceq v_2],\ldots,
   \mathbf 1[v_{m-1}\preceq v_m],\\
 u':\quad&\mathbf 1[v_1\preceq v_2],\ldots,
   \mathbf 1[v_{m-1}\preceq v_m],\mathbf 1[v_m\preceq a].
\end{split}
\]
If \(\alpha(u)=\alpha(u')\), these two lists agree position by position,
which gives
\[
 \mathbf 1[a\preceq v_1]=\mathbf 1[v_1\preceq v_2]=\cdots
 =\mathbf 1[v_{m-1}\preceq v_m]=\mathbf 1[v_m\preceq a].
\]
Thus all comparisons around the cycle \(a,v_1,\ldots,v_m,a\) have the same
value.  If that value is one, then \(a\preceq v_1\preceq\cdots\preceq
v_m\preceq a\), and antisymmetry forces every symbol in the cycle to equal
\(a\), contrary to \(v_1\ne a\).  If it is zero, the cycle is strictly
decreasing, which is impossible in a total order.
\end{proof}

\begin{proposition}[One-deletion correction from the sum and Tenengolts checksum]
\label{prop:tenengolts-one-deletion}
Let \(n\ge2\).  If distinct \(x,y\in\F_q^n\) have a common one-deletion
descendant, then
\[
 (S(x),T_n(x))\ne(S(y),T_n(y)).
\]
Equivalently, every fiber of \((S,T_n)\) corrects one deletion.
\end{proposition}

The fiber form is stated because it is the usual coding statement, as for the VT codes; later
arguments use the proposition only in its first form
(Lemma~\ref{lem:last-insertion-gate} and
Sections~\ref{sec:graph}--\ref{sec:adjacent}).

\begin{proof}
Let \(z\) be a common one-deletion descendant of \(x\) and \(y\), and let
\(a,b\) be their respective deleted values.  If \(S(x)=S(y)\), then
\[
 S(z)+a=S(x)=S(y)=S(z)+b,
\]
so \(a=b\).  Lemma~\ref{lem:ascent-deletion} shows that \(\alpha(x)\) and
\(\alpha(y)\) have the common one-deletion descendant \(\alpha(z)\).  If
also \(T_n(x)=T_n(y)\), Lemma~\ref{lem:binary-vt-separation}, applied to
binary words of length \(n-1\), forces \(\alpha(x)=\alpha(y)\).  The two
words are obtained by inserting the same recovered value into the same
word \(z\); Lemma~\ref{lem:ascent-insertion-unique} then forces \(x=y\), a
contradiction.
\end{proof}

The following two lemmas are used in Section~\ref{sec:graph} to list the
neighbors of a given message in the confusion graphs.

\begin{lemma}[A prescribed insertion order is complete]
\label{lem:prescribed-insertion-order}
Let \(\zeta_1,\ldots,\zeta_m\) be the elements of a finite multiset \(U\),
listed in a fixed order.  Let \(z\) be a word, and let \(u\) be a word
that contains \(z\) as a subsequence and whose remaining symbols form the
multiset \(U\).  Then
\(u\) can be obtained from \(z\) by inserting first \(\zeta_1\), then
\(\zeta_2\), and so on up to \(\zeta_m\), each at a suitable position.
\end{lemma}

\begin{proof}
Choose an embedding of \(z\) in \(u\), and match \(\zeta_1,\ldots,\zeta_m\)
bijectively with the remaining positions of \(u\), so that each \(\zeta_i\)
is matched with a position holding the value \(\zeta_i\); equal values are
interchangeable, so such a matching exists.  Build \(u\) from \(z\) in \(m\)
steps: at step \(i\), insert \(\zeta_i\) at the place of its matched
position, relative to the symbols already present.
\end{proof}

\begin{lemma}[One last completion at a prescribed checksum]
\label{lem:last-insertion-gate}
Let \(n\ge2\).  Fix \(w\in\F_q^{n-1}\), \(b\in\F_q\), and
\(\tau\in\mathbb Z_n\).  Among the
distinct words obtained by inserting \(b\) once into \(w\), at most one has
\(T_n=\tau\).
\end{lemma}

\begin{proof}
Two distinct words obtained by inserting \(b\) into \(w\) have \(w\) as a
common one-deletion descendant and the same sum \(S(w)+b\).  Proposition~\ref{prop:tenengolts-one-deletion}
therefore forces their \(T_n\)-values to differ.
\end{proof}

In summary, given the received word: the full value of \(S\) recovers one
deleted value; the full value of \(V\) leaves at most two candidates for
it; the full pair \((S,V)\) recovers the unordered pair of two deleted
values; and once the inserted value and the received word are fixed, at
most one of the distinct words obtained by inserting that value has a
prescribed value of \(T_n\) (Lemma~\ref{lem:last-insertion-gate}).  The
remaining ambiguity is represented by the sparse confusion graphs
of Section~\ref{sec:graph}.

\section{Confusion graphs and a computable coloring}
\label{sec:graph}

The checks of Section~\ref{sec:values}, when they survive in full, recover
the deleted values, but their positions can remain ambiguous; when only
parts of them survive, the values can remain ambiguous as well.  This
section describes the remaining ambiguity by graphs and colors them.  The
construction in
Section~\ref{sec:construction} arranges the payloads so that, in each case
of the table at the end of Section~\ref{sec:boundary}, the surviving redundant symbols carry
at least the following checks of the message, and these are the ones used:
\begin{itemize}[leftmargin=*]
\item one information symbol and redundant symbol~\(1\) deleted: parts of
  \(S\) and \(V\), the checksum \(T_n\), and \(\ell_1\);
\item one information symbol and redundant symbol~\(2\) deleted: the
  full value of \(V\) and the color \(\ell_1\);
\item one information symbol and redundant symbol~\(3\) deleted: the full
  values of \(S\) and \(T_n\), which suffice by
  Proposition~\ref{prop:tenengolts-one-deletion};
\item two information symbols deleted: \(S\), \(V\), \(T_n\), and
  \(\ell_1\).
\end{itemize}
In the first, second, and fourth cases, we join two messages by an edge
when they have a common one-deletion descendant (first and second cases)
or a common two-deletion descendant (fourth case), and agree on every
listed check except \(\ell_1\).  This gives the three graphs
\(G_H\), \(G_V\), and \(J_2^T\) defined below.  (As mnemonics: \(G_H\) is used when the heads of \(S\) and \(V\), stored
in redundant symbol~\(1\), are lost; \(G_V\) is used when the full value of
\(V\) is known; and \(J_2^T\) is used when two information symbols are
deleted, the superscript recording that \(T_n\) is among the matched
checks.)  The color \(\ell_1\) is
then a proper coloring of their union, so it separates every pair of
messages that the other checks leave confusable.  The graphs ignore the labels
and markers, so they may contain more edges than there are actual
collisions between codewords; this only makes the coloring requirement
stronger.  The two facts that make the approach work are that every vertex
has polynomially many neighbors (Lemma~\ref{lem:three-residual-degrees})
and that these neighbors can be listed efficiently, which allows the color
of a single vertex to be computed in polynomial time
(Proposition~\ref{prop:deterministic-residual-color}).

Throughout this section \(n\ge2\).  Write \(q=2^s\), identify \(\F_q\) with
\(\F_2^s\) through a fixed public basis, and put
\[
  u=\lceil\log_2 n\rceil,
  \qquad N=2^u.
\]
The bit projections below are defined through this coordinate
representation.  Field arithmetic does not depend on it.  The total order
\(\preceq\) of Section~\ref{sec:values} is any public total order here; in
Section~\ref{sec:construction} it is fixed as the lexicographic order of
coordinate words.

\subsection{Tail projections and the three graphs}

The graph \(G_H\) uses only parts of \(S\) and \(V\), so we first define
these parts.  Choose
integers \(0\le h_S,h_V\le s\), and let
\[
 \pi_S:\F_2^s\longrightarrow\F_2^{s-h_S},
 \qquad
 \pi_V:\F_2^s\longrightarrow\F_2^{s-h_V}
\]
be fixed surjective coordinate projections, so that
\(\dim\ker\pi_S=h_S\) and \(\dim\ker\pi_V=h_V\).  Set
\[
 S_{\mathrm{tail}}(x)=\pi_S(S(x)),
 \qquad
 V_{\mathrm{tail}}(x)=\pi_V(V(x)),
\]
and write \(S_{\mathrm{head}}(x)\in\F_2^{h_S}\) and
\(V_{\mathrm{head}}(x)\in\F_2^{h_V}\) for the words formed by the \(h_S\)
and \(h_V\) coordinates of \(S(x)\) and \(V(x)\) that \(\pi_S\) and
\(\pi_V\) discard; thus \(S(x)\) is determined by the pair
\((S_{\mathrm{head}}(x),S_{\mathrm{tail}}(x))\), and likewise \(V(x)\).
The construction in Section~\ref{sec:construction} uses the following
values, which are derived there from the requirement that all checks fit
exactly into the three payloads:
\begin{equation}
 s=7u+25,
 \qquad h_S=u+4,
 \qquad h_V=6u+17.
 \label{eq:deterministic-tail-dimensions}
\end{equation}
In outline, every bit of a redundant symbol is used: of its \(s\) bits,
four carry the label (two bits, so that \(4>2\) labels are
available) and the two-bit position marker of Section~\ref{sec:boundary},
and the three payloads of \(s-4\) bits hold \(S\) and \(V\) (\(s\) bits
each), \(T_n\) (\(u\) bits) and the color \(\ell_1\) (\(6u+13\) bits)
exactly, which forces \(s=7u+25\).  Redundant symbol~\(1\) holds the
heads of \(S\) and \(V\), symbol~\(2\) holds \(S_{\mathrm{tail}}\) together
with \(T_n\), and symbol~\(3\) holds \(V_{\mathrm{tail}}\) together with
\(\ell_1\); since the last two payloads are full, \(h_S=u+4\) and
\(h_V=6u+17\).  Thus the two tails have respective lengths \(6u+21\) and
\(u+8\), and they are exactly the parts of \(S\) and \(V\) that remain
available when redundant symbol~\(1\) is deleted.  Here \(h_S\) controls the degree bound for \(G_H\) in
Lemma~\ref{lem:three-residual-degrees}.  That bound counts candidates using
\(S_{\mathrm{tail}}\) and \(T_n\); the \(V_{\mathrm{tail}}\) condition only
filters the list, so \(h_V\) does not enter the estimates.  The length
\(ns\) of each message's binary coordinate word enters the coloring bound
because these words serve as the vertex names.
Proposition~\ref{prop:deterministic-residual-color} substitutes the chosen
values into that bound and checks that the color fits into the
\(6u+13\) bits reserved for \(\ell_1\).

Define \(G_V\) to be the simple undirected graph with vertex set \(\F_q^n\)
in which two distinct vertices \(x,y\) are adjacent exactly when
\begin{equation}
 \mathsf D_1(x)\cap\mathsf D_1(y)\ne\varnothing
 \quad\text{and}\quad V(x)=V(y).
 \label{eq:GV-adjacency}
\end{equation}
This graph describes the ambiguity left when one information symbol is
deleted and the full check \(V\) is known.  (The head
\(S_{\mathrm{head}}\) of \(S\) also survives in that case; like the labels
and markers, it is ignored by the graph: the degree bound \(2n^2\) of
Lemma~\ref{lem:three-residual-degrees} is already far below the bound for
\(G_H\), so nothing would be gained.)

Define \(G_H\) to be the simple undirected graph with the same vertex set,
where distinct \(x,y\) are adjacent exactly when
\begin{equation}
 \begin{split}
  &\mathsf D_1(x)\cap\mathsf D_1(y)\ne\varnothing,\\
  &\bigl(S_{\mathrm{tail}}(x),V_{\mathrm{tail}}(x),T_n(x)\bigr)
   =
   \bigl(S_{\mathrm{tail}}(y),V_{\mathrm{tail}}(y),T_n(y)\bigr).
 \end{split}
 \label{eq:GH-adjacency}
\end{equation}
It describes the ambiguity left when one information symbol is deleted and
only the tails of \(S\) and \(V\), together with \(T_n\), are known.

Finally, define \(J_2^T\) to be the simple undirected graph on \(\F_q^n\)
in which distinct \(x,y\) are adjacent exactly when
\begin{equation}
 \begin{split}
  &\mathsf D_2(x)\cap\mathsf D_2(y)\ne\varnothing,\\
  &(S(x),V(x),T_n(x))=(S(y),V(y),T_n(y)).
 \end{split}
 \label{eq:J2-adjacency}
\end{equation}
It describes the ambiguity left when two information symbols are deleted
and all three payloads survive.

Let
\begin{equation}
 G_\star=(\F_q^n,E(G_V)\cup E(G_H)\cup E(J_2^T)).
 \label{eq:union-residual-graph}
\end{equation}
Repeated edges in the union are retained only once.

\subsection{Degree bounds and neighborhood listing}

The following lemma bounds the degrees and shows that neighborhoods can be
listed efficiently.  A function on the vertices of a graph is
\emph{individually computable} if an algorithm computes its value at any one
given vertex without constructing the whole graph on \(q^n\) vertices; in
particular, the neighborhoods are individually computable if such an
algorithm lists the distinct neighbors of a given vertex.  This is the only
kind of access used by the coloring in the last subsection.  In the proofs, a \emph{child}
of a word is a deletion descendant, and a \emph{parent} of a word \(w\) is a
word of which \(w\) is a deletion descendant.

\begin{lemma}[Degree bounds and neighborhood enumeration]
\label{lem:three-residual-degrees}
For every \(x\in\F_q^n\),
\begin{equation}
 \begin{aligned}
  \deg_{G_V}(x)&\le 2n^2,\\
  \deg_{G_H}(x)&\le n2^{h_S},\\
  \deg_{J_2^T}(x)&\le\binom n2(n-2).
 \end{aligned}
 \label{eq:three-residual-degrees}
\end{equation}
All three neighborhoods are individually computable in time polynomial in
\(n\), \(s\), and \(2^{h_S}\).  In particular, under
\eqref{eq:deterministic-tail-dimensions} they are individually computable in
time polynomial in \(n\) and \(s=\log_2 q\).
\end{lemma}

\begin{proof}
Fix \(x\in\F_q^n\).  For \(G_V\), choose one of the at most \(n\) distinct
words \(w\in\mathsf D_1(x)\).  If \(a\) is the value inserted into \(w\),
equation~\eqref{eq:one-deletion-V} prescribes one fiber of the
\(\F_2\)-linear map \(a\mapsto a^2+S(w)a\).  By
Lemma~\ref{lem:one-deletion-nonlinear}, that fiber has at most two elements.
Each value has \(n\) insertion slots, so there are at most \(2n^2\)
candidates.  Generating them, filtering by \eqref{eq:GV-adjacency}, removing
\(x\), and deduplicating gives the whole neighborhood.  The linearized
equation can be solved by Gaussian elimination in the public \(s\)-bit
basis, so this is polynomial in \(n\) and \(s\).

For \(G_H\), again fix a child \(w\in\mathsf D_1(x)\).  A parent obtained by
inserting \(a\) has the required \(S\)-tail precisely when
\[
 \pi_S(a)=S_{\mathrm{tail}}(x)+\pi_S(S(w)).
\]
The set of values \(a\) satisfying this equation is either empty or a coset
of \(\ker\pi_S\).  It is nonempty, because the value deleted from \(x\) to
obtain \(w\) is a solution, so it contains exactly \(2^{h_S}\) values.
For each such \(a\), among the distinct words obtained by inserting this
same \(a\) into \(w\), at most one has checksum \(T_n(x)\), by
Lemma~\ref{lem:last-insertion-gate}.  The required \(V\)-tail only filters
this list.  Thus there are at most \(2^{h_S}\) possible parents per child
and at most \(n2^{h_S}\) in total.  The following procedure therefore lists
the whole neighborhood: for each child \(w\), enumerate the \(2^{h_S}\)
values \(a\) of the coset, insert each \(a\) into \(w\) at every position,
keep the words satisfying \eqref{eq:GH-adjacency}, remove \(x\), and delete
duplicates.  With
\(h_S=u+4\),
\[
 2^{h_S}=16N<32n,
\]
where the strict inequality uses \(N<2n\).  Thus the running time is
polynomial.

For \(J_2^T\), fix \(w\in\mathsf D_2(x)\).  The target values \(S(x)\) and
\(V(x)\), together with \(w\), recover the unordered deleted multiset
\(\{a,b\}\) by Lemma~\ref{lem:two-deletion-nonlinear}; repeated values are
included.  Every \(y\) with \(w\in\mathsf D_2(y)\) and the same \((S,V)\)
arises from \(w\) by inserting the same two values, and by
Lemma~\ref{lem:prescribed-insertion-order} it suffices to insert them in
one fixed order.  There are at most \(n-1\) distinct intermediate words
after the first insertion, and for
each one Lemma~\ref{lem:last-insertion-gate} permits at most one distinct
last insertion with checksum \(T_n(x)\).  Hence at most \(n-1\) words
\(y\) satisfy \(w\in\mathsf D_2(y)\) and
\((S,V,T_n)(y)=(S,V,T_n)(x)\).  One of them is \(x\), so \(w\) contributes
at most \(n-2\) neighbors.  There are at most \(\binom n2\)
distinct choices of \(w\), proving the third bound.

To list the neighbors, enumerate all two-position deletions of \(x\).  For
each distinct \(w\in\mathsf D_2(x)\), the multiset \(\{a,b\}\) deleted from
\(x\) is determined by \(w\), \(S(x)\), and \(V(x)\)
(Lemma~\ref{lem:two-deletion-nonlinear}).  It is therefore also the multiset
deleted from every neighbor \(y\) with \(w\in\mathsf D_2(y)\), because such
a \(y\) has the same \(S\) and \(V\).  By
Lemma~\ref{lem:prescribed-insertion-order}, every such \(y\) arises by
inserting \(a\) first and \(b\) second.  So try every position for \(a\)
and then every position for \(b\), keep the words with the required
\((S,V,T_n)\), remove \(x\), and deduplicate.  This takes time polynomial in
\(n\) and \(s\); in particular, it needs the deleted values only as they are
read off from \(x\), and no quadratic equation has to be solved.  When
\(n=2\), the bound \(n-2\) is zero, so \(J_2^T\) has no edges.
\end{proof}

Consequently
\begin{equation}
 \Delta(G_\star)\le
 D_n:=2n^2+n2^{h_S}+\binom n2(n-2).
 \label{eq:union-degree-bound}
\end{equation}
For \(h_S=u+4\) as in \eqref{eq:deterministic-tail-dimensions}, if
\(u\ge2\), then
\begin{equation}
 D_n<2N^2+16N^2+\tfrac12N^3\le5N^3.
 \label{eq:stable-union-degree}
\end{equation}
At the endpoint \(n=2\), the three summands in
\eqref{eq:union-degree-bound} are \(8,64,0\), respectively, and hence
\begin{equation}
 D_2=72.
 \label{eq:endpoint-union-degree}
\end{equation}

\begin{proposition}[Individual access to the union graph]
\label{prop:gstar-indexed-oracle}
The neighborhoods of the simple union graph \(G_\star\) are individually
computable in time polynomial in \(n\), \(s\), and \(2^{h_S}\).  Moreover,
naming every vertex by its public binary coordinate word in
\(\{0,1\}^{sn}\), a deterministic oracle answers the degree of a vertex and
its \(i\)-th neighbor with the same running time per query.
\end{proposition}

\begin{proof}
Run the three complete enumerators from
Lemma~\ref{lem:three-residual-degrees}, concatenate their outputs, remove the
queried vertex, sort by the public name, and delete repetitions.  The result
is exactly the neighborhood in the simple union
\eqref{eq:union-residual-graph}, which proves the first claim; its length
and its \(i\)-th entry answer the two oracle queries.  Sorting at most
\(D_n\) names of length \(sn\) adds only a polynomial factor.  Notice that
the number of vertices is \(q^n=2^{sn}\), so the names have exactly
\(sn\) bits.
\end{proof}

The coloring below applies Lemma~\ref{lem:two-round-local-coloring} to
\(G_\star\); its hypothesis that neighborhoods are individually computable
is the first claim of Proposition~\ref{prop:gstar-indexed-oracle}, and the
encoder of Section~\ref{sec:construction} uses the same enumeration.  The
indexed form of access (degree and \(i\)-th neighbor) is the usual
interface in local computation and is recorded for comparison; it is not
used elsewhere in the paper.

\subsection{A coloring computable at one vertex}

We now give \(G_\star\) an individually computable coloring, in which the
color of one vertex is computed from its radius-two neighborhood alone.  The idea is
classical color reduction.  In the first round,
the name of each vertex, written in base \(q_1\), is read as the coefficient
vector of a polynomial of degree at most
\(\ell_{\mathrm{name}}\) (the length of the names; a generous bound, since the
base-\(q_1\) digits are fewer) over a field \(\F_{q_1}\),
where \(q_1\) is the least power of two larger than \(D\ell_{\mathrm{name}}\), for a bound \(D\) on the maximum degree
(Appendix~\ref{app:two-round-coloring}).  A vertex is colored by the first point
\(t\in\F_{q_1}\) at which its polynomial differs from the polynomials of all
its neighbors, together with the value there.  In the second round, the
same rule is applied with the first-round colors in place of the names, now
as polynomials of degree at most three over a smaller field.

This lemma is a quantitative specialization of the color-reduction method
of Erd\H{o}s, Frankl, and F\"uredi~\cite{ErdosFranklFuredi1985} and
Linial~\cite{Linial1992}, as formulated for systematic codes with
individually computable colors by Li and Farnoud~\cite{LiFarnoud2023} and by
Li, Gabrys, and Farnoud; compare especially
\cite[Lemma~13, Corollary~14, Lemmas~15--16, and
Theorems~2--3]{LiGabrysFarnoud2025}.
We need explicit constants because the number of colors must fit into the
\(6u+13\) bits reserved for \(\ell_1\) in Section~\ref{sec:construction};
two rounds suffice for this (Proposition~\ref{prop:deterministic-residual-color}).

\begin{lemma}[Two-round color reduction]
\label{lem:two-round-local-coloring}
Let \(G\) be a simple graph whose vertices have distinct
\(\ell_{\mathrm{name}}\)-bit public names, whose maximum degree is at most
\(D\), and whose neighborhoods are individually computable.  If \(D=0\),
one color suffices.  If \(D\ge1\), then \(G\) has an individually computable
proper coloring with at most
\begin{equation}
 64D^2+16D\ell_{\mathrm{name}}
 \label{eq:two-round-palette}
\end{equation}
colors.  Evaluating one color inspects only the radius-two neighborhood of
the requested vertex.
\end{lemma}

The proof, with the choice of the two auxiliary fields, the constant
calculation, and the cost of evaluating one color (at most \(D+1\) neighborhood computations and further work polynomial in \(D\) and \(\ell_{\mathrm{name}}\)), is in
Appendix~\ref{app:two-round-coloring}.  In that proof
each final color is a pair of elements of a field \(\F_{q_2}\), where \(q_2\)
is a power of two and \(q_2^2\) is at most \eqref{eq:two-round-palette}.
The two terms of \eqref{eq:two-round-palette} reflect the two lower limits
on the size of the second-round field, and their sum covers whichever limit
is larger.  This field must exceed \(3D\), because each neighbor excludes at most three points, and it must be at least the square
root of the first-round field, which has more than \(D\ell_{\mathrm{name}}\)
elements, so that first-round colors fit into four of its digits.
In Appendix~\ref{app:two-round-coloring}, \(q_2\) is the least power of two meeting both limits; hence \(q_2<2\max\{3D+1,\lceil\sqrt{q_1}\rceil\}\).  Therefore
\(q_2^2<4(3D+1)^2+4\lceil\sqrt{q_1}\rceil^2\le64D^2+8q_1\), and
\(q_1\le2D\ell_{\mathrm{name}}\) bounds the second term by
\(16D\ell_{\mathrm{name}}\).

\begin{proposition}[An individually computable proper color]
\label{prop:deterministic-residual-color}
Under the dimensions
\eqref{eq:deterministic-tail-dimensions}, the graph \(G_\star\) has a
public proper coloring
\[
 \ell_1:\F_q^n\longrightarrow\{0,1\}^{6u+13}
\]
whose value at one requested message is computable in time polynomial in
\(n\) and \(s=\log_2 q\).
\end{proposition}

\begin{proof}
Apply Lemma~\ref{lem:two-round-local-coloring} once to the union graph
\(G_\star\), using the public coordinate word as a distinct name of length
\begin{equation}
 \ell_{\mathrm{name}}=ns.
 \label{eq:message-name-length}
\end{equation}
The hypotheses of the lemma hold with \(D=D_n\): the maximum degree of \(G_\star\) is at
most \(D_n\) by \eqref{eq:union-degree-bound}, and its neighborhoods are
individually computable by Proposition~\ref{prop:gstar-indexed-oracle}.
Write \(K_{\mathrm{col}}\) for the resulting number of colors.  We check
that, for \(s=7u+25\), the color fits into the \(6u+13\) bits reserved for
it in Section~\ref{sec:construction}.

Indeed, suppose first that \(u\ge2\).  Equations
\eqref{eq:stable-union-degree} and \eqref{eq:message-name-length} give
\begin{equation}
 \begin{split}
 K_{\mathrm{col}}
 &<1600N^6+80N^4(7u+25)\\
 &\le1795N^6
 <2^{13}N^6=2^{6u+13}.
 \end{split}
 \label{eq:stable-color-budget}
\end{equation}
Here \(80(7u+25)\le195N^2\) for every \(u\ge2\): equality holds at
\(u=2\), after which the right side quadruples at each step while the left
side increases by only \(560\).  If \(u=1\), then \(n=2\), \(s=32\),
\(\ell_{\mathrm{name}}=64\), and the direct endpoint calculation is
\begin{equation}
 K_{\mathrm{col}}
 \le64\cdot72^2+16\cdot72\cdot64
 =405504<2^{19}=2^{6u+13}.
 \label{eq:endpoint-color-budget}
\end{equation}

The proof of Lemma~\ref{lem:two-round-local-coloring} realizes its final
color as a pair of elements of a canonical auxiliary field.  Concatenate
the two canonical binary coordinate words and pad on the left with zeros to
length \(6u+13\).  The concatenation has \(\log_2(q_2^2)\) bits, and \(q_2^2\) is at most
\eqref{eq:two-round-palette}, which is less than \(2^{6u+13}\) by
\eqref{eq:stable-color-budget} and \eqref{eq:endpoint-color-budget}.  So
the padding is possible, and the resulting map from colors to
\(\{0,1\}^{6u+13}\) is injective and computable from \(n\).
Denote the resulting color by
\begin{equation}
 \ell_1:\F_q^n\longrightarrow\{0,1\}^{6u+13}.
 \label{eq:residual-color}
\end{equation}
It is individually computable in time polynomial in \(n\) and \(s\), and it
is proper simultaneously on \(G_V\), \(G_H\), and \(J_2^T\).
\end{proof}

Therefore two messages that are joined in one of the three graphs always
have different colors \(\ell_1\).  This color is the only output of the
graph construction stored in the redundant symbols.  The decoder of
Section~\ref{subsec:algorithms} does not search these graphs for the received
word: it re-encodes each candidate message, computing that candidate's color
exactly as the encoder does, and compares the result with the received
word.

\section{The construction with three redundant symbols}
\label{sec:construction}

We now place the checks of Sections~\ref{sec:boundary}--\ref{sec:graph}
into three redundant symbols and prove the following result.

\begin{theorem}[Efficient construction with three redundant symbols]
\label{thm:main}
For every \(n\ge2\), there is a power of two \(q<2^{32}n^7\) and a
systematic encoder
\[
  E_n:\F_q^n\longrightarrow\F_q^{n+3}
\]
whose image corrects up to two arbitrary deletions.  Encoding and
decoding take time polynomial in \(n\)
for this family, under the computational convention of Section~\ref{sec:model}.
\end{theorem}

The
requirement is the one listed at the beginning of Section~\ref{sec:graph}:
whichever redundant symbol is deleted, the two surviving payloads must
contain the checks needed in that case, and all three payloads together must
contain the checks needed when two information symbols are deleted.  We
first choose the parameters and verify that every component fits; then we
define the public representation of the alphabet and the encoder.

\subsection{Parameters and the public alphabet}
\label{subsec:parameters}

Fix \(n\ge2\), and retain the notation
\[
 u=\lceil\log_2 n\rceil,
 \qquad N=2^u
\]
from Section~\ref{sec:graph}.  Set
\begin{equation}
 \begin{aligned}
 d_{\mathrm{tag}}&=2,
 &\delta&=d_{\mathrm{tag}}+2=4,\\
 b_{\mathrm{col}}&=6u+13,
 &s&=u+b_{\mathrm{col}}+3\delta=7u+25,\\
 \rho&=s-\delta=7u+21,
 &q&=2^s.
 \end{aligned}
 \label{eq:construction-parameters}
\end{equation}
Here \(d_{\mathrm{tag}}\) is the number of label bits, the extra two bits in
\(\delta\) hold the position marker, \(b_{\mathrm{col}}\) is the number of
bits reserved for the color \(\ell_1\), and \(\rho\) is the number of
payload bits in each redundant symbol.

\emph{How the parameters are determined.}  Every bit of a symbol is used.
Once \(d_{\mathrm{tag}}=2\), the color budget \(b_{\mathrm{col}}\), and the
payload layout \eqref{eq:three-payloads} of Section~\ref{subsec:payloads} are fixed, the
remaining parameters are determined by the following requirements.
\begin{enumerate}[leftmargin=*]
\item The choice of label uses \(|C_{\mathrm{tag}}|=2^{d_{\mathrm{tag}}}>2\), so that
  a label absent from the last two information symbols exists;
  \(d_{\mathrm{tag}}=2\) suffices.  With the two
  marker bits, each symbol has \(\delta=4\) overhead bits and
  \(\rho=s-\delta\) payload bits.
\item The three payloads together must hold \(S\) and \(V\) (\(s\) bits
  each), \(T_n\) (\(u\) bits), and \(\ell_1\) (\(b_{\mathrm{col}}\) bits).
  Packing them exactly gives \(3(s-\delta)=2s+u+b_{\mathrm{col}}\), that is,
  \(s=3\delta+u+b_{\mathrm{col}}\).
\item The second payload \(w_2\), which holds \(S_{\mathrm{tail}}\) and the
  \(u\)-bit encoding \(\ell_0\) of \(T_n\) (both defined in Section~\ref{subsec:payloads}),
  must fill one payload, so
  \((s-h_S)+u=s-\delta\), that is, \(h_S=u+\delta=u+4\).  This makes
  \(2^{h_S}=16N\), so the degree of \(G_H\) is polynomial
  (Lemma~\ref{lem:three-residual-degrees}).
\item The first payload \(w_1\), which holds \(S_{\mathrm{head}}\) and
  \(V_{\mathrm{head}}\), must fill one payload, so \(h_V=\rho-h_S=6u+17\); the third payload
  \(w_3=(V_{\mathrm{tail}},\ell_1)\) then has exactly
  \(b_{\mathrm{col}}\) bits left for the color.
\item The number of colors depends on the name length
  \(\ell_{\mathrm{name}}=ns\), and hence on \(s\), which in turn depends on
  \(b_{\mathrm{col}}\).  The value \(b_{\mathrm{col}}=6u+13\) is a value for
  which the bound on the number of colors \(K_{\mathrm{col}}\) proved in
  Section~\ref{sec:graph} is less than \(2^{b_{\mathrm{col}}}\); we do not
  claim that it is the least possible.
\end{enumerate}
Requirement~3 and the two clauses of requirement~4 are the three payload
identities assumed in Proposition~\ref{prop:proper-color-interface} below;
requirement~2 is their sum.
Requirement~5 is verified in the proof of
Proposition~\ref{prop:deterministic-residual-color}, by
\eqref{eq:stable-color-budget} for \(u\ge2\) and by
\eqref{eq:endpoint-color-budget} for \(n=2\).  Thus that proposition
supplies a proper color \(\ell_1(x)\) in exactly \(b_{\mathrm{col}}\) bits
for every \(n\ge2\).  Moreover,
\begin{equation}
 q=2^{7u+25}=2^{25}N^7<2^{32}n^7,
 \label{eq:uniform-alphabet-closure}
\end{equation}
because \(N<2n\).  This is the alphabet bound of Theorem~\ref{thm:main}.

We fix the field and all binary coordinates canonically, as in
Section~\ref{sec:model}: \(\F_q=\F_2[X]/(f_s(X))\), with the polynomial
basis \(1,\theta,\ldots,\theta^{s-1}\) and the public bijection
\(\operatorname{coord}:\F_q\to\{0,1\}^s\) onto coordinate words.  The
lexicographic order of coordinate words is the public total order
\(\preceq\) used in Section~\ref{sec:values}.  (Any fixed public total
order would serve; the proofs use only that it is a total order.)  Since
\(2^s=q=O(n^7)\), the exhaustive search for \(f_s\) is polynomial in the
family parameter \(n\), so this field choice meets the computational
convention.

Let
\[
 C_{\mathrm{tag}}=\{0,1\}^{d_{\mathrm{tag}}},
\]
enumerated as binary numbers with the first bit most significant; this is
the public enumeration that the least-absent rule
\eqref{eq:least-absent-tag} uses.  Splitting the public coordinate word into
blocks of lengths \(d_{\mathrm{tag}},2,\rho\) defines a set bijection
\begin{equation}
 \beta:\F_q\longrightarrow
 C_{\mathrm{tag}}\times\{0,1,2,3\}\times\{0,1\}^{\rho}.
 \label{eq:public-symbol-bijection}
\end{equation}
The middle two bits are read as the binary representation of an integer in
\(\{0,1,2,3\}\), the first of the two bits being the more significant.  Write
\(\pi_{\mathrm{tag}}\) and \(\pi_{\mathrm{mark}}\) for the first and second
coordinate projections of \(\beta\).  This is only a public set
representation of field symbols; it does not alter the field operations.
Information symbols are arbitrary field elements, so their label and marker
fields are unconstrained.
Since
\[
 |C_{\mathrm{tag}}|=2^{d_{\mathrm{tag}}}=4>2,
\]
the least-absent choice in Equation~\eqref{eq:least-absent-tag} gives a
well-defined label
\begin{equation}
 \Gamma(x)=\min\bigl(C_{\mathrm{tag}}\setminus
 \{\pi_{\mathrm{tag}}(x_{n-1}),\pi_{\mathrm{tag}}(x_n)\}\bigr)
 \label{eq:construction-absent-tag}
\end{equation}
for every message \(x\in\F_q^n\).

\subsection{The three payloads}
\label{subsec:payloads}

From now on, \(\pi_S\) and \(\pi_V\) keep the last \(s-h_S\) and
\(s-h_V\) coordinates of the coordinate word, respectively, so that the
heads \(S_{\mathrm{head}}\) and \(V_{\mathrm{head}}\) of
Section~\ref{sec:graph} are the first \(h_S\) and \(h_V\) coordinates.
Section~\ref{sec:graph} used only that \(\pi_S\) and \(\pi_V\) are
coordinate projections with kernel dimensions \(h_S\) and \(h_V\), so this
choice is consistent with everything proved there.
With the dimensions in Equation~\eqref{eq:deterministic-tail-dimensions},
write
\begin{equation}
 \begin{aligned}
 S(x)&=(S_{\mathrm{head}}(x),S_{\mathrm{tail}}(x)),
 &|S_{\mathrm{head}}|&=u+4,
 &|S_{\mathrm{tail}}|&=6u+21,\\
 V(x)&=(V_{\mathrm{head}}(x),V_{\mathrm{tail}}(x)),
 &|V_{\mathrm{head}}|&=6u+17,
 &|V_{\mathrm{tail}}|&=u+8.
 \end{aligned}
 \label{eq:head-tail-slices}
\end{equation}
Let \(\ell_0(x)\in\{0,1\}^u\) be the \(u\)-bit binary representation
(first bit most significant) of
the least nonnegative representative of \(T_n(x)\in\mathbb Z_n\).  This is
injective because \(n\le2^u\).  The color
\(\ell_1(x)\in\{0,1\}^{b_{\mathrm{col}}}\) is the one fixed in
Equation~\eqref{eq:residual-color}.

Define
\begin{equation}
 \begin{aligned}
  w_1(x)&=(S_{\mathrm{head}}(x),V_{\mathrm{head}}(x)),\\
  w_2(x)&=(S_{\mathrm{tail}}(x),\ell_0(x)),\\
  w_3(x)&=(V_{\mathrm{tail}}(x),\ell_1(x)).
 \end{aligned}
 \label{eq:three-payloads}
\end{equation}
The complete payload allocation is as follows.  The last column verifies
that every row fills one \(\rho\)-bit payload exactly.

\begin{center}
\begin{tabularx}{\textwidth}{@{}c >{\raggedright\arraybackslash}X
                              >{\raggedright\arraybackslash}X@{}}
\toprule
redundant symbol & payload & bit count \\
\midrule
1 & \(w_1=(S_{\mathrm{head}},V_{\mathrm{head}})\)
  & \((u+4)+(6u+17)=7u+21=\rho\) \\
2 & \(w_2=(S_{\mathrm{tail}},\ell_0)\)
  & \((6u+21)+u=7u+21=\rho\) \\
3 & \(w_3=(V_{\mathrm{tail}},\ell_1)\)
  & \((u+8)+(6u+13)=7u+21=\rho\) \\
\bottomrule
\end{tabularx}
\end{center}

The three redundant symbols are
\begin{equation}
 p_j(x)=\beta^{-1}\bigl(\Gamma(x),j,w_j(x)\bigr),
 \qquad j=1,2,3,
 \label{eq:three-trailer-symbols}
\end{equation}
and the systematic encoder is
\begin{equation}
 E_n(x)=x\concat p_1(x)p_2(x)p_3(x).
 \label{eq:uniform-three-trailer-encoder}
\end{equation}
Every redundant symbol therefore has the label \(\Gamma(x)\), which neither
\(x_{n-1}\) nor \(x_n\) carries, and its marker records its position.
Table~\ref{tab:split-verification} lists, for each way of splitting exactly
two deletions between the prefix and the redundancy block, the checks that
the surviving payloads provide and the result that separates the remaining
candidates.  Correction of fewer deletions follows from
Lemma~\ref{lem:exact-monotone}.

\begin{table}[htb]
\caption{The checks recovered in each split of two deletions and the result
that separates the remaining candidates.}
\label{tab:split-verification}
\begin{tabularx}{\textwidth}{@{}>{\raggedright\arraybackslash}p{0.17\textwidth}
                              >{\raggedright\arraybackslash}p{0.2\textwidth}
                              >{\raggedright\arraybackslash}X
                              >{\raggedright\arraybackslash}X@{}}
\toprule
deleted information symbols & deleted redundant symbols & checks recovered
  & separating result \\
\midrule
0 & two & intact \(x\) & systematic prefix \\
1 & symbol \(1\) & \((S_{\mathrm{tail}},V_{\mathrm{tail}},T_n,\ell_1)\)
  & proper color on \(G_H\) \\
1 & symbol \(2\) & \((V,\ell_1)\) & proper color on \(G_V\) \\
1 & symbol \(3\) & \((S,T_n)\) & Proposition~\ref{prop:tenengolts-one-deletion} \\
2 & none & \((S,V,T_n,\ell_1)\) & proper color on \(J_2^T\) \\
\bottomrule
\end{tabularx}
\end{table}

The table lists only the information needed by each separating result; a
surviving pair can contain additional head bits.  The marker must be stored
outside the payload because the proof of
Proposition~\ref{prop:proper-color-interface} needs the index of each
surviving redundant symbol before it can use that symbol's payload, and it
reads the index from the marker; the decoder of
Section~\ref{subsec:algorithms} reads no markers.

\subsection{Correction of every deletion split}

The correction argument uses the layout of the payloads and the properness
of the color, but no numerical value of the bit counts and no property of
how the color is computed (the numerical values enter only the degree and
color-budget estimates of Section~\ref{sec:graph}, the alphabet bound
\eqref{eq:uniform-alphabet-closure}, and the running-time bounds of
Section~\ref{subsec:algorithms}).  We therefore state it for every choice of bit
counts that satisfies the three payload identities of
Section~\ref{subsec:parameters}; Theorem~\ref{thm:main} and the global
greedy benchmark of Proposition~\ref{prop:global-greedy} are two instances,
both with \(d_{\mathrm{tag}}=2\) and with different color budgets
\(b_{\mathrm{col}}\).

\begin{proposition}[Proper-color interface]
\label{prop:proper-color-interface}
Let \(n\ge2\) and \(u=\lceil\log_2 n\rceil\).  Let \(d_{\mathrm{tag}}\ge2\),
put \(\delta=d_{\mathrm{tag}}+2\), and let \(s\), \(h_S\), \(h_V\), and
\(b_{\mathrm{col}}\) be nonnegative integers with \(h_S,h_V\le s\) such that
\(\rho=s-\delta\) satisfies the three payload identities
\[
 \rho=h_S+h_V=(s-h_S)+u=(s-h_V)+b_{\mathrm{col}}.
\]
Let \(q=2^s\), with \(\F_q\) identified with \(\F_2^s\) through a public
basis, and form the bijection \(\beta\) of
\eqref{eq:public-symbol-bijection} with block lengths
\((d_{\mathrm{tag}},2,\rho)\), the label \(\Gamma(x)\) of
\eqref{eq:construction-absent-tag} in
\(C_{\mathrm{tag}}=\{0,1\}^{d_{\mathrm{tag}}}\), the slices
\(S_{\mathrm{head}},S_{\mathrm{tail}}\) of lengths \(h_S,s-h_S\) and
\(V_{\mathrm{head}},V_{\mathrm{tail}}\) of lengths \(h_V,s-h_V\), and the
\(u\)-bit encoding \(\ell_0\) of \(T_n\), all by the formulas of
Sections~\ref{subsec:parameters} and~\ref{subsec:payloads} with these
values in place of \eqref{eq:construction-parameters}.  Let
\(\ell:\F_q^n\to\{0,1\}^{b_{\mathrm{col}}}\) be any function that is
proper on the simple union graph \(G_\star=G_V\cup G_H\cup J_2^T\) of
Section~\ref{sec:graph}, formed with the projections \(\pi_S,\pi_V\) of
Section~\ref{subsec:payloads}, which keep the last \(s-h_S\) and
\(s-h_V\) coordinates.  Then the
payloads
\[
 w_1=(S_{\mathrm{head}},V_{\mathrm{head}}),\qquad
 w_2=(S_{\mathrm{tail}},\ell_0),\qquad
 w_3=(V_{\mathrm{tail}},\ell)
\]
and the redundant symbols \(p_j(x)=\beta^{-1}(\Gamma(x),j,w_j(x))\),
\(j=1,2,3\), define a systematic code
\(\{x\concat p_1(x)p_2(x)p_3(x):x\in\F_q^n\}\) that corrects up to two
deletions.
\end{proposition}

\begin{proof}
Every redundant symbol \(p_j(x)\) has the label \(\Gamma(x)\), which
exists because \(|C_{\mathrm{tag}}|=2^{d_{\mathrm{tag}}}\ge4>2\) and
which by \eqref{eq:construction-absent-tag} is the label of neither
\(x_{n-1}\) nor \(x_n\), so \eqref{eq:tag-absent} holds; and \(p_j(x)\) carries the marker
\(j\) of \eqref{eq:marker-convention}.  By
Theorem~\ref{thm:equal-split-reduction}, the code corrects up to two
deletions if and only if the equal-split condition
\eqref{eq:equal-split-condition} holds.  (Informally: the common label rules out
every collision in which some received symbols are redundant symbols on one
side and information symbols on the other, and
\eqref{eq:equal-split-condition} excludes the remaining collisions, in which
the two deletion patterns remove the same number \(a\) of information
symbols.)  For \(a=0\) the condition is vacuous, since
\(\mathsf D_0(x)\cap\mathsf D_0(y)\ne\varnothing\) means \(x=y\).  Let
\(x\ne y\).

Suppose \(a=1\): the messages \(x\) and \(y\) have a common one-deletion
descendant, and we must show that the redundancy blocks \(p(x)\) and
\(p(y)\) have no common one-deletion descendant.  Suppose they had one,
\(w\).  Its two symbols are redundant symbols on both sides, and projecting
them through \(\pi_{\mathrm{mark}}\) gives one of the ordered pairs
\[
 (2,3),\qquad(1,3),\qquad(1,2),
\]
according as redundant symbol \(1,2,3\), respectively, is missing.  These
pairs are distinct, so the same redundant symbol is missing on both sides,
and the two surviving payloads of \(x\) equal those of \(y\).  Rows~2--4 of
Table~\ref{tab:split-verification} list the checks that the two surviving
payloads determine through \eqref{eq:three-payloads}, with \(\ell\) in
place of \(\ell_1\).
\begin{itemize}[leftmargin=*]
\item \emph{Redundant symbol \(1\) is missing.}  The surviving payloads
  determine \((S_{\mathrm{tail}},V_{\mathrm{tail}},T_n,\ell)\), so \(x\)
  and \(y\) satisfy \eqref{eq:GH-adjacency}: they form an edge of \(G_H\)
  with equal colors, contradicting properness of \(\ell\).
\item \emph{Redundant symbol \(2\) is missing.}  The surviving payloads
  determine the full \(V\) and \(\ell\); the common one-deletion descendant
  and \(V(x)=V(y)\) make \(\{x,y\}\) an edge of \(G_V\) by
  \eqref{eq:GV-adjacency}, again contradicting properness of \(\ell\).
\item \emph{Redundant symbol \(3\) is missing.}  The surviving payloads
  determine the full sum \(S\) and \(\ell_0\), hence \(T_n\), and
  Proposition~\ref{prop:tenengolts-one-deletion} gives \(x=y\), a
  contradiction.
\end{itemize}

Suppose \(a=2\): the messages have a common two-deletion descendant, and we
must show \(p(x)\ne p(y)\).  If \(p(x)=p(y)\), all three payloads agree and
determine \((S,V,T_n,\ell)\); the common two-deletion descendant and
equality of \((S,V,T_n)\) make \(\{x,y\}\) an edge of \(J_2^T\) by
\eqref{eq:J2-adjacency}, whereas equality of \(\ell\) contradicts its
properness on that graph.

Thus \eqref{eq:equal-split-condition} holds, and
Theorem~\ref{thm:equal-split-reduction} gives correction up to two
deletions.
\end{proof}

Proposition~\ref{prop:proper-color-interface} separates the correctness
argument from the way the color is computed and from the particular bit
counts: any proper coloring that fits into \(b_{\mathrm{col}}\) bits can be
used, with any parameters satisfying the three identities.

\begin{proof}[Proof of Theorem~\ref{thm:main}]
The parameters \eqref{eq:construction-parameters} have
\(d_{\mathrm{tag}}=2\) and satisfy the three payload identities, as the
bit-count table of Section~\ref{subsec:payloads} verifies.
Proposition~\ref{prop:deterministic-residual-color} supplies a color
\(\ell_1\) that is proper on \(G_\star\) formed with the dimensions
\eqref{eq:deterministic-tail-dimensions}, that is, with the values
\(h_S=u+4\) and \(h_V=6u+17\) of \eqref{eq:deterministic-tail-dimensions}, and
with the projections \(\pi_S,\pi_V\) and the order \(\preceq\) fixed in
Sections~\ref{subsec:parameters}--\ref{subsec:payloads}, both allowed by
that proposition; the color fits into \(b_{\mathrm{col}}=6u+13\) bits.
Proposition~\ref{prop:proper-color-interface} with \(\ell=\ell_1\) shows
that the image of the encoder \(E_n\) in
\eqref{eq:uniform-three-trailer-encoder} corrects up to two deletions.
Equation~\eqref{eq:uniform-alphabet-closure} gives the alphabet bound
\(q<2^{32}n^7\), and Section~\ref{subsec:algorithms} shows that the
encoder and the decoder meet the computational convention and run in time
polynomial in \(n\).
\end{proof}

Combining this construction with Theorem~\ref{thm:lower-bound} shows that
three full redundant symbols are necessary and sufficient in this model.

\subsection{Encoding and decoding algorithms}
\label{subsec:algorithms}

For each \(n\), the encoder first constructs the canonical field described
above.  Given \(x\in\F_q^n\), it computes \(S(x),V(x),T_n(x)\), evaluates
the individually computable color \(\ell_1(x)\) using the neighborhood
enumeration of Proposition~\ref{prop:gstar-indexed-oracle}, reads the labels of the last two
information symbols to find
\(\Gamma(x)\), and forms Equation~\eqref{eq:three-trailer-symbols}.  The
field construction, field arithmetic, graph-neighborhood enumeration, and
two-round color evaluation all take time polynomial in \(n\) and
\(s=\log_2 q\) under the parameters in
Equation~\eqref{eq:construction-parameters}; we do not estimate the degrees of
these time bounds.

The decoder is zero-error.  Polynomial running time is all we require of
it, so it simply enumerates candidates and re-encodes them; we do not pursue
a faster decoder along the lines of the correction proof.  Given an arbitrary
received word \(z\), return failure
unless \(|z|\in\{n+1,n+2,n+3\}\).  Otherwise put
\(d_{\mathrm{ch}}=n+3-|z|\in\{0,1,2\}\).  For every
\(a\in\{0,\ldots,d_{\mathrm{ch}}\}\),
interpret the first \(n-a\) received symbols as the surviving part of the
prefix and the remaining \(3-d_{\mathrm{ch}}+a\) symbols as the surviving
part of the redundancy block.
Enumerate
all length-\(n\) supersequences of that prefix obtained by inserting \(a\)
field symbols.  There are at most
\begin{equation}
 \binom na q^a
 \label{eq:decoder-candidate-count}
\end{equation}
such candidates.  Re-encode each candidate \(x\), and retain it exactly
when the proposed surviving redundant part lies in
\(\mathsf D_{d_{\mathrm{ch}}-a}(p_1(x)p_2(x)p_3(x))\).  A candidate \(x\) is
thus retained exactly when \(z\) is a \(d_{\mathrm{ch}}\)-deletion descendant
of \(E_n(x)\) via a split with \(a\) deleted information symbols.  If \(z\)
is a channel output, the true message is retained for its actual split.  By
the correction proof, at most one message is retained in total, and by
Theorem~\ref{thm:absent-label-boundary} with \(x=y\) it is retained for only one value of \(a\), so \(z\) also determines how many information symbols were deleted.  The decoder returns the
retained message if there is one, and failure otherwise.

For the present family, \(a\le2\) and \(q<2^{32}n^7\), so the total
candidate count is polynomial in \(n\) (in particular,
\(q^a=O(n^{14})\), so the branch \(a=2\) alone may re-encode \(O(n^{16})\) candidates, computing a color for each).  The decoder is therefore polynomial in \(n\) along this
family; it would not be polynomial in \(\log q\) for an alphabet that is
not polynomially bounded in \(n\).

\section{Adjacent deletions: a construction from parities and checksums}
\label{sec:adjacent}

This section gives a systematic construction for one adjacent burst.

\begin{theorem}[Parity and checksum construction for one adjacent burst]
\label{thm:adjacent-closed-form}
For every \(n\ge1\), there is a power of two
\(q_{\mathrm{adj}}<(n+2)^4\) and a
systematic map
\[
  E_n^{\mathrm{adj}}:\F_{q_{\mathrm{adj}}}^n
  \longrightarrow\F_{q_{\mathrm{adj}}}^{n+3}
\]
that corrects one burst of zero, one, or two adjacent deletions.  Each of
the three redundant symbols passes through the channel and may be deleted,
and encoding and decoding take polynomial time.
\end{theorem}

When the two
deleted symbols are adjacent, neither the confusion graphs nor the quadratic
check \(V\) of the previous sections is needed.  The redundant symbols are
two XOR parities and a final symbol, defined in
Section~\ref{subsec:final-symbol}.  No graph is
enumerated and no color reduction is needed, in contrast with
Section~\ref{sec:construction}.  The construction rests on the
split of positions into odd and even ones.  Deleting two adjacent symbols
removes exactly one odd and one even position, and every later symbol moves
two places to the left, so it keeps the parity of its position.  Hence one
parity check for the odd positions and one for the even positions recover
the two deleted values, and once the start of the burst is known they also
recover the message (Lemma~\ref{lem:adjacent-parity-inversion}).  To distinguish the candidate
messages, we use three Tenengolts checksums, on the word formed by the message
and the two parities and on the odd and even subsequences of that word
(Lemma~\ref{lem:adjacent-locator-separation}).  These
checksums are combined into one final symbol \(H(x)\).  The final symbol
may itself be deleted; Lemma~\ref{lem:adjacent-collision-characterization}
records exactly which collisions any choice of final symbol must avoid.
Those in which the final symbol survives on only one side are avoided by
choosing \(H(x)\) outside an explicit small set \(U_x\)
(Lemma~\ref{lem:adjacent-final-label-avoidance}); those in which it
survives on both sides are excluded because \(H(x)\) encodes the three
checksums injectively, and those in which it is deleted on both sides
cannot occur, because each received word then begins with its whole message.

Correction of either one binary deletion or two
adjacent binary deletions goes back to Levenshtein
\cite{Levenshtein1967Adjacent}; a direct quaternary construction for a burst
of at most two insertions or deletions was given by Khuat and Kim
\cite{KhuatKim2021}.  Recent nonbinary constructions address other
redundancy and channel regimes
\cite{WangTangSimaGabrysFarnoud2024,SongCai2023,SongCaiQuek2025,
SunLuZhangGe2025,YeSunGe2026}; over a fixed alphabet, Ye and Ge attain
\(\log_2 n+O_{q,t}(1)\) redundant bits for one burst of at most \(t\)
deletions~\cite{YeGe2026SubstringEdit}.  Recent work also gives general
finite upper bounds for multiple burst deletions
\cite{WangKongYaakobiDuman2026}.  Full/odd/even differential-VT checks appear in
Nguyen--Cai--Siegel \cite[Construction~7]{NguyenCaiSiegel2024}, and their
encoder \cite[Construction~8]{NguyenCaiSiegel2024} places the message
unchanged before a message-dependent least-allowed marker and the check
symbols, over a fixed alphabet and for a burst whose location is known to
within a window of positions.  Here the three checksums are stored in one
final symbol that may itself be deleted, no marker or location window is
used, and three full redundant symbols suffice because the alphabet grows
with \(n\).

\subsection{The burst channel and the parity symbols}
\label{subsec:burst-parities}

For a word \(z=z_1\cdots z_L\), a length-\(r\) adjacent deletion beginning
at \(d\) is
\begin{equation}
 \mathsf B_{r,d}(z)
 =z_1\cdots z_{d-1}z_{d+r}\cdots z_L,
 \qquad
 1\le r\le L,\quad 1\le d\le L-r+1.
 \label{eq:adjacent-burst-map}
\end{equation}
Put \(\mathsf B_0(z)=z\).  A fixed-length code corrects one burst of length
at most two when, for every fixed \(r\in\{0,1,2\}\), distinct codewords have
disjoint sets of length-\(r\) burst descendants; for \(r=0\) the condition
holds for every code, and it is included because the channel of
Theorem~\ref{thm:adjacent-closed-form} also allows zero deletions.  Correcting one burst of
exactly \(r\) adjacent deletions means that this condition holds for that
single value of \(r\).  Descendants obtained with
different values of \(r\) have different lengths, so there is no
cross-length collision to consider.  Unlike Lemma~\ref{lem:exact-monotone}
for arbitrary deletions, correcting one burst of exactly two adjacent
deletions does not imply correcting one burst of length one: the code
\(\{abab:a,b\in\mathcal A\}\) corrects one exact adjacent pair, but for \(a\ne b\) the words \(abab\) and \(baba\) share the one-deletion descendant \(aba\)
(Remark~\ref{rem:adjacent-converse-endpoints}); this is why the three burst
lengths are treated separately here and in the converses.

Fix \(n\ge1\), and set
\[
 M=n+2,\qquad
 M_o=\left\lceil\frac M2\right\rceil,\qquad
 M_e=\left\lfloor\frac M2\right\rfloor.
\]
Thus \(M_o\) and \(M_e\) are the numbers of odd and even positions in a
word of length \(M\); here \(M\) is the length of the word
\(W(x)=(x,A(x),B(x))\) defined in \eqref{eq:adjacent-parity-prefix} below.
Let \(\F_{\mathrm{adj}}=\F_{q_{\mathrm{adj}}}\) be a field of
characteristic two, realized canonically as in Section~\ref{sec:model},
where \(q_{\mathrm{adj}}\) is a power of two depending only on \(n\); its
value is fixed in Section~\ref{subsec:final-symbol}, after the set of
values that the final symbol must avoid has been bounded
(Proposition~\ref{prop:adjacent-prefix-return-bound}), and
Section~\ref{subsec:adjacent-algorithms} verifies the bound
\(q_{\mathrm{adj}}<(n+2)^4\) of Theorem~\ref{thm:adjacent-closed-form}.
We order the elements of \(\F_{\mathrm{adj}}\) by their public binary
labels \(\operatorname{lab}:\F_{\mathrm{adj}}\to\{0,\ldots,q_{\mathrm{adj}}-1\}\)
and identify an integer in that range with the uniquely labelled field
element.  Field addition is bitwise XOR.  The construction below uses this
addition and the public order, but no field multiplication, and it does
not require the order to be compatible with addition.

For \(x=(x_1,\ldots,x_n)\in\F_{\mathrm{adj}}^n\), define
\begin{equation}
 \begin{split}
 A(x)&=\sum_{\substack{1\le i\le n\\i\not\equiv n\pmod2}}x_i,\\
 B(x)&=\sum_{\substack{1\le i\le n\\i\equiv n\pmod2}}x_i,\\
 W(x)&=(x,A(x),B(x))\in\F_{\mathrm{adj}}^M.
 \end{split}
 \label{eq:adjacent-parity-prefix}
\end{equation}
Position \(n+1\) has parity opposite to \(n\), whereas position \(n+2\)
has the same parity as \(n\).  Thus \(A(x)\), placed at position \(n+1\),
is the XOR of the information symbols in its own parity class, and \(B(x)\),
placed at position \(n+2\), is the XOR of those in its class.  Hence, in
characteristic two, the odd- and
even-position subsequences \(W_o(x)\in\F_{\mathrm{adj}}^{M_o}\) and
\(W_e(x)\in\F_{\mathrm{adj}}^{M_e}\) satisfy
\begin{equation}
 \sum_i W_o(x)_i=0,\qquad
 \sum_i W_e(x)_i=0.
 \label{eq:adjacent-zero-sums}
\end{equation}
For \(r\in\{1,2\}\) and every admissible start \(d\), put
\begin{equation}
 P_{r,d}(x)=\mathsf B_{r,d}(W(x))
       \in\F_{\mathrm{adj}}^{M-r}.
 \label{eq:adjacent-prefix-template}
\end{equation}

\begin{lemma}[Recovering the message when the burst start is known]
\label{lem:adjacent-parity-inversion}
Let \(r\in\{1,2\}\), let \(d\) be an admissible start, and let
\(c=P_{r,d}(x)\).  Each symbol of \(W(x)\) deleted by \(\mathsf B_{r,d}\)
equals the XOR of the symbols of \(c\) whose positions in \(W(x)\) have the
same parity as the deleted position.  In particular, \(P_{r,d}\) is
injective, and \(W(x)\), hence \(x\), is recovered from \(c\) and \(d\) in
\(O(n)\) field additions.
\end{lemma}

\begin{proof}
The known start \(d\) identifies the original position, and hence the
parity class, of every surviving symbol.  If \(r=1\), exactly one parity
class has lost a symbol.  By \eqref{eq:adjacent-zero-sums}, that missing
value is the XOR of the surviving values in the same class.  If \(r=2\),
the adjacent pair contains one odd and one even position, so the two
zero-sum equations recover one missing value from each class.  Reinserting
the recovered value or values at the known start reconstructs \(W(x)\);
its first \(n\) coordinates are the unique message \(x\).

The argument applies at every admissible start, including bursts that
delete a parity symbol.  In
particular, it covers deletion of either \(A(x)\) or \(B(x)\), the boundary
pair \((x_n,A(x))\), and the pair \((A(x),B(x))\) of parity symbols.
\end{proof}

For a word \(z\in\F_{\mathrm{adj}}^L\) of any length \(L\ge1\), we use
the Tenengolts checksum \(T_L(z)\) of \eqref{eq:tenengolts-checksum},
formed with the public order fixed above; in particular \(T_1(z)=0\).  We
use the following nonbinary Tenengolts separation for words
\(u,v\in\F_{\mathrm{adj}}^L\):
\begin{equation}
 \mathsf D_1(u)\cap\mathsf D_1(v)\ne\varnothing,\qquad
 \sum_i u_i=\sum_i v_i,\qquad T_L(u)=T_L(v)
 \quad\Longrightarrow\quad u=v.
 \label{eq:adjacent-tenengolts-input}
\end{equation}
For \(L\ge2\), this is Proposition~\ref{prop:tenengolts-one-deletion},
applied with \(n=L\); the result
goes back to Tenengolts~\cite{Tenengolts1984}.  The proof of that
proposition uses only field addition and the total order.
For \(L=1\), the common one-deletion descendant is empty and equality of the
two sums is already equality of the sole symbols.  This endpoint is used
when \(n=1\), for which \(M_e=1\).

\subsection{Three checksums and the final symbol}
\label{subsec:final-symbol}

Applying a VT-type checksum to a full word and to its odd and even
subsequences is a known full/odd/even differential-VT motif; see
Nguyen--Cai--Siegel \cite{NguyenCaiSiegel2024}.  Unlike that work, we use
the Tenengolts checksum \eqref{eq:tenengolts-checksum} itself rather
than a differential VT checksum.  Define
\begin{equation}
 \Lambda(x)=
 \bigl(T_M(W(x)),T_{M_o}(W_o(x)),T_{M_e}(W_e(x))\bigr)
 \in\mathbb Z_M\times\mathbb Z_{M_o}\times\mathbb Z_{M_e}.
 \label{eq:adjacent-locator-triple}
\end{equation}
The whole-word checksum serves a single deletion, which shifts the parity
classes of the suffix; the two class checksums serve an adjacent pair, which
preserves them.

\begin{lemma}[Three-checksum separation]
\label{lem:adjacent-locator-separation}
Let \(r\in\{1,2\}\).  If
\[
 P_{r,d}(x)=P_{r,d'}(y)
\]
for admissible starts \(d,d'\) and distinct messages \(x,y\), then
\(\Lambda(x)\ne\Lambda(y)\).
\end{lemma}

\begin{proof}
Suppose first that \(r=1\).  The words \(W(x)\) and \(W(y)\) share a
one-deletion descendant, and both have total sum zero by
\eqref{eq:adjacent-zero-sums}.  If the first coordinates of \(\Lambda(x)\)
and \(\Lambda(y)\) were equal, \eqref{eq:adjacent-tenengolts-input} would give
\(W(x)=W(y)\), hence \(x=y\).

Now let \(r=2\).  An adjacent pair removes exactly one odd and one even
position.  The suffix following the burst shifts by two positions, so its
positional parity is unchanged.  Consequently \(W_o(x)\) and \(W_o(y)\)
share a one-deletion descendant, and the same is true of
\(W_e(x)\) and \(W_e(y)\).  Each of these four subsequences has sum zero.
Equality of the last two coordinates of \(\Lambda\), followed by
\eqref{eq:adjacent-tenengolts-input}, would give
\[
 W_o(x)=W_o(y),\qquad W_e(x)=W_e(y).
\]
Interleaving the two equal subsequences yields \(W(x)=W(y)\), and again
\(x=y\).  Thus distinct messages cannot have equal checksum triples.
\end{proof}

It remains to store \(\Lambda(x)\) in the final symbol \(H(x)\), which can
itself be deleted.  This creates one new kind of collision.  Suppose a burst
keeps \(H(x)\) in the codeword of \(x\) but deletes \(H(y)\) in the codeword
of \(y\).  Then the last received symbol is \(H(x)\) on one side and a parity
symbol \(A(y)\) or \(B(y)\) on the other.  We first record exactly which
collisions an arbitrary final symbol must avoid.
Lemma~\ref{lem:adjacent-parity-inversion}, the next definitions and
Lemma~\ref{lem:adjacent-collision-characterization} use only field
addition and the burst maps, so they apply over any field \(\F_q\) of
characteristic two, with \(A\), \(B\), \(W\), and \(P_{r,d}\) formed by the
formulas of Section~\ref{subsec:burst-parities}; we state the definitions
and Lemma~\ref{lem:adjacent-collision-characterization} in that generality
because Section~\ref{sec:adjacent-n4} uses them for \(q=2^m\).

For a message \(x\in\F_q^n\), define the \emph{exact list}
\begin{align}
 \mathcal L_x={}&
 \{B(y):y\ne x,\ \mathsf B_{1,d}(W(x))=(y,A(y))
                    \text{ for some }d\}\notag\\
 &{}\cup
 \{A(y):y\ne x,\ \mathsf B_{2,d}(W(x))=y
                    \text{ for some }d\},
 \label{eq:adjacent-exact-list}
\end{align}
and let \(\calG_{n,q}\) be the simple graph on \(\F_q^n\) that joins
distinct messages \(x\) and \(y\) whenever \(W(x)\) and \(W(y)\) have a
common descendant under one deletion or under one adjacent-pair deletion,
that is, whenever \(P_{r,d}(x)=P_{r,d'}(y)\) for some \(r\in\{1,2\}\) and
admissible starts \(d,d'\).  By Lemma~\ref{lem:adjacent-locator-separation},
\(\Lambda\) is a proper coloring of \(\calG_{n,q_{\mathrm{adj}}}\), given by
a closed formula; it plays the role of the computed color \(\ell_1\) of
Section~\ref{sec:graph}.  No common label absent from the last two information
symbols is needed here: the parity symbols lie inside the word \(W\) on which
\(\Lambda\) is computed, so \(\Lambda\) separates every
collision in which the final symbol survives on both sides, whatever the
burst positions; the only remaining collisions, those in which the final
symbol survives on one side alone, are excluded by the avoidance set \(U_x\)
below.

\begin{lemma}[Collisions for the parity prefix]
\label{lem:adjacent-collision-characterization}
Let \(q\) be a power of two, let \(n\ge1\), and let \(H:\F_q^n\to\F_q\) be
any function.  The code \(\{(W(x),H(x)):x\in\F_q^n\}\) corrects one burst
of length at most two if and only if
\begin{enumerate}[label=(\roman*),leftmargin=*]
\item \(H(x)\notin\mathcal L_x\) for every \(x\), and
\item \(H(x)\ne H(y)\) for every edge \(xy\) of \(\calG_{n,q}\).
\end{enumerate}
\end{lemma}

\begin{proof}
First let (i) and (ii) hold.  Fix \(r\in\{0,1,2\}\), and suppose that
bursts of length \(r\) applied to the codewords of two messages \(x\) and
\(y\) give the same word.  The following cases exhaust all burst starts,
including those crossing the boundary between the information symbols and
the redundant symbols.
\begin{enumerate}[leftmargin=*]
\item If \(r=0\), the codewords are equal, and their systematic prefixes
      give \(x=y\).

\item Suppose \(H\) survives on both sides.  The two received prefixes then
      have the form \(P_{r,d}(x)=P_{r,d'}(y)\), and equality of the final
      received symbols gives \(H(x)=H(y)\).  If \(x\ne y\), then \(xy\) is
      an edge of \(\calG_{n,q}\) with equal values of \(H\), contrary to
      (ii).  Hence \(x=y\).

\item Suppose \(H\) survives on exactly one side; after exchanging the two
      messages if necessary, \(H(x)\) survives and \(H(y)\) is deleted.  If
      \(r=1\), deleting \(H(y)\) leaves \(W(y)\), whereas the side retaining
      \(H(x)\) has the form \((P_{1,d}(x),H(x))\).  Equality says precisely
      that \(P_{1,d}(x)=(y,A(y))\) and \(H(x)=B(y)\).  If \(r=2\), the only
      adjacent pair that deletes \(H(y)\) is \((B(y),H(y))\), and that side
      leaves \((y,A(y))\); the side retaining \(H(x)\) has the form
      \((P_{2,d}(x),H(x))\), so \(P_{2,d}(x)=y\) and \(H(x)=A(y)\).  In both
      cases \(y\ne x\) would put \(H(x)\) into \(\mathcal L_x\), contrary to
      (i).  Hence \(x=y\).

\item Suppose \(H\) is deleted on both sides and \(r=1\).  Both bursts
      delete the last coordinate, so the two outputs are \(W(x)\) and
      \(W(y)\).  Their first \(n\) coordinates give \(x=y\).

\item Suppose \(H\) is deleted on both sides and \(r=2\).  On both sides the
      deleted pair must be \((B,H)\), so the outputs are
      \((x,A(x))\) and \((y,A(y))\).  Again their first \(n\) coordinates
      give \(x=y\).
\end{enumerate}
Conversely, suppose that (i) fails.  If \(H(x)=B(y)\) for some \(y\ne x\)
with \(\mathsf B_{1,d}(W(x))=(y,A(y))\), then deleting position \(d\) from
the codeword of \(x\) gives \((y,A(y),B(y))=W(y)\), which is also obtained
by deleting \(H(y)\) from the codeword of \(y\): two distinct codewords
share a length-one burst descendant.  If \(H(x)=A(y)\) for some \(y\ne x\)
with \(\mathsf B_{2,d}(W(x))=y\), then deleting positions \(d,d+1\) from the
codeword of \(x\) gives \((y,A(y))\), as does deleting the pair
\((B(y),H(y))\) from the codeword of \(y\).  If (ii) fails, let \(xy\) be an
edge with \(H(x)=H(y)\), and let \(z=P_{r,d}(x)=P_{r,d'}(y)\) be a common
descendant of the two prefixes; both bursts lie inside the prefixes, so
\((z,H(x))=(z,H(y))\) is a common length-\(r\) burst descendant of the two
codewords.  Thus (i) and (ii) are necessary.
\end{proof}

In the construction, \(H(x)\) is chosen outside a superset
\(U_x\supseteq\mathcal L_x\).  For a fixed message \(x\), the
\emph{avoidance set} \(U_x\subseteq\F_q\) is defined by the two clauses of
\eqref{eq:adjacent-exact-list} without the restriction \(y\ne x\):
\begin{itemize}[leftmargin=*]
\item \emph{One-deletion clause.}  For every \(d\in\{1,\ldots,M\}\), let
  \(\widetilde w=P_{1,d}(x)\), and let \(y\) be the word formed by the
  first \(n\) symbols of \(\widetilde w\).  If \(\widetilde w=(y,A(y))\),
  include \(B(y)\) in \(U_x\).
\item \emph{Two-deletion clause.}  For every \(d\in\{1,\ldots,M-1\}\),
  include \(A(P_{2,d}(x))\) in \(U_x\).
\end{itemize}
The set \(U_x\) may contain values that arise from the case \(y=x\), which
causes no collision; we use it in place of \(\mathcal L_x\) because it has
an explicit form (Proposition~\ref{prop:adjacent-forbidden-formula}) and a
bound on its size, the budget \(\kappa_n\) of
Proposition~\ref{prop:adjacent-prefix-return-bound}, which fixes the
alphabet.

To describe \(U_x\) explicitly, write
\[
 \begin{aligned}
 \Sigma_0(x)&=0,\qquad
 \Sigma_d(x)=\sum_{i=1}^d x_i\quad(1\le d\le n),\\
 \mathcal R_x&=
 \{W(x)_i:1\le i\le M,\ i\not\equiv n\pmod2\}.
 \end{aligned}
\]
Thus \(\Sigma_d(x)\) is a prefix sum, and \(\mathcal R_x\) is the set of
values in the parity class of positions that contains \(A(x)\).  Put
\[
 \mathcal Z_x=
 \{x_d:1\le d\le n,\ d\equiv n\pmod2,\ \Sigma_d(x)=0\}.
\]

\begin{proposition}[Explicit form of the avoidance set]
\label{prop:adjacent-forbidden-formula}
For every message,
\begin{equation}
 U_x=\mathcal R_x\cup\{B(x)\}\cup\mathcal Z_x.
 \label{eq:adjacent-forbidden-formula}
\end{equation}
In particular, every element of \(U_x\) is a coordinate of \(W(x)\).
Explicitly: if \(n=2k\) and \(x=(a_1,b_1,\ldots,a_k,b_k)\), put
\(\psi_j=\Sigma_{2j}(x)=\sum_{i\le j}(a_i+b_i)\) and
\(J=\{1\le j\le k:\psi_j=0\}\); then
\begin{equation}
 U_x=\{a_1,\ldots,a_k,A(x),B(x)\}
       \cup\{b_j:j\in J\}.
 \label{eq:adjacent-even-forbidden-form}
\end{equation}
If \(n=2k+1\) and \(x=(b_0,a_1,b_1,\ldots,a_k,b_k)\), put
\(\psi_j=\Sigma_{2j+1}(x)=b_0+\sum_{i\le j}(a_i+b_i)\) and
\(J=\{0\le j\le k:\psi_j=0\}\); then
\begin{equation}
 U_x=\{a_1,\ldots,a_k,A(x),B(x)\}
       \cup\{b_j:j\in J\}.
 \label{eq:adjacent-odd-forbidden-form}
\end{equation}
\end{proposition}

\begin{proof}
The two-deletion clause contributes exactly \(\mathcal R_x\).  An
adjacent pair contains one position of each parity, and shifting the suffix
by two preserves parity.  The surviving sum in the class containing
\(A(x)\) is therefore the deleted coordinate in that class.  Conversely,
as the adjacent start varies, every coordinate in that class is deleted by
some pair.

Now delete an information symbol \(x_d\), and write
\[
 P_{1,d}(x)=(y,B(x)).
\]
Here \(y_i=x_i\) for \(i<d\), \(y_i=x_{i+1}\) for \(d\le i<n\), and
\(y_n=A(x)\).  The symbols \(x_{d+1},\ldots,x_n\) move one position to the
left and therefore change parity class, while \(y_n=A(x)\) lies in the
class \(i\equiv n\).  Hence
\[
 A(y)=\sum_{\substack{i<d\\ i\not\equiv n}}x_i
      +\sum_{\substack{d<j\le n\\ j\equiv n}}x_j,
 \qquad\text{so}\qquad
 A(y)+B(x)=\sum_{\substack{i<d\\ i\not\equiv n}}x_i
      +\sum_{\substack{j\le d\\ j\equiv n}}x_j,
\]
with congruences modulo two.  The right side is a prefix sum, namely
\begin{equation}
 A(y)+B(x)=
 \begin{cases}
  \Sigma_d(x),&d\equiv n\pmod2,\\
  \Sigma_{d-1}(x),&d\not\equiv n\pmod2.
 \end{cases}
 \label{eq:adjacent-prefix-return-test}
\end{equation}
The one-deletion clause applies at \(d\) exactly when
\(P_{1,d}(x)=(y,A(y))\), that is, when \(B(x)=A(y)\), that is, when the
right side of \eqref{eq:adjacent-prefix-return-test} vanishes.  In that
case the value put into \(U_x\) is \(B(y)=x_d\): the XOR of all symbols of
\(P_{1,d}(x)\) equals the XOR of \(W(x)\), which is zero by
\eqref{eq:adjacent-zero-sums}, plus \(x_d\); it also equals
\(A(y)+B(y)+A(y)=B(y)\).  If \(d\not\equiv n\), then \(x_d\) lies in the
parity class of \(A(x)\), so it already belongs to \(\mathcal R_x\).  If
\(d\equiv n\), the clause applies exactly when \(\Sigma_d(x)=0\), and the
values obtained form \(\mathcal Z_x\).  Deleting \(B(x)\) (start
\(d=M\)) leaves \((x,A(x))\), so the clause applies with \(y=x\) and
contributes \(B(x)\).  Deleting \(A(x)\) (start \(d=n+1\)) leaves
\((x,B(x))\); the clause applies only if \(B(x)=A(x)\), and then the
contributed value \(B(x)=A(x)\) already lies in \(\mathcal R_x\).  Combining the
one- and two-deletion clauses proves
\eqref{eq:adjacent-forbidden-formula}.

In the interleaved coordinates, the positions \(i\not\equiv n\pmod 2\) of
\(W(x)\) carry \(a_1,\ldots,a_k\) and \(A(x)\), so
\(\mathcal R_x=\{a_1,\ldots,a_k,A(x)\}\); the positions
\(d\equiv n\pmod 2\) carry the \(b_j\), with \(\Sigma_d(x)=\psi_j\), so
\(\mathcal Z_x=\{b_j:j\in J\}\).  Thus \eqref{eq:adjacent-forbidden-formula}
reads \eqref{eq:adjacent-even-forbidden-form} for even \(n\) and
\eqref{eq:adjacent-odd-forbidden-form} for odd \(n\).
\end{proof}

The explicit formula gives a better bound than the trivial bound
\(|U_x|\le M\): the budget \(\kappa_n\) that fixes the alphabet.  Put
\begin{equation}
 \kappa_n=
 \begin{cases}
  2,&n=1,2,\\
  3,&n=3,\\
  k+2+\lfloor(k-1)/2\rfloor,&n=2k\ge4,\\
  k+2+\lceil k/2\rceil,&n=2k+1\ge5.
 \end{cases}
 \label{eq:adjacent-scan-budget}
\end{equation}

\begin{proposition}[Size of the avoidance set]
\label{prop:adjacent-prefix-return-bound}
For every \(n\ge1\) and every message,
\begin{equation}
 |U_x|\le \kappa_n,
 \label{eq:adjacent-forbidden-bound}
\end{equation}
where \(\kappa_n\) is defined in
\eqref{eq:adjacent-scan-budget}.
\end{proposition}

\begin{proof}
First let \(n=2k\), with \(a_i\), \(b_j\), \(\psi_j\) and \(J\) as in
Proposition~\ref{prop:adjacent-forbidden-formula}, so that \(U_x\) is given
by \eqref{eq:adjacent-even-forbidden-form}.
The first set, the \emph{base}, has at most \(k+2\) values.  Call
\(j\in J\) \emph{new} if \(b_j\) is not in the base.  If \(1\in J\), then
\(\psi_1=a_1+b_1=0\), so \(b_1=a_1\) and \(1\) is not new.  If both \(j-1\) and
\(j\) lie in \(J\), then \(a_j+b_j=\psi_j+\psi_{j-1}=0\), so \(b_j=a_j\) and \(j\) is
not new.  Hence a new index \(j<k\) satisfies \(j\ge2\) and
\(j-1\notin J\), so no two new indices below \(k\) are consecutive, and
there are at most \(\lfloor(k-1)/2\rfloor\) of them in
\(\{2,\ldots,k-1\}\).  If \(k\in J\), then \(\psi_k=A(x)+B(x)=0\), so
\(A(x)=B(x)\) and the base has at most \(k+1\) distinct values; a possible
new value \(b_k\) therefore does not increase the bound \(k+2\).  This proves the even line of
\eqref{eq:adjacent-scan-budget}; the case \(k=1\) gives
\(|U_x|\le2\), since then \(A(x)=a_1\) and the base has at most two values.

For \(n=2k+1\), with the coordinates of
Proposition~\ref{prop:adjacent-forbidden-formula}, \(U_x\) is given by
\eqref{eq:adjacent-odd-forbidden-form}.
The base \(\{a_1,\ldots,a_k,A(x),B(x)\}\) again has at most \(k+2\) values;
call \(j\in J\) new if \(b_j\) is not in the base.  For \(1\le j\le k\), if
\(j-1\) and \(j\) both lie in \(J\), then \(\psi_{j-1}=\psi_j=0\) forces
\(b_j=a_j\), so \(j\) is not new.  Hence no two new indices in
\(\{0,\ldots,k-1\}\) are consecutive, and at most \(\lceil k/2\rceil\) of
them are new.  If \(k\in J\), then \(A(x)=B(x)\), and as in the even case a
possible new value \(b_k\) does not increase the bound.  When \(k=1\), also
\(A(x)=a_1\), so the base has at most two values and
\(|U_x|\le3\).  For the remaining endpoint \(n=1\), \(A(x)=0\) (an empty sum) and \(B(x)=x_1\), so \(U_x\subseteq\{0,x_1\}\) by
\eqref{eq:adjacent-forbidden-formula}.
\end{proof}

We now fix the size of the alphabet \(\F_{\mathrm{adj}}\) of
Section~\ref{subsec:burst-parities} and define the final symbol.  The
final symbol \(H(x)\) is chosen outside \(U_x\); since
\(|U_x|\le\kappa_n\), any \(R_n=\kappa_n+1\) distinct candidate values
contain an admissible one.  Put \(K=M\cdot M_o\cdot M_e\cdot R_n\): here
\(MM_oM_e\) is the number of checksum triples in
\(\mathbb Z_M\times\mathbb Z_{M_o}\times\mathbb Z_{M_e}\), and \(R_n\) is
the number of candidate values of the final symbol reserved for each
triple.  Let \(q_{\mathrm{adj}}\) be the least power of two at least
\(K\).  Represent the three residues of a checksum triple by
\[
 0\le\lambda<M,\qquad 0\le\lambda_o<M_o,\qquad 0\le\lambda_e<M_e.
\]
For \(0\le j<R_n\), define the mixed-radix field label
\begin{equation}
 \iota(\lambda,\lambda_o,\lambda_e,j)
   =\bigl(((\lambda\,M_o+\lambda_o)M_e+\lambda_e)R_n+j\bigr).
 \label{eq:adjacent-mixed-radix}
\end{equation}
The right-hand side ranges bijectively over
\(\{0,\ldots,K-1\}\), which lies among the public field labels because
\(q_{\mathrm{adj}}\ge K\).  If
\(\Lambda(x)=(\lambda,\lambda_o,\lambda_e)\), choose
\begin{equation}
 j_x=\min\{0\le j<R_n:
          \iota(\lambda,\lambda_o,\lambda_e,j)\notin U_x\},
 \qquad
 H(x)=\iota(\lambda,\lambda_o,\lambda_e,j_x),
 \label{eq:adjacent-final-label}
\end{equation}
and define the systematic codeword
\begin{equation}
 E_n^{\mathrm{adj}}(x)
   =(W(x),H(x))=(x,A(x),B(x),H(x)).
 \label{eq:adjacent-encoder}
\end{equation}

\begin{lemma}[Collisions in which one final symbol is deleted]
\label{lem:adjacent-final-label-avoidance}
The choice in \eqref{eq:adjacent-final-label} always exists.  Moreover,
suppose that a burst of length \(r\in\{1,2\}\) applied to
\(E_n^{\mathrm{adj}}(x)\) and a burst of the same length applied to
\(E_n^{\mathrm{adj}}(y)\) give the same word, where the first burst does
not delete \(H(x)\) and the second deletes \(H(y)\).  Then
\(H(x)\in U_x\).
\end{lemma}

\begin{proof}
For the fixed checksum triple \(\Lambda(x)\), the \(R_n=\kappa_n+1\) values in
\eqref{eq:adjacent-mixed-radix} are distinct.  At most \(\kappa_n\) of them
belong to \(U_x\) by \eqref{eq:adjacent-forbidden-bound}, proving
existence.

The second claim is the case of the proof of
Lemma~\ref{lem:adjacent-collision-characterization} in which \(H\)
survives on exactly one side, read for the superset
\(U_x\supseteq\mathcal L_x\): equality of the received words forces
\(P_{1,d}(x)=(y,A(y))\) and \(H(x)=B(y)\) when \(r=1\), and
\(P_{2,d}(x)=y\) and \(H(x)=A(y)\) when \(r=2\), which are the two clauses
defining \(U_x\).
\end{proof}

\subsection{Correction and algorithms}
\label{subsec:adjacent-algorithms}

\begin{proof}[Proof of Theorem~\ref{thm:adjacent-closed-form}]
By Lemma~\ref{lem:adjacent-collision-characterization}, applied over
\(\F_{\mathrm{adj}}\), the code corrects one burst of length at most two as
soon as the final symbol \(H\) of \eqref{eq:adjacent-final-label} satisfies
conditions (i) and (ii) of that lemma.  Condition~(i): the clauses defining
\(\mathcal L_x\) are those defining \(U_x\) restricted to \(y\ne x\), so
\(\mathcal L_x\subseteq U_x\), and \(H(x)\notin U_x\) by
\eqref{eq:adjacent-final-label}; the choice in
\eqref{eq:adjacent-final-label} exists by
Lemma~\ref{lem:adjacent-final-label-avoidance}.  Condition~(ii): if \(xy\)
is an edge of \(\calG_{n,q_{\mathrm{adj}}}\), then
\(P_{r,d}(x)=P_{r,d'}(y)\) for some \(r\in\{1,2\}\) and admissible starts
\(d,d'\), so Lemma~\ref{lem:adjacent-locator-separation} gives
\(\Lambda(x)\ne\Lambda(y)\); since the map \(\iota\) is injective in all
four displayed arguments, \(H(x)\ne H(y)\).

This proves zero-error correction for a burst of length at most two.  For
the alphabet estimate, minimality of the chosen power of two gives
\begin{equation}
 q_{\mathrm{adj}}<2K=2MM_oM_eR_n
 \le \frac{M^3R_n}{2}
 \le \frac38M^4+\frac34M^3
 <M^4=(n+2)^4,
 \label{eq:adjacent-alphabet-bound}
\end{equation}
where \(4M_oM_e\le M^2\), \(M\ge3\), and \(R_n\le3M/4+3/2\) by
\eqref{eq:adjacent-scan-budget} (for \(n\le3\) directly; for \(n=2k\) since
\(\lfloor(k-1)/2\rfloor\le k/2\), and for \(n=2k+1\) since
\(\lceil k/2\rceil\le(k+1)/2\)).

The encoder computes \(A(x)\), \(B(x)\), and the three Tenengolts checksums
in \(O(n)\) operations on \(O(\log q_{\mathrm{adj}})\)-bit labels.  A direct
construction of \(U_x\) scans \(2M-1\) starts and recomputes the relevant
parity sum for each start, using \(O(n^2)\) label operations.  It then tests
at most the \(R_n\) labels \(\iota(\Lambda(x),j)\), \(0\le j<R_n\).  Thus encoding has polynomial
bit complexity.

For decoding, reject a word whose length does not reveal a value
\(r\in\{0,1,2\}\).  If \(r=0\), take the first \(n\) symbols as the
candidate message and return it exactly when re-encoding reproduces the
received word.  For \(r\in\{1,2\}\), enumerate the \(O(n)\) possible burst
starts.  When \(H\) survives, split off the last received symbol, attempt
the inverse of \(P_{r,d}\) from
Lemma~\ref{lem:adjacent-parity-inversion}, and discard the proposed start
when the reconstructed word is not of the form \(W(x')\) for a message
\(x'\).  When \(H\) is
deleted, the only possibilities are the final coordinate for \(r=1\) and
the final pair \((B,H)\) for \(r=2\); in either case the first \(n\)
received symbols give the candidate message.  Re-encode every candidate,
retain it only when deletion at the proposed start reproduces the received
word, and deduplicate equal candidates.  A valid channel output retains its
true message, while the correction proof gives at most one retained
message.  Direct re-encoding at every start takes \(O(n^3)\) label
operations, up to polynomial factors in \(\log q_{\mathrm{adj}}\).

The canonical field representation can itself be found in polynomial time
in \(n\), since \(q_{\mathrm{adj}}<(n+2)^4\).  Once it is fixed, encoding
and decoding use XOR, order comparisons, and direct enumeration, with all
public data determined by \(n\).  This proves all assertions of the theorem
for every \(n\ge1\); the range includes \(n=1\), where the construction is
valid although two redundant symbols suffice
(Remark~\ref{rem:adjacent-converse-endpoints}).
\end{proof}

The budget \(\kappa_n\) used in the construction cannot be improved as a
bound on \(|U_x|\).  For
the following sharpness statement and its proof in
Appendix~\ref{app:adjacent-sharpness}, let \(q\) be any power of two and
instantiate \(A,B,W,P_{r,d}\), and \(U_x\) by the formulas of
Sections~\ref{subsec:burst-parities}--\ref{subsec:final-symbol} with
\(\F_q\) in place of \(\F_{\mathrm{adj}}\).

\begin{proposition}[The size bound is attained]
\label{prop:adjacent-prefix-return-sharpness}
For every power of two \(q\) and every \(n\ge1\),
\[
 \max_{x\in\F_q^n}|U_x|=\kappa_n
 \quad\Longleftrightarrow\quad q\ge \kappa_n.
\]
Equivalently, the least power-of-two field size at which
\eqref{eq:adjacent-forbidden-bound} is attained is
\[
 2^{\lceil\log_2 \kappa_n\rceil}.
\]
In particular, the bound is attained over the field
\(\F_{\mathrm{adj}}\) used in
Theorem~\ref{thm:adjacent-closed-form}.
\end{proposition}

\begin{proof}
The upper bound is
Proposition~\ref{prop:adjacent-prefix-return-bound}, and
\(U_x\subseteq\F_q\) makes equality impossible when \(q<\kappa_n\).
Appendix~\ref{app:adjacent-sharpness} constructs an \(x\) with
\(|U_x|=\kappa_n\) whenever \(q\ge \kappa_n\).  Finally,
\(q_{\mathrm{adj}}\ge K=M M_oM_e(\kappa_n+1)>\kappa_n\).
\end{proof}

Thus the bound on \(|U_x|\) is sharp.  This proposition is not needed for
Theorem~\ref{thm:adjacent-closed-form}; it shows that the alphabet of that
construction cannot be reduced by a better bound on \(|U_x|\) alone.
Choosing the final symbols of different messages jointly could lead to a
smaller alphabet; the present construction chooses each final symbol
separately and does not optimize this.

\section{Converse bounds for adjacent bursts}
\label{sec:adjacent-converses}

This section proves two lower bounds.  The first,
Theorem~\ref{thm:adjacent-converse}, applies to every code of size \(q^n\);
it shows that, for \(n\ge2\), Theorem~\ref{thm:adjacent-closed-form} uses
the least possible number of redundant symbols, and that with three
redundant symbols the alphabet must
grow at least linearly.  The second, Theorem~\ref{thm:adjacent-parity-prefix-converse},
applies only to codes with the specific structure of
Section~\ref{sec:adjacent}: parity symbols for residue classes of positions,
followed by one freely chosen final symbol.  It quantifies how large the
alphabet must be for that structure. The inequality of Theorem~\ref{thm:adjacent-converse} and the bound of Theorem~\ref{thm:adjacent-parity-prefix-converse} are proved for exact adjacent bursts; they apply to the channel of Theorem~\ref{thm:adjacent-closed-form} because, by the definition in Section~\ref{sec:adjacent}, a code correcting one burst of length at most two corrects in particular one burst of exactly two adjacent deletions (the converse fails, Remark~\ref{rem:adjacent-converse-endpoints}).

\begin{theorem}[Adjacent-burst converse]
\label{thm:adjacent-converse}
Let \(n\ge1\) and \(r_{\mathrm{red}}\ge0\) be integers, let \(q\ge2\),
and let \(\C\subseteq\mathcal A^{n+r_{\mathrm{red}}}\) have
\(|\C|=q^n\), where \(|\mathcal A|=q\) and \(n+r_{\mathrm{red}}\ge4\).
If \(\C\) corrects one exact
adjacent pair of deletions, then
\[
 (q-1)(n+r_{\mathrm{red}}-3)\le q^{r_{\mathrm{red}}-1}-q^{2-n}.
\]
Consequently \(r_{\mathrm{red}}\ge3\) whenever \(n\ge3\); and if
\(r_{\mathrm{red}}=3\), then \(q\ge n-1\). If \(\C\) corrects one burst of length at
most two, then \(r_{\mathrm{red}}\ge3\) already for \(n=2\).
\end{theorem}

Together with Theorem~\ref{thm:adjacent-closed-form}, the theorem shows
that three redundant symbols are the optimal number for one burst of length
at most two, among all codes of size \(q^n\), for every \(n\ge2\).  The last clause
is stated for the at-most-two channel because the
inequality is an equality at \(n=2\), \(r_{\mathrm{red}}=2\), where two
redundant symbols do correct one exact adjacent pair
(Remark~\ref{rem:adjacent-converse-endpoints}).
The proof of Theorem~\ref{thm:adjacent-converse} treats the at-most-two
obstruction at \(n=2\) separately.  The hypothesis
\(n+r_{\mathrm{red}}\ge4\) is the condition \(L\ge4\) of
Proposition~\ref{prop:adjacent-fractional-packing}, from which the
inequality is derived; the cases it excludes with \(n\ge2\) have
\(r_{\mathrm{red}}\le1\) and need no inequality, because one adjacent-pair
deletion leaves only \(q^{n+r_{\mathrm{red}}-2}<q^n\) possible words, so no
code of size \(q^n\) corrects one exact adjacent pair.

The rational expression below appears for even block length in
Wang--Tang--Sima--Gabrys--Farnoud
\cite[arXiv:2210.11818v1, Theorem~3]{WangTangSimaGabrysFarnoud2024}; see also the binary
precursor~\cite[Theorem~4]{SchoenyWachterZehGabrysYaakobi2017} and the
recent summary in Wang--Kong--Yaakobi--Duman
\cite[Table~I]{WangKongYaakobiDuman2026}.  For every block length, the expression is
the value for one exact adjacent pair of the general fractional-transversal
bound \cite[Theorem~3]{WangKongYaakobiDuman2026}; the case of odd \(L\) is
therefore an evaluation of that
known bound.  For completeness, the following proof,
by the fractional-transversal method of Kulkarni and
Kiyavash~\cite{KulkarniKiyavash2013} used in
\cite{SchoenyWachterZehGabrysYaakobi2017,WangTangSimaGabrysFarnoud2024},
derives the bound in this case and evaluates it for every \(L\ge4\).
The proof gives every word in \(\mathcal A^{L-2}\)
a weight equal to the reciprocal of its own number of adjacent-pair
descendants, a number that never exceeds that of a word it descends from;
the descendants of any codeword then carry total weight at least one, so the
sum of all weights bounds the number of codewords.

\begin{proposition}[Adjacent-pair fractional-transversal bound]
\label{prop:adjacent-fractional-packing}
If \(\C\subseteq\mathcal A^L\), \(L\ge4\), corrects one exact adjacent
pair of deletions and \(|\mathcal A|=q\ge2\), then
\begin{equation}
 |\C|\le
 \left\lfloor
 \frac{q^{L-1}-q^2}{(q-1)(L-3)}
 \right\rfloor .
 \label{eq:adjacent-fractional-bound}
\end{equation}
\end{proposition}

\begin{proof}
For a word \(z\) of length at least two, let \(z_o\) and \(z_e\) be its
odd- and even-coordinate subsequences, let \(\operatorname{runs}(u)\)
denote the number of runs (maximal blocks of equal consecutive symbols) of a
nonempty word \(u\), and put
\[
 \mathcal B_2(z)=
 \{\mathsf B_{2,d}(z):1\le d\le |z|-1\},
 \qquad
 \mu_{\mathrm{adj}}(z)=
 \operatorname{runs}(z_o)+\operatorname{runs}(z_e)-1.
\]
We first claim that
\begin{equation}
 |\mathcal B_2(z)|=\mu_{\mathrm{adj}}(z).
 \label{eq:adjacent-shadow-run-count}
\end{equation}
This count is classical: it is Levenshtein's formula for the number of
distinct outputs of a burst of two deletions~\cite{Levenshtein1967Adjacent},
restated in \cite[eq.~(2)]{SchoenyWachterZehGabrysYaakobi2017} and, for
arbitrary alphabets, in \cite[Sec.~I-A]{WangKongYaakobiDuman2026}; we include
a short proof for completeness.
Call \(z_o\) and \(z_e\) the two \emph{rows} of \(z\).  An adjacent deletion
removes one coordinate from each row, and the output is determined by its
two rows.  Deleting different coordinates of one row gives the same row
exactly when the two coordinates lie in the same run.  Hence
\(\mathsf B_{2,d}(z)\) is determined by the pair (run of \(z_o\) containing
the deleted odd coordinate, run of \(z_e\) containing the deleted even
coordinate), and distinct pairs give distinct outputs.  As the start \(d\)
increases by one, exactly one of the two deleted row indices increases by
one; the row indices follow the path
\[
 (1,1),(2,1),(2,2),(3,2),(3,3),\ldots,
\]
truncated at the appropriate endpoint.  Along this path the pair of runs
changes in one coordinate at a time and never decreases, and it passes
through every run of each row.  It therefore visits exactly
\(1+(\operatorname{runs}(z_o)-1)+(\operatorname{runs}(z_e)-1)
=\mu_{\mathrm{adj}}(z)\)
distinct pairs, proving~\eqref{eq:adjacent-shadow-run-count} for both
parities of \(|z|\).

Now set \(m=L-2\).  If \(y\in\mathcal B_2(x)\), then each row of \(y\)
is a one-deletion descendant of the corresponding row of \(x\).  Deletion
cannot increase the number of runs
(cf.\ \cite[Lemma~1]{WangKongYaakobiDuman2026}), so
\(\mu_{\mathrm{adj}}(y)\le\mu_{\mathrm{adj}}(x)\).  Hence the weights
\[
 w(y)=\frac1{\mu_{\mathrm{adj}}(y)},\qquad y\in\mathcal A^m,
\]
form a fractional transversal of the adjacent-deletion balls: by
\eqref{eq:adjacent-shadow-run-count},
\[
 \sum_{y\in\mathcal B_2(x)}w(y)
 \ge \frac{|\mathcal B_2(x)|}{\mu_{\mathrm{adj}}(x)}=1.
\]

It remains to sum the weights.  Write
\(m_o=\lceil m/2\rceil\) and \(m_e=\lfloor m/2\rfloor\).
A row of length \(m_o\) with \(i_o\) runs can be chosen in
\(q(q-1)^{i_o-1}\binom{m_o-1}{i_o-1}\) ways, and similarly for the other row;
since \(\mu_{\mathrm{adj}}(y)=i\) means \(i_o+i_e=i+1\), the number of
length-\(m\) words with \(\mu_{\mathrm{adj}}(y)=i\) is
\[
 q^2(q-1)^{i-1}
 \sum_{i_o+i_e=i+1}
 \binom{m_o-1}{i_o-1}\binom{m_e-1}{i_e-1}
 =q^2(q-1)^{i-1}\binom{m-2}{i-1},
\]
where the last identity is Vandermonde's convolution.  Using
\(\binom{m-2}{i-1}/i=\binom{m-1}{i}/(m-1)\) and the binomial theorem,
we obtain
\begin{align*}
 \sum_{y\in\mathcal A^m}w(y)
 &=q^2\sum_{i=1}^{m-1}
   \frac{(q-1)^{i-1}}{i}\binom{m-2}{i-1}\\
 &=\frac{q^2(q^{m-1}-1)}{(q-1)(m-1)}
  =\frac{q^{L-1}-q^2}{(q-1)(L-3)}.
\end{align*}
The adjacent-deletion balls of distinct codewords are disjoint, so summing the weights over
them bounds \(|\C|\) by this quantity.  Taking the integer part proves
\eqref{eq:adjacent-fractional-bound}; when \(L=4\), the same calculation
gives the endpoint value \(q^2\).
\end{proof}

\begin{proof}[Proof of Theorem~\ref{thm:adjacent-converse}]
Apply Proposition~\ref{prop:adjacent-fractional-packing} with
\(L=n+r_{\mathrm{red}}\) and \(|\C|=q^n\).  After division by \(q^n\) it gives
\[
 (q-1)(n+r_{\mathrm{red}}-3)
 \le q^{r_{\mathrm{red}}-1}-q^{2-n}.
\]
This inequality is used below only for \(r_{\mathrm{red}}\in\{2,3\}\); the
case \(r_{\mathrm{red}}\le1\) is excluded directly.  Deleting the first
adjacent pair is injective on \(\C\), since equality of two resulting words
would itself be a channel collision.  Hence
\(q^n\le q^{n+r_{\mathrm{red}}-2}\), so \(r_{\mathrm{red}}\ge2\).
For \(r_{\mathrm{red}}=2\) and \(n\ge3\), the left side is at least
\(2(q-1)\ge q\), whereas the right side is strictly smaller than \(q\).
Thus \(r_{\mathrm{red}}\ge3\).

For \(r_{\mathrm{red}}=3\),
\[
 n(q-1)\le q^2-q^{2-n}<q^2.
\]
If \(q\le n-2\), then
\[
 n(q-1)\ge(q+2)(q-1)=q^2+q-2\ge q^2,
\]
a contradiction.  Hence \(q\ge n-1\).

It remains to check the endpoint used by the at-most-two statement.  Let
\(n=2\) and suppose that \(r_{\mathrm{red}}=2\) and that \(\C\) corrects one
burst of length at most two.
Its \(q^2\) sets \(\mathcal B_2(c)\), \(c\in\C\), of exact-pair descendants are
nonempty, pairwise disjoint subsets of \(\mathcal A^2\), so each of them
is a singleton.  For a word
\(z_1z_2z_3z_4\), equality of the outputs obtained by deleting at starts
one and two forces \(z_1=z_3\), and equality at starts two and three
forces \(z_2=z_4\).  Thus every codeword has period two, and since exactly
\(q^2\) words of length four have period two, \(\C\) consists of all of
them.  For \(a\ne b\), however, the one-deletion balls of the codewords
\(abab\) and \(baba\) share both \(aba\) and
\(bab\), contradicting at-most-two correction.  Thus
\(r_{\mathrm{red}}\ge3\) at
\(n=2\), and Theorem~\ref{thm:adjacent-closed-form} supplies equality.
\end{proof}

\begin{remark}[Exact-two endpoints]
\label{rem:adjacent-converse-endpoints}
At \(n=1\), \(E(a)=aaa\) corrects a burst of length at most two with
only two redundant symbols.  At \(n=2\), \(E(a,b)=abab\) corrects an
exact adjacent pair with two redundant symbols, although the preceding
proof shows that it cannot correct the at-most-two channel.  Thus the
exact-two range \(n\ge3\) in the symbol-count converse is sharp.
\end{remark}

\subsection{A converse for one XOR parity per residue class and one final symbol}

Theorem~\ref{thm:adjacent-converse} bounds the alphabet of every code with
three redundant symbols that corrects one exact adjacent pair of deletions
by \(q\ge n-1\), whereas the
construction of Theorem~\ref{thm:adjacent-closed-form} uses
\(q_{\mathrm{adj}}<(n+2)^4\).  The following theorem asks how much of this
gap is forced by the structure of that construction, namely the information
symbols followed by one XOR parity per residue class of positions and then
by one freely chosen final symbol, and proves a lower bound for every code
of that form.  The argument does not use \(t=2\) anywhere, so we state it
for \(t\) residue classes modulo \(t\) and bursts of exactly \(t\) adjacent
deletions; \(t=2\) is the structure of Theorem~\ref{thm:adjacent-closed-form}
(for \(t=2\) the parities defined below are \(C_1=A\) and \(C_2=B\), so
\(W_{n,2}=W\)),
and Section~\ref{sec:adjacent-n4} determines the exact threshold in its
first case \(n=4\) for bursts of length at most two.  The final symbol \(H\)
may be an arbitrary function of
the message.

\begin{theorem}[Converse for one XOR parity per residue class and one final symbol]
\label{thm:adjacent-parity-prefix-converse}
Let \(t\ge2\), \(n\ge2t\), and \(q=2^m\) with \(m\ge1\).
For \(1\le j\le t\), put
\[
 C_j(x)=
 \bigoplus_{\substack{1\le i\le n\\i\equiv n+j\pmod t}}x_i,
 \qquad
 W_{n,t}(x)=(x,C_1(x),\ldots,C_t(x)).
\]
If an arbitrary function \(H:\F_q^n\to\F_q\) makes the family
\(\{(W_{n,t}(x),H(x)):x\in\F_q^n\}\) correct one exact adjacent burst
of length \(t\), then
\(q\ge n+\lfloor n/t\rfloor+2\).
\end{theorem}

\begin{proof}[Proof of
Theorem~\ref{thm:adjacent-parity-prefix-converse}]
\emph{Outline.}  We fix a word \(z\) of length \(n\) and construct \(n+1\)
distinct messages \(x^{(g)}\) such that deleting one length-\(t\) block from
\(W_{n,t}(x^{(g)})\) leaves \(z\).  The received words
\((z,H(x^{(g)}))\) must then be pairwise distinct, so the values
\(H(x^{(g)})\) are pairwise distinct.  We then exhibit a set of
\(\lfloor n/t\rfloor+1\) values that none of these final symbols may take.
Together, these facts require \(q\ge(n+1)+\lfloor n/t\rfloor+1\).

All sums in this proof are XORs.  For a word \(u\), let
\[
 S_r(u)=\bigoplus_{\substack{i\\i\equiv r\pmod t}}u_i
 \qquad(r\in\mathbb Z_t).
\]
The subsequence of a word formed by the positions in one residue class
modulo \(t\) is called a \emph{row}.  A word \(v\) is \emph{\(t\)-periodic} if
\(v_i=v_{i+t}\) whenever both positions exist.  A \emph{child} of a word is
the result of deleting one block of \(t\) adjacent symbols, and the word is
then a \emph{parent} of the child.  Since \(C_j(x)\) is the XOR of the
information symbols in the residue class of position \(n+j\), a word of
length \(n+t\) equals \(W_{n,t}(x)\) for some \(x\) (namely, its first \(n\)
symbols) if and only if each of its \(t\) rows has XOR zero.

\smallskip
\noindent\emph{Step 1: many parents with a common child.}
Fix a word \(z\in\F_q^n\) and write \(s_r=S_r(z)\); indices of \(s\) are
read modulo \(t\).  At each gap \(g\in\{0,\ldots,n\}\), that is, after the
first \(g\) symbols of \(z\), insert the block
\[
 I_g=(s_{g+1},s_{g+2},\ldots,s_{g+t}),
\]
and call the resulting length-\((n+t)\) word
\[
 P_g=(z_1,\ldots,z_g,\,s_{g+1},\ldots,s_{g+t},\,z_{g+1},\ldots,z_n);
\]
thus \(z_i\)
occupies position \(i\) of \(P_g\) for \(i\le g\) and position \(i+t\) for
\(i>g\), while \(I_g\) occupies positions \(g+1,\ldots,g+t\).  The inserted block
contributes \(s_r\) to row \(r\), so every row of \(P_g\) has XOR zero.
Hence \(P_g=W_{n,t}(x^{(g)})\) for the unique message \(x^{(g)}\) formed by
its first \(n\) symbols.  Deleting \(I_g\) from the codeword
\((P_g,H(x^{(g)}))\) gives \((z,H(x^{(g)}))\).  Exact correction therefore
forces \(H\) to take distinct values on distinct parents \(P_g\).

Moving \(I_g\) across \(z_g\) shows that
\[
 P_{g-1}=P_g\quad\Longleftrightarrow\quad z_g=s_g.
\]
More generally, \(P_g=P_h\) for \(g<h\) precisely when every coordinate
\(z_i\) with \(g<i\le h\) equals \(s_i\).  (If \(P_g=P_h\), compare the two
words from position \(g+1\) on: since \(s\) has period \(t\), this gives
\(z_{g+1}=s_{g+1}\), then \(z_{g+2}=s_{g+2}\), and so on up to \(z_h=s_h\).
The converse follows by chaining the case \(h=g+1\).)  Therefore
\begin{equation}
 \bigl|\{P_g:0\le g\le n\}\bigr|
 =1+\bigl|\{i:z_i\ne s_i\}\bigr|.
 \label{eq:adjacent-parent-count}
\end{equation}

\smallskip
\noindent\emph{Step 2: \(q>n\).}
This auxiliary bound is needed in Step~3.
If \(q=2\), place two ones in one row of \(z\) and zeros elsewhere.  That
row has sum zero, so \eqref{eq:adjacent-parent-count} gives three
parents, more than the two available values of \(H\).  If \(4\le q\le n\),
which covers every remaining case since \(q\) is a power of two,
choose distinct nonzero \(u,v\) and put \(w=u\oplus v\).  In every row of
even length \(\ell\), use only \(u\); in every row of odd length \(\ell\),
use \(\ell-2\) copies of \(u\), followed by \(v,w\).  (Every row has length
at least \(\lfloor n/t\rfloor\ge2\), because \(n\ge2t\).)  Each row has sum
zero and contains no zero, so all \(n+1>q\) parents are distinct.  Both
cases contradict the injectivity of \(H\) on parents.

\smallskip
\noindent\emph{Step 3: a child with \(n+1\) distinct parents.}
We use one elementary fact about a single row.

\emph{Row fact.}  If \(\ell\ge2\) and \(q>2\ell-2\), there are \(\ell\)
pairwise distinct elements of \(\F_q\) whose XOR is not among them.
Indeed, first choose an \((\ell-1)\)-set \(Y\) with nonzero XOR \(\xi\).
Such a set exists: otherwise, for a fixed \((\ell-2)\)-set \(T\) and two
distinct elements \(u,v\notin T\), the sets \(T\cup\{u\}\) and
\(T\cup\{v\}\) would both have XOR zero, which forces \(u=v\).  Since
\(|Y\cup(\xi+Y)|\le2\ell-2<q\), choose \(u\notin Y\cup(\xi+Y)\).  The set
\(Y\cup\{u\}\) has XOR \(\xi\oplus u\), which is neither \(u\) (as
\(\xi\ne0\)) nor an element of \(Y\) (as \(u\notin\xi+Y\)).

Every row has length at most \(\lceil n/t\rceil\), and
\(2\lceil n/t\rceil-2\le n<q\) because \(t\ge2\) and by Step~2.  Apply the
row fact to each of the \(t\) rows and interleave the resulting rows into
\(z\).  Then:
\begin{enumerate}[label=(\roman*),leftmargin=*]
\item in every row of \(z\), the entries are pairwise distinct;
\item in every row of \(z\), the XOR of the row is not among its entries.
\end{enumerate}
By (ii), \(z_i\ne s_i\) for every \(i\), so all \(n+1\) parents \(P_g\) are
distinct by \eqref{eq:adjacent-parent-count}.  From now on \(z\) is fixed
with properties (i) and (ii).

\smallskip
\noindent\emph{Step 4: a value forbidden by an internal deletion.}
Let \(P=W_{n,t}(x)\), delete a block of length \(t\) beginning at position
\(d\), where \(1\le d\le n+1\), and call the resulting length-\(n\) word
\(y\).  Let \(c\) be the deleted symbol in residue class \(n+1\pmod t\),
the class of \(C_1\).  Since every row of \(P\) has XOR zero, \(C_1(y)=c\).
Moreover, \(y\) coincides with the word obtained by deleting the last \(t\)
symbols of \(P\), that is, \(y=x\), exactly when the part of \(P\) from
position \(d\) on is \(t\)-periodic:
\begin{equation}
 y=x\quad\Longleftrightarrow\quad
 P_i=P_{i+t}\ \text{ for every }\ d\le i\le n.
 \label{eq:adjacent-general-self-deletion}
\end{equation}
Suppose \(y\ne x\) and \(H(x)=c\).  Deleting the block from \((P,H(x))\)
gives \((y,c)\).  Deleting the terminal block
\((C_2(y),\ldots,C_t(y),H(y))\) from the codeword of \(y\) also gives
\((y,C_1(y))=(y,c)\).  Exact correction therefore forbids the value \(c\)
for \(H(x)\) whenever the part of \(P\) from position \(d\) on is not
\(t\)-periodic.

\smallskip
\noindent\emph{Step 5: a common forbidden set.}
The row of \(z\) in residue class \(n+1\pmod t\) has
\(k=\lfloor n/t\rfloor\) entries; write them, from left to right, as
\(\alpha_1,\ldots,\alpha_k\), with XOR \(a\).  The position of \(\alpha_k\)
in \(z\) is \(p=n+1-t\).  Let \(b\) be the XOR of the row in class
\(n\pmod t\).  In the notation of Step~1, \(a=s_{n+1}\) and \(b=s_n\); thus
\(a\) and \(b\) are the entries of \(I_g\) in the classes \(n+1\) and
\(n\pmod t\), for every \(g\).  By (i), the last two entries \(z_{n-t},z_n\) of the row in class \(n\pmod t\)
are distinct, and by (ii), \(b\) is not an entry of that row.  Set
\begin{equation}
 \mathcal S=\{\alpha_1,\ldots,\alpha_k,a\},
 \qquad |\mathcal S|=k+1,
 \label{eq:adjacent-common-forbidden-set}
\end{equation}
where the elements are distinct by (i) and (ii).  We claim that every
value in \(\mathcal S\) is forbidden for \(H(x^{(g)})\), for every
\(g\).  In each case we apply Step~4 with \(x=x^{(g)}\) and \(P=P_g\), and
we exhibit a start \(d\) for which the deleted symbol in class
\(n+1\pmod t\) is the value in question and the part of \(P_g\) from
position \(d\) on is not \(t\)-periodic.

\emph{The values \(\alpha_j\) with \(j<k\).}  Begin the deletion at the
position of \(\alpha_j\) in \(P_g\); this position is at most \(n\), and the
deleted symbol in class \(n+1\pmod t\) is \(\alpha_j\).  The part of \(z\)
starting at \(\alpha_j\) is not \(t\)-periodic, because \(t\) positions
later it contains \(\alpha_{j+1}\ne\alpha_j\).  If \(I_g\) lies before
\(\alpha_j\), the part of \(P_g\) starting at \(\alpha_j\) equals the part of
\(z\) starting at \(\alpha_j\).  If \(I_g\) lies after \(\alpha_j\), deleting
\(I_g\) from the part of \(P_g\) gives the part of \(z\); since deleting a
block of length \(t\) from a \(t\)-periodic word leaves a \(t\)-periodic
word, the part of \(P_g\) is not \(t\)-periodic either.  Step~4 forbids
\(\alpha_j\).

\emph{The value \(\alpha_k\).}  If \(g\ge p\), use the block beginning at
\(p\): positions \(p\) and \(p+t=n+1\) of \(P_g\) contain \(\alpha_k\) and
\(a\), respectively (position \(n+1\) lies in \(I_g\), since \(g+1\le n+1=p+t\le g+t\)), and \(a\ne\alpha_k\) by (ii).  If \(g<p\), then
\(\alpha_k\) sits at position \(n+1\) of \(P_g\); use the block beginning at
\(n\).  When \(g\le n-t-1\), positions \(n\) and \(n+t\) of \(P_g\) contain
\(z_{n-t}\ne z_n\); when \(g=n-t\), they contain \(b\ne z_n\) (position \(n=g+t\) lies in \(I_g\)).  Since
\(p=n-t+1\), these cases cover every \(g\).  In each case the relevant part
is not \(t\)-periodic, and \(\alpha_k\) is forbidden.

\emph{The value \(a\).}  If \(g<n\), delete \(I_g\) itself.  The result is
\(z\), and \(z\ne x^{(g)}\), since otherwise \(P_g=P_n=W_{n,t}(z)\),
contradicting distinctness of the parents; the deleted symbol in class
\(n+1\pmod t\) is \(a\).  At \(g=n\), instead delete the block beginning at
\(n\): positions \(n\) and \(n+t\) contain \(z_n\ne b\) (position \(n+t=g+t\) lies in \(I_g\)).  Hence \(a\) is
also forbidden for every parent.

\smallskip
\noindent\emph{Conclusion.}  The \(n+1\) values \(H(x^{(g)})\) are distinct
and all lie outside the same \((k+1)\)-set \(\mathcal S\).  Consequently
\[
 q-k-1\ge n+1,
\]
which is the asserted bound \(q\ge n+\lfloor n/t\rfloor+2\).
\end{proof}

For \(t=2\) the bound reads \(q\ge n+\lfloor n/2\rfloor+2\), linear in
\(n\), whereas the construction of Theorem~\ref{thm:adjacent-closed-form}
has \(q_{\mathrm{adj}}<(n+2)^4\).  How much of the remaining gap is inherent
in the architecture is Open problem~3 in Section~\ref{sec:conclusion};
Section~\ref{sec:adjacent-n4} settles the first case \(n=4\) for bursts of
length at most two. The range \(n<2t\), and the threshold at \(n=4\) for exact adjacent pairs alone, are not settled here.

\section{The at-most-two threshold at \texorpdfstring{\(n=4\)}{n=4}}
\label{sec:adjacent-n4}

Theorem~\ref{thm:adjacent-parity-prefix-converse} applies to the parity
prefix of Section~\ref{sec:adjacent} from \(n\ge2t\) on.  Its first case
with \(t=2\) is \(n=4\), where it gives \(q\ge8\), whereas the construction
of Theorem~\ref{thm:adjacent-closed-form} uses \(q_{\mathrm{adj}}=512\) at
\(n=4\): there \(M=6\), \(\kappa_4=4\), and
\(K=6\cdot3\cdot3\cdot5=270\).  This section determines the exact
threshold for that architecture at \(n=4\) when bursts of length zero,
one, or two must be corrected.  With \(W\) the parity prefix
of Section~\ref{sec:adjacent} for \(n=4\), the theorem below finds the
smallest power-of-two alphabet for which some choice of the final symbol
\(H\) works, namely \(16\); so, for this channel, neither the bound \(8\)
nor the alphabet \(512\) of the construction is tight there.  The proof has three parts:
Theorem~\ref{thm:adjacent-parity-prefix-converse} excludes \(q\le4\); an
explicit eleven-message configuration excludes \(q=8\); and a
computer-verified assignment for \(q=16\) is transported to every larger
field.

\begin{theorem}[Power-of-two threshold for bursts of length at most two]
\label{thm:adjacent-n4-exact-threshold}
Let \(q=2^m\) with \(m\ge1\), and put
\[
 W(x_1,x_2,x_3,x_4)
 =(x_1,x_2,x_3,x_4,x_1\oplus x_3,x_2\oplus x_4).
\]
There is a function \(H:\F_q^4\to\F_q\) for which
\(\{(W(x),H(x)):x\in\F_q^4\}\) corrects every adjacent burst of
length zero, one, or two if and only if \(q\ge16\). For every
\(q=2^m\ge16\), one such function is obtained from a single fixed
assignment for \(q=16\) by a change of coordinates: an \(\F_2\)-linear
bijection from a four-dimensional \(\F_2\)-subspace of \(\F_q\) that contains
the symbols of the message onto \(\F_2^4\), depending only on the
\(\F_2\)-span of these symbols, is applied to the message, and its inverse to
the value of the assignment.
\end{theorem}

For each \(m\), realize \(\F_q=\F_{2^m}\) canonically as in
Section~\ref{sec:model} and identify its polynomial basis with
\(\F_2^m\).  Only addition, which is bitwise XOR, and the binary labels are
used, so the choice of polynomial does not affect the statements below.
For \(y=(y_1,y_2,y_3,y_4)\in(\F_2^m)^4\) put
\[
 A(y)=y_1\oplus y_3,\qquad B(y)=y_2\oplus y_4.
\]
These are the parity symbols of Section~\ref{sec:adjacent} for \(n=4\).
For \(x\in(\F_2^m)^4\), write \(G_m\) for the graph \(\calG_{4,2^m}\) of
Section~\ref{subsec:final-symbol} (thus \(G_3\) and \(G_4\) belong to
\(q=8\) and \(q=16\)): it joins distinct messages whose prefixes \(W(x)\)
have a common child after one deletion or one adjacent-pair deletion, a
child of a word being one of its burst descendants.  Let \(\mathcal L_x\)
be the exact list \eqref{eq:adjacent-exact-list} for \(n=4\): the values
\(B(y)\), respectively \(A(y)\), of the messages \(y\ne x\) that one
deletion, respectively one adjacent-pair deletion, of \(W(x)\) can leave in
the form \((y,A(y))\), respectively \(y\).  By
Lemma~\ref{lem:adjacent-collision-characterization}, a function \(H\)
works if and only if
\begin{equation}
 H(x)\notin\mathcal L_x\quad\text{for every }x,\qquad
 H(x)\ne H(y)\quad\text{for every }xy\in E(G_m).
 \label{eq:adjacent-n4-interface}
\end{equation}
The transport of an assignment from \(q=16\) to larger fields rests on the
following closure property of these constraints.

\begin{lemma}[Span closure]
\label{lem:adjacent-n4-span-closure}
Put \(\mathsf V_x=\langle x_1,x_2,x_3,x_4\rangle_{\F_2}\).
If \(xy\in E(G_m)\), then \(\mathsf V_x=\mathsf V_y\), and
\(\mathcal L_x\subseteq\mathsf V_x\).
\end{lemma}

\begin{proof}
Every symbol of \(W(x)\), and hence of each of its children, lies in
\(\mathsf V_x\).  Let \(xy\in E(G_m)\), so that
\(c:=P_{r,d}(x)=P_{r,d'}(y)\) for some \(r\in\{1,2\}\) and admissible
starts \(d,d'\).  By Lemma~\ref{lem:adjacent-parity-inversion}, applied to
\(y\) and \(d'\), every symbol of \(W(y)\) that is not a symbol of \(c\) is
an XOR of symbols of \(c\), and the symbols of \(c\) are symbols of
\(W(x)\).  Hence all symbols of \(W(y)\), in particular
\(y_1,\ldots,y_4\), lie in \(\mathsf V_x\); thus
\(\mathsf V_y\subseteq\mathsf V_x\), and symmetry gives equality.  If
\(B(y)\in\mathcal L_x\) through \(\mathsf B_{1,d}(W(x))=(y,A(y))\), then
\((y,A(y))=P_{1,6}(y)\), and Lemma~\ref{lem:adjacent-parity-inversion}
shows that \(B(y)\) is an XOR of symbols of \((y,A(y))\), a child of
\(W(x)\).  If \(A(y)\in\mathcal L_x\) through \(\mathsf B_{2,d}(W(x))=y\),
then \(y=P_{2,5}(y)\), and the lemma shows that \(A(y)\) is an XOR of
symbols of \(y\), which are symbols of \(W(x)\).  Hence
\(\mathcal L_x\subseteq\mathsf V_x\).
\end{proof}

\begin{proof}[Proof of
Theorem~\ref{thm:adjacent-n4-exact-threshold}]
Theorem~\ref{thm:adjacent-parity-prefix-converse} with \(t=2\), \(n=4\)
applies, since a code correcting every burst of length at most two corrects
exact adjacent pairs, and gives \(q\ge4+\lfloor4/2\rfloor+2=8\), which
excludes \(q\in\{2,4\}\).  At \(q=8\), the following eleven-message
subinstance of \eqref{eq:adjacent-n4-interface} is already impossible
(each message is a \(4\)-tuple of elements of \(\F_8\), each written by its
binary label in \(\{0,\ldots,7\}\); to avoid a clash with other notation, the messages are named by sans-serif letters):
\[
\begin{gathered}
\mathsf{a}=(4,2,5,3),\ \mathsf{b}=(4,3,1,1),\ \mathsf{c}=(4,3,1,2),\\
\mathsf{d}=(4,3,5,1),\ \mathsf{e}=(4,3,5,2),\ \mathsf{f}=(5,1,4,3),\\
\mathsf{g}=(5,2,4,3),\ \mathsf{h}=(5,3,1,1),\ \mathsf{i}=(5,3,1,2),\\
\mathsf{j}=(5,3,4,1),\ \mathsf{k}=(5,3,4,2).
\end{gathered}
\]
Write \(\mathcal F_x=\F_8\setminus\mathcal L_x\) for the set of values allowed
for \(H(x)\), and put \(\mathcal P=\{0,2,3,6,7\}\).  For all eleven messages,
\(\mathcal F_x\subseteq\mathcal P\), that is, \(\{1,4,5\}\subseteq\mathcal L_x\).  The graph contains
\[
\begin{gathered}
G_3[\{\mathsf{a},\mathsf{b},\mathsf{c},\mathsf{e},\mathsf{g}\}]=G_3[\{\mathsf{a},\mathsf{g},\mathsf{h},\mathsf{i},\mathsf{k}\}]=K_5,\\
G_3[\{\mathsf{b},\mathsf{c},\mathsf{d},\mathsf{e}\}]=G_3[\{\mathsf{h},\mathsf{i},\mathsf{j},\mathsf{k}\}]
=G_3[\{\mathsf{b},\mathsf{c},\mathsf{d},\mathsf{f}\}]=G_3[\{\mathsf{f},\mathsf{h},\mathsf{i},\mathsf{j}\}]=K_4,
\qquad \mathsf{dj},\mathsf{ek}\in E(G_3).
\end{gathered}
\]
These facts can be checked by hand: for each message, compute \(W(x)\)
and its six one-deletion and five adjacent-pair children.
Appendix~\ref{app:n4-witnesses} lists, for each of the \(33\) edges above,
a common child, and for each message the deletions that put \(1\), \(4\),
and \(5\) into \(\mathcal L_x\).  For the contradiction, put
\(\mathcal T=\{H(\mathsf{a}),H(\mathsf{g})\}\) and \(\mathcal Q=\mathcal P\setminus\mathcal T\).
The two \(K_5\)'s force each of \((\mathsf{b},\mathsf{c},\mathsf{e})\) and \((\mathsf{h},\mathsf{i},\mathsf{k})\) to use
the three colors of \(\mathcal Q\).  The first two displayed \(K_4\)'s and the
edge \(\mathsf{dj}\) make \(H(\mathsf{d}),H(\mathsf{j})\) the two distinct colors of \(\mathcal T\).
The remaining \(K_4\)'s make \(\mathsf{f}\) adjacent to both of them, so
\(H(\mathsf{f})\in\mathcal Q\).  In the \(K_4\) on \(\{\mathsf b,\mathsf c,\mathsf d,\mathsf f\}\),
\(H(\mathsf f)\) differs from \(H(\mathsf b)\) and \(H(\mathsf c)\); since
\(\mathsf b,\mathsf c,\mathsf e\) use the three colors of \(\mathcal Q\) and
\(H(\mathsf f)\in\mathcal Q\), this forces \(H(\mathsf f)=H(\mathsf e)\).  In
the same way, the \(K_4\) on \(\{\mathsf f,\mathsf h,\mathsf i,\mathsf j\}\),
together with the colors of \(\mathsf h,\mathsf i,\mathsf k\), forces
\(H(\mathsf f)=H(\mathsf k)\).  Hence \(H(\mathsf e)=H(\mathsf k)\),
contradicting the edge \(\mathsf{ek}\).  Hence \(q=8\) is impossible.

For \(q=16\), an explicit assignment
\[
 h_*:(\F_2^4)^4\longrightarrow\F_2^4
\]
was checked in two ways: against both constraints in
\eqref{eq:adjacent-n4-interface}, that is, against all \(457{,}035\) edges of
\(G_4\) and every exact list \(\mathcal L_x\); and directly against the
definition of correction, by verifying for each burst length zero, one, and
two that no received word arises from two codewords, which covers all
\(875{,}115\) pairs formed by a codeword and one of its distinct received
words.  This is the only computer-assisted existence step; the assignment and
two independent verifiers are archived in the verification companion
\cite{GrishinZabokritskiy2026Verification}.

It remains to transport \(h_*\) to every \(q=2^m\ge16\).  The idea is that
all constraints on \(H(x)\) live inside the \(\F_2\)-span of the symbols of
\(x\), which has dimension at most four, so a change of coordinates maps them to
constraints over \(\F_2^4\), where \(h_*\) satisfies them.  Only XOR is
involved, so an \(\F_2\)-linear change of coordinates is legitimate.

For \(m\ge4\), let \(\mathcal E_x\) be the ordered list of four vectors
obtained greedily as follows: scan the elements of \(\mathsf V_x\) in
increasing binary-label order and keep each one that is independent of
those already kept; then scan the standard basis vectors
\(e_1,e_2,\ldots\) of \(\F_2^m\) in order and keep each one not in the
current span, until four vectors are kept.  The list depends only on
\(\mathsf V_x\).  The second scan reaches four vectors: the span of
\(e_1,\ldots,e_4\) has dimension four, whereas \(\dim\mathsf V_x\le4\).
Four vectors are needed so that the map \(\varphi_x\) below is a bijection
onto \(\F_2^4\) and \(\varphi_x^{-1}\) is defined on every value of \(h_*\).  Let
\(\widetilde{\mathsf V}_x=\langle\mathcal E_x\rangle\)
and let \(\varphi_x:\widetilde{\mathsf V}_x\to\F_2^4\) send the ordered frame
to the standard basis.  For \(x=(x_1,\ldots,x_4)\), write
\(\varphi_x[x]=(\varphi_x(x_1),\ldots,\varphi_x(x_4))\).  Define
\begin{equation}
 \widehat H_m(x)=
 \varphi_x^{-1}\!\bigl(
 h_*(\varphi_x[x])\bigr).
 \label{eq:adjacent-n4-span-transport}
\end{equation}

If \(xy\) is an edge, Lemma~\ref{lem:adjacent-n4-span-closure} gives
\(\mathsf V_x=\mathsf V_y\), so both vertices use the same list
\(\mathcal E_x\) and the same map \(\varphi_x\).  Messages with different
spans are never adjacent, so they impose no joint constraint.  The map
\(\varphi_x\) is an \(\F_2\)-linear bijection applied coordinatewise, so it
commutes with forming \(W\) and with deletions.  By
Lemma~\ref{lem:adjacent-n4-span-closure}, every neighbor of \(x\) in \(G_m\)
and every element of \(\mathcal L_x\) lies in
\(\widetilde{\mathsf V}_x\), the domain of \(\varphi_x\).  Applied in
\(G_4\) (the graph for \(q=16\)), the same lemma shows that every neighbor
of \(\varphi_x[x]\) and every element of \(\mathcal L_{\varphi_x[x]}\) lies in
the span of the symbols of \(\varphi_x[x]\), which is
\(\varphi_x(\mathsf V_x)\), inside the image of \(\varphi_x\).  Hence
\(\varphi_x\) maps the neighborhood of \(x\) in \(G_m\) bijectively onto the
neighborhood of \(\varphi_x[x]\) in \(G_4\), and
\[
 \varphi_x(\mathcal L_x)=\mathcal L_{\varphi_x[x]}.
\]
The validity of \(h_*\) therefore implies both constraints in
\eqref{eq:adjacent-n4-interface} for \(\widehat H_m\).  This proves
sufficiency for every \(m\ge4\), while the preceding lower-bound and
\(q=8\) arguments prove necessity.
\end{proof}

The confusion-graph formulation follows the standard graph model
\cite{CullinaKulkarniKiyavash2012}.  For the prescribed construction with one
XOR parity per residue class modulo \(t\) and one final symbol, the preceding
analysis gives the lower bound of
Theorem~\ref{thm:adjacent-parity-prefix-converse} and the exact threshold
at \(n=4\).  Whether \(q=O(n)\) suffices for this architecture at general
\(n\) is part of Open problem~3 in Section~\ref{sec:conclusion}.

\section{Alphabet size and evaluation cost}
\label{sec:comparison}

For \(n\ge2\), both constructions use the optimal number of redundant
symbols among all codes of size \(q^n\); for adjacent deletions at \(n=1\), two
redundant symbols suffice (Remark~\ref{rem:adjacent-converse-endpoints}).
The alphabet sizes, however, depend on the deletion pattern and on how the
color \(\ell_1\) may be computed.  This section shows that a color that is
cheap to store but expensive to compute gives a smaller alphabet, and then
compares all alphabet bounds.

\subsection{A canonical global greedy color}

Greedy coloring of the entire confusion graph gives a code that is
explicitly defined, but for which no polynomial-time evaluation is given:
the color is defined by a global lexicographic pass through the confusion
graph, so evaluating it at a single vertex may take exponential time.  We claim no
lower bound on this cost.  The
code serves as an existence benchmark for the arbitrary-deletion
construction, whereas the color of Theorem~\ref{thm:main} is individually
computable in polynomial time.  The three payloads retain their layout:
the \(S\)-head and \(V\)-tail keep their lengths, while the \(S\)-tail, the
\(V\)-head and the color are shortened.

\begin{proposition}[Global-greedy comparison]
\label{prop:global-greedy}
For every \(n\ge2\), defining the color \(\ell\) of
Proposition~\ref{prop:proper-color-interface} by a canonical greedy
coloring of the entire finite confusion graph \(G_\star\), formed with suitably
chosen bit counts, yields a systematic code
\(E:\F_q^n\to\F_q^{n+3}\) that corrects up to two arbitrary deletions over a
power-of-two alphabet
\(q<2^{20}n^4\).
\end{proposition}

\begin{proof}[Proof of Proposition~\ref{prop:global-greedy}]
Retain \(d_{\mathrm{tag}}=2\) and \(\delta=4\).  For this comparison set
\[
 h_S=u+4,\qquad h_V=3u+8,\qquad
 b_{\mathrm{col}}=3u+4,\qquad s=4u+16,\qquad \rho=4u+12,
\]
and take \(\pi_S,\pi_V\) to be the terminal-coordinate projections with
these kernel dimensions.
The three payload identities of Proposition~\ref{prop:proper-color-interface}
reduce to
\[
\begin{aligned}
 (u+4)+(3u+8)&=(3u+12)+u\\
               &=(u+8)+(3u+4)=\rho.
\end{aligned}
\]
Let \(G_\star\) be the union graph of Section~\ref{sec:graph} formed with
these values of \(h_S\) and \(h_V\).  Order its vertices lexicographically
by their public binary coordinate words, and greedily assign to each vertex
the least color absent from its previously colored neighbors.  This is a
canonical proper coloring with at most \(D_n+1\) colors, where \(D_n\) is
the degree bound in Equation~\eqref{eq:union-degree-bound}: the \(G_H\)
degree calculation depends only on the \(S\)-head dimension \(u+4\), so
the same \(D_n\) applies.  Evaluating one color may require traversing
exponentially many earlier vertices.  For \(u\ge2\),
\[
 D_n+1<16N^3=2^{3u+4}=2^{b_{\mathrm{col}}}.
\]
For \(u=1\), necessarily \(n=2\), and
\(D_2+1=73<128=2^{b_{\mathrm{col}}}\).  Thus the color fits in the third
payload, Proposition~\ref{prop:proper-color-interface} applies with these
bit counts and this coloring, and
\[
 q=2^{4u+16}=2^{16}N^4<2^{20}n^4.
\]
Equation~\eqref{eq:decoder-candidate-count} then bounds candidate enumeration
by \(O(n^{10})\) candidates in total, over the at most three values of
\(a\).
The decoder re-encodes each candidate and so evaluates the color, which
remains the global step described above; it is the only step of encoding or
decoding not shown to run in polynomial time, so this code is not shown to
admit polynomial-time encoding and decoding under the convention of
Section~\ref{sec:model}.
\end{proof}

\subsection{Comparison points and scope}

The four alphabet bounds separate different access models and deletion
patterns.  Theorem~\ref{thm:main}
gives polynomial-time encoding and decoding for two arbitrary deletions at
\(q<2^{32}n^7\); Proposition~\ref{prop:global-greedy} lowers this to
\(q<2^{20}n^4\) with global evaluation; the Haxell benchmark in
Section~\ref{sec:model} proves existence at \(q=O(n^4)\); and
Theorem~\ref{thm:adjacent-closed-form} gives a construction from two XOR
parities and Tenengolts checksums, with polynomial-time encoding and decoding, at
\(q_{\mathrm{adj}}<(n+2)^4\) for one adjacent burst of length at most two.

The two adjacent converses also have different scopes.  Every code of size
\(q^n\) and length \(n+3\) that corrects one burst of exactly two adjacent
deletions has \(q\ge n-1\) for \(n\ge3\), whereas, for \(q=2^m\), a final
symbol following the information prefix and the prescribed XOR parities,
one for each residue class of positions modulo \(t\), forces
\(q\ge n+\lfloor n/t\rfloor+2\) if the code
corrects one exact adjacent burst of length \(t\), where
\(t\ge2,n\ge2t\).  The latter bound applies to that prescribed architecture;
for \(t=2,n=4\),
correction of bursts of length zero, one, or two in the same architecture
has exact power-of-two threshold \(16\).

Measured in bits, the three redundant symbols of our polynomial-time
arbitrary-deletion construction use
at most \(21\log_2 n+96\) bits, while the adjacent construction uses
\(12\log_2n+O(1)\) bits.  These are the bit costs of the present
constructions.

For comparison with fixed-alphabet two-deletion codes, Ye, Sun, Yu, Ge,
and Elishco obtain \(5\log_2 L+O(\log_2\log_2 L)\) redundant bits
for fixed \(q>2\), by a separate construction for the case \(b=1\) of their work on two
bursts of \(b\) deletions~\cite{YeSunYuGeElishco2026}, where \(L\)
denotes the codeword length.  Sun and Ge's partitioning method yields
an existential binary two-deletion code with \(6\log_2 L+O(1)\)
redundant bits; their explicit encoder with protected checks treats
the related channel with one deletion and one substitution
\cite{SunGe2026Partitioning}.  En Gad proves the existence of binary
two-deletion codes with redundancy
\(3\log_2 L+O(\log_2\log_2 L)\) by random linear hashing and a
conflict-graph two-coloring~\cite{EnGad2026}; that result is
nonconstructive.

For growing alphabets and a fixed deletion count \(k\), the construction of Sima,
Gabrys, and Bruck has a redundancy guarantee of at most
\((30k+1)\log_2 q\) bits for \(n\) information symbols and \(q\ge n\),
with polynomial-time encoding and decoding
\cite{SimaGabrysBruck2020Qary}.  It is not systematic.
For codeword length \(L\) and \(q>L^{2+\varepsilon}\) with a fixed
\(\varepsilon>0\), Schaller, Toesca, and Vu give a redundancy bound
of \(k\log_2 q+(3k-1)\log_2 L+O(1)\) bits, or
\(k+(3k-1)/c+o(1)\) symbol widths for \(q=L^c\)
\cite{SchallerToescaVu2025}.  Their codes are defined on words with
pairwise distinct symbols and are not systematic; no efficient encoder
is supplied.  The two-deletion comparison below uses a common codeword
length.

For any alphabet of size \(q\), storing three redundant symbols costs
\(3\log_2q\) bits.  For two arbitrary deletions over alphabets
\(q=n^{c}\) with \(c>2\), the codes of Schaller, Toesca, and
Vu~\cite{SchallerToescaVu2025} use \(2\log_2 q+5\log_2 n+O(1)\) redundant
bits, counted against all \(q\)-ary words of the same length.  We compare
codes of block length \(L=n+3\), using the bit redundancy
\(L\log_2q-\log_2|\mathcal C|\).  Since
\(\log_2(n+3)=\log_2n+O(1/n)\), aligning the lengths in this way leaves
the displayed asymptotic formulas unchanged.  For
large \(n\), this bound is smaller than the three-symbol cost when \(c>5\).
At the alphabet scale \(q=\Theta(n^7)\) chosen in
Theorem~\ref{thm:main}, these codes use \(19\log_2n+O(1)\) redundant bits
against \(21\log_2n+O(1)\) for three symbols; at the scale
\(q=\Theta(n^4)\) chosen in Proposition~\ref{prop:global-greedy}, they use
\(13\log_2n+O(1)\) against \(12\log_2n+O(1)\).  The comparison concerns
redundant bits
only: these codes are defined as sets of words with pairwise distinct
symbols and are not systematic, and no efficient encoder for them is
given.

The smaller bit redundancy of these codes for \(c>5\) does not contradict
Theorem~\ref{thm:lower-bound}, which counts full information symbols.  At length
\(n+3\), a code cardinality strictly between \(q^n\) and \(q^{n+1}\)
would give bit redundancy strictly between two and three symbol widths,
while allowing only \(n\) full \(q\)-ary information symbols.

Codes over a fixed alphabet, and \(q\)-ary codes
measured in redundant bits
\cite{SimaGabrysBruck2020Systematic,SimaGabrysBruck2020Qary,SongCai2023,
LiuTjuawinataXing2024}, optimize a different quantity: they minimize the
number of redundant bits, and in general they need not transmit the
information symbols unchanged.  The systematic coloring framework of Li,
Gabrys, and Farnoud assumes that the color (\(\Phi(x)\) in their notation)
is received
error-free \cite[Definition~9 and Theorems~2--3]{LiGabrysFarnoud2025}, and
notes that protecting a short color against deletions costs
\(O(\log\log n)\) further bits in its bit-redundancy regime
\cite[Remark~1]{LiGabrysFarnoud2025}; in the present channel the three
symbols carrying the checks are themselves part of the transmitted word and
receive no separate protection.

Deletion codes for multisets retain
symbol multiplicities but discard their order
\cite{KovacevicTan2018,KreindelEssayagZabokritskiy2026Cyclic,
KreindelEssayagZabokritskiy2026Geometry}.  Since order is discarded in the
multiset model, reconstructing the order of the symbols, which is central
here, does not arise there.

For two arbitrary deletions and systematic codes with three redundant
symbols, if polynomial-time encoding and decoding are not required,
Proposition~\ref{prop:appended-alphabet-lower} and
Theorem~\ref{thm:haxell-systematic} together place the optimal alphabet
exponent between \(2\) and \(4\).  Our construction with polynomial-time
encoding and decoding has exponent \(7\).  The code of Proposition~\ref{prop:global-greedy} has exponent
\(4\), the exponent of the Haxell benchmark; unlike that benchmark, which
chooses the redundant symbols of all messages jointly, it has the structure
of our polynomial-time construction, and only its color is evaluated globally.
In both constructions \(s=u+b_{\mathrm{col}}+12\), so replacing the
\(6u+13\)-bit color by a \(3u+4\)-bit color lowers the alphabet exponent
by three.
The best exponent achievable
with polynomial-time encoding and decoding under the computational
convention lies between \(2\) and \(7\) and is open.

Table~\ref{tab:results-overview} collects the constructions and bounds.
Comparisons between alphabet bounds must account for the channel, code
class, and computational guarantees in each row.

\begin{table}[t]
\caption{Overview of the results.  Here ``systematic'' means the strict
prefix form of Section~\ref{sec:model}, and ``converse'' rows are lower bounds.}
\label{tab:results-overview}
\small
\renewcommand{\arraystretch}{1.15}
\begin{tabularx}{\textwidth}{@{}>{\raggedright\arraybackslash}p{0.3\textwidth}
  >{\raggedright\arraybackslash}X
  >{\raggedright\arraybackslash}p{0.12\textwidth}@{}}
\toprule
code class & statement & section \\
\midrule
\multicolumn{3}{@{}l}{\emph{Two arbitrary deletions}}\\[1pt]
any code of size \(q^n\)
  & at least three redundant symbols (known; \(n\ge2\))
  & Sec.~\ref{sec:model} \\
linear codes
  & no \([n+3,n]_q\) code corrects two deletions (known; \(n\ge3\))
  & Sec.~\ref{sec:model} \\
systematic, three redundant symbols
  & converse: \(q\ge\frac14\binom n2\) (\(n\ge2\))
  & Sec.~\ref{sec:model} \\
systematic, three redundant symbols
  & exists for \(q=O(n^4)\); no efficient encoder supplied
  & Sec.~\ref{sec:model} \\
systematic, three redundant symbols
  & \(q<2^{20}n^4\); an existence benchmark: explicitly defined, but no
    polynomial-time evaluation is given (computing the global greedy
    coloring may take exponential time)
  & Sec.~\ref{sec:comparison} \\
systematic, three redundant symbols
  & \(q<2^{32}n^7\); polynomial-time encoding and decoding
  & Thm.~\ref{thm:main}; Sec.~\ref{sec:construction} \\
\addlinespace
\multicolumn{3}{@{}l}{\emph{Adjacent deletions}}\\[1pt]
at most two adjacent deletions; systematic, three redundant symbols
  & \(q_{\mathrm{adj}}<(n+2)^4\); parity symbols and checksums;
    polynomial-time encoding and decoding
  & Thm.~\ref{thm:adjacent-closed-form}; Sec.~\ref{sec:adjacent} \\
exactly two adjacent deletions; any code of size \(q^n\)
  & converse: at least three redundant symbols for \(n\ge3\); with exactly
    three, \(q\ge n-1\)
  & Sec.~\ref{sec:adjacent-converses} \\
at most two adjacent deletions; any code of size \(q^n\)
  & converse: three redundant symbols are needed already at \(n=2\)
  & Sec.~\ref{sec:adjacent-converses} \\
exactly \(t\ge2\) adjacent deletions; the information prefix, one XOR
parity per residue class modulo \(t\), and one final symbol
(Section~\ref{sec:adjacent-converses}), \(q=2^m\)
  & converse: \(q\ge n+\lfloor n/t\rfloor+2\) (\(n\ge2t\))
  & Sec.~\ref{sec:adjacent-converses} \\
at most two adjacent deletions; the same architecture with \(t=2\),
  \(n=4\)
  & the least admissible power of two is exactly \(q=16\); one finite base
    assignment is verified by exhaustive computation
  & Sec.~\ref{sec:adjacent-n4} \\
\bottomrule
\end{tabularx}
\end{table}

\section{Conclusion and further directions}
\label{sec:conclusion}

For every \(n\ge2\), we constructed a systematic encoder with three
redundant symbols, which may themselves be deleted, that corrects two
arbitrary deletions over
\(q<2^{32}n^7\), with polynomial-time encoding and decoding under the
computational convention of Section~\ref{sec:model}, in which running
time is counted in bit operations and all auxiliary data are computed
from \(n\).  The proof combines
two reserved bits that distinguish how many deletions hit the information
part with a coloring of the
confusion graph \(G_\star\) whose value at a single message can be
computed in polynomial time by color reduction, as in~\cite{LiFarnoud2023,LiGabrysFarnoud2025}.  A known
lower bound proves that three is the optimal number of
redundant symbols in
this model for \(n\ge2\).  Using Haxell's independent-transversal theorem,
we also proved that such codes exist over alphabets of size \(O(n^4)\);
efficient encoding and decoding there remain open.  For one adjacent burst of length at most two, a
construction combining parity symbols and Tenengolts checksums attains the
same number of redundant symbols over a power-of-two alphabet of size
\(q<(n+2)^4\).  The adjacent analysis also gives two converse
bounds, one for all codes of size \(q^n\) and one for codes with the parity
prefix of the construction, that is, the information symbols followed by
its two XOR parity symbols; for that architecture with \(n=4\), the exact
power-of-two threshold for correcting bursts of length zero, one, or two is
\(16\).
Existence at \(q=16\) uses a base assignment verified by exhaustive
computation.  The obstruction for smaller powers of two and the extension
to every larger power of two are proved~deductively.

The results suggest four open problems.

\begin{enumerate}[leftmargin=*]
\item \emph{Computing a color at a single message.}
      Is there a proper coloring of \(G_\star\) with \(O(D_n)\) colors, as
      the global greedy coloring of Proposition~\ref{prop:global-greedy}
      has, whose value at a single message is computable in polynomial
      time?  Here \(D_n=O(n^3)\) bounds the maximum degree of \(G_\star\).
      Such a coloring would lower the alphabet exponent of an efficient
      construction from seven toward four.
\item \emph{The alphabet exponent with efficient algorithms.}
      For systematic encoders with three redundant symbols that correct
      two arbitrary deletions, the insertion-ball argument of
      Proposition~\ref{prop:appended-alphabet-lower} gives
      \(q=\Omega(n^2)\). What is the correct alphabet exponent under
      polynomial-time encoding and decoding with the computational
      convention of Section~\ref{sec:model}?
\item \emph{The adjacent-burst gap.}
      For \(n\ge3\), the unrestricted exact-pair converse and the
      power-of-two systematic construction leave the bounds
      \(n-1\le q<\tfrac38(n+2)^4+\tfrac34(n+2)^3\); the upper end is a
      sharper form of \(q<(n+2)^4\), and the lower end, proved for one exact
      adjacent pair, also applies to bursts of length at most two.  In the
      prescribed parity prefix, the architectural converse for one exact
      adjacent pair gives
      \(q\ge n+\lfloor n/2\rfloor+2\) for \(n\ge4\); only the first case
      \(n=4\), for correction of lengths zero, one, or two, is proved here
      to have exact power-of-two threshold \(16\).
      Can a systematic construction with three redundant symbols
      achieve \(q=O(n)\)?  More narrowly, can one approach the bound
      \(n+\lfloor n/t\rfloor+2\) for the analogous construction with one
      parity for each of \(t\) residue classes and one further symbol,
      correcting one burst of exactly \(t\) adjacent deletions?
      Since the bound \(|U_x|\le\kappa_n\) on the set \(U_x\) of values that
      the final symbol of \(x\) must avoid is attained
      (Proposition~\ref{prop:adjacent-prefix-return-sharpness}), the
      alphabet of the present construction cannot be reduced by a better
      bound on \(|U_x|\) alone; an improvement could choose the final
      symbols of different messages jointly or change the parity prefix.
\item \emph{More than two deletions.}
      For each fixed \(k>2\) and every \(n\ge2\), Haxell's theorem gives
      a systematic code with the optimal \(k+1\) redundant
      symbols over every alphabet of size at least \(2\binom{n+k+1}{k}^2\)
      (Theorems~\ref{thm:haxell-systematic} and~\ref{thm:lower-bound}).
      Is there a construction
      with polynomial-time encoding and decoding over an alphabet of
      polynomial size with the same exact symbol count?
\end{enumerate}

All four problems concern systematic codes in the prefix form and
redundancy counted in full symbols.
Fixed-alphabet and bit-redundancy optimization form separate regimes.

\FloatBarrier

\appendix
\begingroup
\let\baselemma\lemma
\let\endbaselemma\endlemma
\renewenvironment{lemma}{
  \ifnum\value{section}=1\relax
    \ifnum\value{theorem}=1\relax
      \par\medskip
    \fi
  \fi
  \baselemma
}{\endbaselemma}
\section{Messages attaining the bound on the avoidance set}
\label{app:adjacent-sharpness}

This appendix proves the constructive direction of
Proposition~\ref{prop:adjacent-prefix-return-sharpness}.  All sums are in
the additive group of \(\F_q\cong\F_2^m\); multiplication is not used. Within this appendix, \(s\), \(r\), \(C\) and \(Z\) are local names for sizes and sets, unrelated to their meanings in the body; the labels \(A\) and \(B\) are explained below.
The additive existence statements of Lemmas~\ref{lem:zero-xor-spectrum} and~\ref{lem:zero-avoiding-xor} below are known special cases of
Bajnok--Edwards
\cite[Theorem~3.2 and Corollaries~3.4--3.5]{BajnokEdwards2017}.
We retain short proofs to keep the construction self-contained.  The
part specific to this paper is the construction of extremal messages from
these sets, at the end of the appendix.

\begin{lemma}[Zero-XOR size spectrum]
\label{lem:zero-xor-spectrum}
Let \(q\ge4\) be a power of two.  For every \(s\) with \(0\le s\le q\) and \(s\notin\{2,q-2\}\), there is a set
\(Z\subseteq\F_q\) of size \(s\) and sum zero.
The two excluded sizes are impossible.
\end{lemma}

\begin{proof}
A two-element zero-sum set would repeat its element.  Since the sum of
all elements of \(\F_q\) is zero, complementation gives the same
obstruction at size \(q-2\).

For sufficiency, induct on \(q\).  At \(q=4\), the sizes
\(0,1,3,4\) are realized by
\(\varnothing,\{0\},\F_4\setminus\{0\},\F_4\).
Write \(q=2h\ge8\), and let \(\mathsf H\) be a codimension-one subspace, so
\(|\mathsf H|=h\).  Taking complements in \(\F_q\) preserves the zero sum, so it
suffices to treat \(s\le h\).  By induction inside \(\mathsf H\), every size
\(s\le h\) with \(s\notin\{2,h-2\}\) is realized by a subset of \(\mathsf H\).
Size \(2\) is excluded anyway, so only \(s=h-2\) remains.

For \(h\ge8\), choose linearly independent
\(e_1,e_2,e_3\in\mathsf H\) and choose
\(v\notin\mathsf H\).  Then
\[
 \bigl(\mathsf H\setminus\{0,e_1,e_2,e_3\}\bigr)
 \cup\{v,v+e_1+e_2+e_3\}
\]
has size \(h-2\) and sum zero.  When \(h=4\), this is the already
excluded size two.  The induction follows.
\end{proof}

\begin{lemma}[Zero-avoiding zero-XOR sets]
\label{lem:zero-avoiding-xor}
If \(q\ge8\) and \(3\le s\le q-4\), then there is a set
\(Z\subseteq\F_q\setminus\{0\}\) of size \(s\) and sum zero.
\end{lemma}

\begin{proof}
For even \(s\), take a zero-sum \(s\)-set \(T\) from
Lemma~\ref{lem:zero-xor-spectrum}, choose \(v\notin T\), and translate to
\(T+v\).  Translation removes zero and preserves the sum because \(s\)
is even.  For odd \(s\), the number \(q-1-s\) is even and lies between
four and \(q-4\).  Apply the even case just proved to construct a
zero-avoiding zero-sum set of that size, and take its complement inside
\(\F_q\setminus\{0\}\).
\end{proof}

\begin{proof}[Proof of the constructive direction of Proposition~\ref{prop:adjacent-prefix-return-sharpness}]
The endpoints are immediate.  For \(n=1\), take \(x=(1)\), giving
\(U_x=\{0,1\}\).  For \(n=2\), take \(x=(0,1)\), again giving
\(U_x=\{0,1\}\).  For \(n=3\), the condition \(q\ge \kappa_3=3\) and the
power-of-two hypothesis give \(q\ge4\).  In the binary labels of Section~\ref{sec:adjacent}, take \(x=(0,1,2)\); then
\(A(x)=1\), \(B(x)=2\), and \(U_x=\{0,1,2\}\) by
Proposition~\ref{prop:adjacent-forbidden-formula}.

We now treat the remaining lengths, in the interleaved coordinates of Proposition~\ref{prop:adjacent-forbidden-formula}: \(x=(a_1,b_1,\ldots,a_k,b_k)\) for \(n=2k\), and \(x=(b_0,a_1,b_1,\ldots,a_k,b_k)\) for \(n=2k+1\).  In each case the base \(\{a_1,\ldots,a_k,A(x),B(x)\}\) is made to have \(k+2\) values, and the set \(J\) of Proposition~\ref{prop:adjacent-forbidden-formula} is made equal to a prescribed set, every index \(j\) of which is new, that is, has \(b_j\) outside the base.  To this end, \(a_1,\ldots,a_k,A(x)\) form the complement of a zero-sum set, and the prefix sums \(\psi_j\), which determine \(J\), are chosen before the \(b_j\).

First suppose either \(n=2k\ge4\),
or \(n=2k+1\ge5\) with \(k\) even.  Define
\[
 (r,J)=
 \begin{cases}
 \bigl(\lfloor(k-1)/2\rfloor,\{2,4,\ldots,2r\}\bigr),
       &n=2k,\\[2mm]
 \bigl(k/2,\{1,3,\ldots,k-1\}\bigr),
       &n=2k+1,\ k\ {\rm even}.
 \end{cases}
\]
Thus \(\kappa_n=k+r+2\).  Put \(s=q-k-1\).  Since
\(q\ge \kappa_n\), we have \(s\ge r+1\).  Also \(s\ne q-2\), because
\(k\ge2\).  If \(s=2\), then \(q=k+3\) and \(r\le1\), so \(k\le4\) (\(k=2\) for odd \(n\)); the only
possibilities are \(q\in\{5,6,7\}\), none a power of two.  Hence, as \(q\ge\kappa_n\ge4\),
Lemma~\ref{lem:zero-xor-spectrum} supplies a zero-sum \(s\)-set \(Z\).
(In each case below, the labels \(A\) and \(B\) are chosen first, and
the message \(x\) is built so that \(A(x)=A\) and \(B(x)=B\).)  Choose
\(A\notin Z\) and enumerate
\[
 \F_q\setminus(Z\cup\{A\})=\{a_1,\ldots,a_k\}.
\]
Then \(\sum_i a_i=A\).  In \(Z\), the complement of
\(\{a_1,\ldots,a_k,A\}\), choose pairwise distinct labels
\[
 B,\qquad c_j\quad(j\in J);
\]
there are enough because \(s\ge r+1\).

We now choose the prefix sums \(\psi_0,\ldots,\psi_k\) of
\eqref{eq:adjacent-even-forbidden-form} and
\eqref{eq:adjacent-odd-forbidden-form}, called \emph{states} below, and
then define the symbols \(b_j\) from them.  For even \(n\), set \(\psi_0=0\);
for odd \(n\), \(\psi_0\) is chosen together with the other states.
Prescribe
\[
 \psi_j=0\quad(j\in J),\qquad
 \psi_{j-1}=a_j+c_j\quad(j\in J),\qquad
 \psi_k=A+B,
\]
and give every state not yet set a nonzero value.  The prescribed nonzero
states are indeed nonzero because all displayed labels are distinct, and
their indices are disjoint (for \(n=2k\), \(J\) consists of even and \(J-1\) of odd indices, all in \(\{1,\ldots,k-1\}\); for \(n=2k+1\), \(J\) consists of odd and \(J-1\) of even indices, all in \(\{0,\ldots,k-1\}\)).  Define
\[
 b_j=a_j+\psi_{j-1}+\psi_j\quad(1\le j\le k),
\]
and, in the odd case, put \(b_0=\psi_0\).  Telescoping gives the prescribed
prefix states and
\[
 \sum_{j=1}^k b_j=B\quad(n=2k),\qquad
 b_0+\sum_{j=1}^k b_j=B\quad(n=2k+1).
\]
Moreover \(b_j=c_j\) for every \(j\in J\). Every other state is nonzero, so the zero states of \(x\) are exactly those prescribed, and the set \(J\) of Proposition~\ref{prop:adjacent-forbidden-formula} for this message is the set \(J\) above (for even \(n\), that proposition does not count \(\psi_0\)).  Equations
\eqref{eq:adjacent-even-forbidden-form} and
\eqref{eq:adjacent-odd-forbidden-form} therefore give the disjoint union
\[
 U_x=\{a_1,\ldots,a_k,A,B\}\mathbin{\dot\cup}
      \{c_j:j\in J\},
\]
so \(|U_x|=k+2+r=\kappa_n\).

It remains to consider \(n=2k+1\ge7\) with \(k\) odd.  Put
\[
 r=(k+1)/2,\qquad J=\{0,2,4,\ldots,k-1\}.
\]
Every message attaining \(\kappa_n\) has this \(J\), and this is why the case needs Lemma~\ref{lem:zero-avoiding-xor}: for odd \(k\), the only way to have \(r\) new
indices in \(\{0,\ldots,k-1\}\),
which are pairwise non-consecutive by the proof of
Proposition~\ref{prop:adjacent-prefix-return-bound}, is this \(J\), which
contains \(0\); since \(b_0=\psi_0=0\), the index \(0\) is then new only if the
base avoids \(0\), which Lemma~\ref{lem:zero-avoiding-xor} provides.  The index \(k\) cannot help: by the same proof, \(k\in J\) would
force \(A(x)=B(x)\), as \(\psi_k=A(x)+B(x)\), and the base would lose a value, and \(b_k=a_k\) would not be new, since \(k-1\in J\).  Nor can an odd \(j<k\) lie in \(J\), since \(j+1\), an even index of this \(J\), would then not be new.
Again \(\kappa_n=k+r+2\).
Here
\[
 3\le q-k-2\le q-4.
\]
Indeed the lower bound follows from \(q\ge \kappa_n\) when \(k\ge5\);
for \(k=3\), the power-of-two condition strengthens \(q\ge7\) to
\(q\ge8\).  Lemma~\ref{lem:zero-avoiding-xor} gives a zero-sum
\((q-k-2)\)-set \(C\) avoiding zero.  Put \(Z=\{0\}\cup C\), choose a
nonzero \(A\notin Z\), and enumerate
\[
 \F_q\setminus(Z\cup\{A\})=\{a_1,\ldots,a_k\}.
\]
Then the \(a_i\) and \(A=\sum_i a_i\) are distinct and nonzero.
Outside \(\{0,a_1,\ldots,a_k,A\}\), choose distinct
\[
 B,\qquad c_j\quad(j\in J\setminus\{0\});
\]
exactly \(r\) labels are required and \(q-k-2\ge r\) are available.

Set \(\psi_0=0\), and for positive \(j\in J\) prescribe
\[
 \psi_j=0,\qquad \psi_{j-1}=a_j+c_j,\qquad \psi_k=A+B,
\]
which sets every state.  Define \(b_0=\psi_0=0\) and
\(b_j=a_j+\psi_{j-1}+\psi_j\).  The same telescoping calculation applies,
and the elements \(b_j\), \(j\in J\), which are the extra members of \(U_x\)
in \eqref{eq:adjacent-odd-forbidden-form}, are \(b_0=0\) and \(b_j=c_j\) for
positive \(j\in J\).  Hence
\[
 U_x=\{a_1,\ldots,a_k,A,B\}\mathbin{\dot\cup}\{0\}
      \mathbin{\dot\cup}\{c_j:j\in J\setminus\{0\}\},
\]
whose size is \(k+2+r=\kappa_n\).  This completes every case.
\end{proof}

\endgroup
\section{Proof of the two-round coloring lemma}
\label{app:two-round-coloring}

\begin{proof}[Proof of Lemma~\ref{lem:two-round-local-coloring}]
If \(D=0\), the constant color proves the claim.  Hence assume \(D\ge1\).
Every auxiliary field of order \(2^r\) used below is realized canonically
as in Section~\ref{sec:model}, by the polynomial-basis model
\(\F_2[X]/(f_r(X))\), and its elements are ordered by their binary labels.
Thus every occurrence of ``first'' below is public and canonical; finding
\(f_r\) by exhaustive irreducibility tests is polynomial in the displayed
field size.

Let \(q_1\) be the least power of two strictly larger than
\(D\ell_{\mathrm{name}}\).  Inject the public names into the polynomials of
degree at most \(\ell_{\mathrm{name}}\) over \(\F_{q_1}\).  This is possible
because \(q_1^{\ell_{\mathrm{name}}+1}\ge2^{\ell_{\mathrm{name}}}\), and a
canonical injection is obtained from the base-\(q_1\) coefficient vector of
the name, read as a binary integer whose first bit is least significant: its
base-\(q_1\) digit \(d_i\) of place value \(q_1^i\) becomes the coefficient
\(\operatorname{lab}^{-1}(d_i)\) of degree \(i\), and likewise with \(q_2\) in the
second round. The degree bound \(\ell_{\mathrm{name}}\) is generous, since a name has fewer base-\(q_1\) digits; it keeps the choice of \(q_1\) simple and affects only the second term of \eqref{eq:two-round-palette}.  At a vertex \(v\), choose the first \(t\in\F_{q_1}\) at which its
polynomial differs from every neighbor polynomial.  Each neighbor excludes
at most \(\ell_{\mathrm{name}}\) field points, whereas
\(q_1>D\ell_{\mathrm{name}}\).  A point exists, and the pair consisting of
the selected point and polynomial value is a proper first-round color in a
palette of size \(q_1^2\).  (If two adjacent vertices received the same pair,
their polynomials would agree at the point that each of them chose, contrary
to that choice.)

Let \(q_2\) be the least power of two at least
\[
 \max\{3D+1,\lceil\sqrt{q_1}\rceil\}.
\]
For a first-round color \((a,b)\in\F_{q_1}^2\), let
\(\nu=\operatorname{lab}(a)q_1+\operatorname{lab}(b)\), where
\(\operatorname{lab}\) is the binary label of
Section~\ref{sec:model}.  Expand \(\nu\) in base \(q_2\) as four digits, padding on the left
with zeros, and, as in the first round, let the digit \(d_i\) of place value \(q_2^i\) give the coefficient \(\operatorname{lab}^{-1}(d_i)\) of degree \(i\) in a polynomial of
degree at most three over \(\F_{q_2}\).  Since \(q_2^4\ge q_1^2\), this is a
canonical injection of the first-round colors.  Four digits are the fewest
that let \(q_2\) be of order \(\sqrt{q_1}\), so that \(q_1\) enters \(q_2^2\) only
linearly, while each further digit would add \(D\) excluded points.  The choice of four digits is not claimed to be optimal for \(G_\star\).  The same
rule, choosing the first point at which the polynomial differs from all
neighbor polynomials, excludes at most \(3D<q_2\) points and yields a
proper color in \(\F_{q_2}^2\).

It remains to bound the palette \(q_2^2\).  Since \(q_2\) is the least power
of two at least \(\max\{3D+1,\lceil\sqrt{q_1}\rceil\}\), we have
\(q_2<2\max\{3D+1,\lceil\sqrt{q_1}\rceil\}\), hence
\(q_2^2\le4\bigl((3D+1)^2+\lceil\sqrt{q_1}\rceil^2\bigr)\).  Moreover
\(3D+1\le4D\); \(q_1\le2D\ell_{\mathrm{name}}\) by minimality of \(q_1\);
and \(\lceil\sqrt{q_1}\rceil^2\le2q_1\) because \(q_1\) is a power of two, so that \(\sqrt{q_1}\) or \(\sqrt{2q_1}\) is an
integer.
Therefore
\[
 \begin{split}
 q_2^2
 &\le4\bigl((3D+1)^2+\lceil\sqrt{q_1}\rceil^2\bigr)
  \le64D^2+8q_1\\
 &\le64D^2+16D\ell_{\mathrm{name}}.
 \end{split}
\]
The first-round color of \(v\) uses the names of its neighbors, which the
neighborhood algorithm returns, as Proposition~\ref{prop:gstar-indexed-oracle}
does for \(G_\star\).  The second
round uses the first-round colors of those neighbors and therefore only
radius-two data; a third round would need radius-three data, and two rounds
already fit the \(6u+13\) bits of
Proposition~\ref{prop:deterministic-residual-color}.  Since both auxiliary fields and all tested point sets have
size polynomial in \(D\) and \(\ell_{\mathrm{name}}\), one color is computed from at most \(D+1\)
neighborhood computations, for the vertex and for each of its neighbors, and
further work polynomial in \(D\) and \(\ell_{\mathrm{name}}\).
\end{proof}

\section{Witnesses for the configuration at \texorpdfstring{\(q=8\)}{q=8}}
\label{app:n4-witnesses}

This appendix lists the data needed to check by hand the eleven-message
configuration in the proof of Theorem~\ref{thm:adjacent-n4-exact-threshold}.
Field elements of \(\F_8\) are written by their binary labels
\(0,\ldots,7\), and addition is XOR.  For a message \(x\), the prefix is
\(W(x)=(x_1,x_2,x_3,x_4,x_1\oplus x_3,x_2\oplus x_4)\).  To avoid a clash
with other notation, the eleven messages are written in sans-serif.

Table~\ref{tab:n4-edges} gives, for each of the \(33\) edges used in the
proof, a common child of the two prefixes: \(r=1\) means one deletion and
\(r=2\) means deletion of an adjacent pair, and the last column gives the
start positions \(d\) in the two prefixes, in the order of the edge.

\begin{table}[ht]
\centering
\caption{A common child for each edge of the \(q=8\) configuration.}
\label{tab:n4-edges}
\small
\setlength{\tabcolsep}{3pt}
\begin{tabular}{@{}cccc|cccc|cccc@{}}
\toprule
edge & \(r\) & child & \(d\) & edge & \(r\) & child & \(d\) & edge & \(r\) & child & \(d\) \\
\midrule
\(\mathsf{a}\mathsf{b}\) & 2 & \((4,3,1,1)\) & 2, 5 & \(\mathsf{b}\mathsf{g}\) & 2 & \((4,3,1,1)\) & 5, 1 & \(\mathsf{f}\mathsf{i}\) & 2 & \((5,3,1,2)\) & 2, 5 \\
\(\mathsf{a}\mathsf{c}\) & 2 & \((4,2,5,1)\) & 4, 2 & \(\mathsf{c}\mathsf{d}\) & 2 & \((4,3,1,2)\) & 5, 3 & \(\mathsf{f}\mathsf{j}\) & 2 & \((5,1,1,2)\) & 3, 2 \\
\(\mathsf{a}\mathsf{e}\) & 2 & \((4,2,1,1)\) & 3, 2 & \(\mathsf{c}\mathsf{e}\) & 2 & \((4,3,1,1)\) & 4, 3 & \(\mathsf{g}\mathsf{h}\) & 2 & \((5,3,1,1)\) & 2, 5 \\
\(\mathsf{a}\mathsf{g}\) & 2 & \((4,3,1,1)\) & 2, 1 & \(\mathsf{c}\mathsf{f}\) & 2 & \((4,3,1,2)\) & 5, 1 & \(\mathsf{g}\mathsf{i}\) & 2 & \((5,2,4,1)\) & 4, 2 \\
\(\mathsf{a}\mathsf{h}\) & 2 & \((5,3,1,1)\) & 1, 5 & \(\mathsf{c}\mathsf{g}\) & 2 & \((4,3,1,1)\) & 4, 1 & \(\mathsf{g}\mathsf{k}\) & 2 & \((5,2,1,1)\) & 3, 2 \\
\(\mathsf{a}\mathsf{i}\) & 2 & \((5,3,1,1)\) & 1, 4 & \(\mathsf{d}\mathsf{e}\) & 1 & \((4,3,5,1,1)\) & 6, 4 & \(\mathsf{h}\mathsf{i}\) & 2 & \((5,3,1,1)\) & 5, 4 \\
\(\mathsf{a}\mathsf{k}\) & 2 & \((4,2,1,1)\) & 3, 1 & \(\mathsf{d}\mathsf{f}\) & 2 & \((4,3,1,2)\) & 3, 1 & \(\mathsf{h}\mathsf{j}\) & 1 & \((5,3,1,1,2)\) & 5, 3 \\
\(\mathsf{b}\mathsf{c}\) & 2 & \((4,3,1,1)\) & 5, 4 & \(\mathsf{d}\mathsf{j}\) & 2 & \((4,1,1,2)\) & 2, 1 & \(\mathsf{h}\mathsf{k}\) & 2 & \((5,3,1,1)\) & 5, 3 \\
\(\mathsf{b}\mathsf{d}\) & 1 & \((4,3,1,1,2)\) & 5, 3 & \(\mathsf{e}\mathsf{g}\) & 2 & \((4,3,1,1)\) & 3, 1 & \(\mathsf{i}\mathsf{j}\) & 2 & \((5,3,1,2)\) & 5, 3 \\
\(\mathsf{b}\mathsf{e}\) & 2 & \((4,3,1,1)\) & 5, 3 & \(\mathsf{e}\mathsf{k}\) & 2 & \((4,2,1,1)\) & 2, 1 & \(\mathsf{i}\mathsf{k}\) & 2 & \((5,3,1,1)\) & 4, 3 \\
\(\mathsf{b}\mathsf{f}\) & 2 & \((4,3,1,2)\) & 4, 1 & \(\mathsf{f}\mathsf{h}\) & 2 & \((5,1,4,2)\) & 4, 2 & \(\mathsf{j}\mathsf{k}\) & 1 & \((5,3,4,1,1)\) & 6, 4 \\
\bottomrule
\end{tabular}
\end{table}

Table~\ref{tab:n4-lists} shows, for each message, why the values \(1\),
\(4\), and \(5\) lie in \(\mathcal L_x\).  An entry ``one, \(d\)'' refers to
the first clause of \eqref{eq:adjacent-exact-list} with \(\mathsf B_{1,d}\), and
``two, \(d\)'' to the second clause with \(\mathsf B_{2,d}\); in each case the
resulting message \(y\) differs from \(x\): since each message has
\(x_1\in\{4,5\}\) and \(x_2\in\{1,2,3\}\), we have \(y_1=x_2\ne x_1\) for
``one, \(d=1\)'', \(y_3=x_1\oplus x_3\ne x_3\) for ``two, \(d=2\)'', and
\(y_4=x_2\oplus x_4\ne x_4\) for ``two, \(d=4\)''.  These three deletions put \(x_1\), \(x_3\) and \(x_1\oplus x_3\), respectively, into \(\mathcal L_x\); for each of the eleven messages these values are \(1\), \(4\), and \(5\) in some order, since \(x_3\in\{1,4,5\}\setminus\{x_1\}\) and \(1\oplus4\oplus5=0\).  Thus the three value columns follow from \(x_1\) and \(x_3\).

\begin{table}[ht]
\centering
\caption{Deletions that place \(1\), \(4\), and \(5\) in \(\mathcal L_x\).}
\label{tab:n4-lists}
\small
\begin{tabular}{@{}lccc@{}}
\toprule
message & value \(1\) & value \(4\) & value \(5\) \\
\midrule
\(\mathsf{a}=(4,2,5,3)\) & two, \(d=4\) & one, \(d=1\) & two, \(d=2\) \\
\(\mathsf{b}=(4,3,1,1)\) & two, \(d=2\) & one, \(d=1\) & two, \(d=4\) \\
\(\mathsf{c}=(4,3,1,2)\) & two, \(d=2\) & one, \(d=1\) & two, \(d=4\) \\
\(\mathsf{d}=(4,3,5,1)\) & two, \(d=4\) & one, \(d=1\) & two, \(d=2\) \\
\(\mathsf{e}=(4,3,5,2)\) & two, \(d=4\) & one, \(d=1\) & two, \(d=2\) \\
\(\mathsf{f}=(5,1,4,3)\) & two, \(d=4\) & two, \(d=2\) & one, \(d=1\) \\
\(\mathsf{g}=(5,2,4,3)\) & two, \(d=4\) & two, \(d=2\) & one, \(d=1\) \\
\(\mathsf{h}=(5,3,1,1)\) & two, \(d=2\) & two, \(d=4\) & one, \(d=1\) \\
\(\mathsf{i}=(5,3,1,2)\) & two, \(d=2\) & two, \(d=4\) & one, \(d=1\) \\
\(\mathsf{j}=(5,3,4,1)\) & two, \(d=4\) & two, \(d=2\) & one, \(d=1\) \\
\(\mathsf{k}=(5,3,4,2)\) & two, \(d=4\) & two, \(d=2\) & one, \(d=1\) \\
\bottomrule
\end{tabular}
\end{table}

\FloatBarrier
\section{Proof of Lemma~\ref{lem:binary-vt-separation}}
\label{app:vt-proof}

\begin{proof}[Proof of Lemma~\ref{lem:binary-vt-separation}]
Let \(z\in\{0,1\}^{\ell_{\mathrm{VT}}-1}\) be a common one-deletion
descendant.  For a bit \(\varepsilon\) and \(1\le i\le\ell_{\mathrm{VT}}\),
write
\[
 I_i(z,\varepsilon)=(z_1,\ldots,z_{i-1},\varepsilon,z_i,\ldots,
 z_{\ell_{\mathrm{VT}}-1})
\]
for the word in which \(\varepsilon\) becomes the \(i\)-th coordinate.
Then \(c=I_i(z,\varepsilon)\) and \(c'=I_j(z,\varepsilon')\) for some \(i,j\) and bits \(\varepsilon,\varepsilon'\).
Since the lemma's hypothesis and conclusion are symmetric in \(c\) and \(c'\), we may assume \(i\le j\).  The two words
agree outside the coordinates \(i,\ldots,j\).  On these coordinates,
\(c\) reads \(\varepsilon,z_i,\ldots,z_{j-1}\) and \(c'\) reads
\(z_i,\ldots,z_{j-1},\varepsilon'\); thus each \(z_r\) with
\(i\le r\le j-1\) has coefficient \(r+1\) in the sum for \(c\) and coefficient \(r\) in the
sum for \(c'\); when \(i=j\), there is no such \(z_r\).
Hence
\begin{equation}
 \sum_{r=1}^{\ell_{\mathrm{VT}}} r c_r
 -\sum_{r=1}^{\ell_{\mathrm{VT}}} r c'_r
 =i\varepsilon-j\varepsilon'+\sum_{r=i}^{j-1}z_r.
 \label{eq:binary-vt-difference}
\end{equation}
We show that this integer lies in
\(\{-\ell_{\mathrm{VT}},\ldots,\ell_{\mathrm{VT}}\}\setminus\{0\}\).
\begin{itemize}[leftmargin=*]
\item If \(\varepsilon=\varepsilon'=0\), it equals
  \(\sum_{r=i}^{j-1}z_r\in\{0,\ldots,j-i\}\).  It vanishes only when all
  \(z_r\) with \(i\le r<j\) are zero, and then \(c=c'\).
\item If \(\varepsilon=\varepsilon'=1\), it equals
  \(-(j-i)+\sum_{r=i}^{j-1}z_r\in\{-(j-i),\ldots,0\}\).  It vanishes
  only when all these \(z_r\) are one, and then again \(c=c'\).
\item If \(\varepsilon=1\) and \(\varepsilon'=0\), it lies in
  \(\{i,\ldots,j\}\subseteq\{1,\ldots,\ell_{\mathrm{VT}}\}\).
\item If \(\varepsilon=0\) and \(\varepsilon'=1\), it lies in
  \(\{-j,\ldots,-i\}\subseteq\{-\ell_{\mathrm{VT}},\ldots,-1\}\).
\end{itemize}
The value is therefore nonzero in all four cases: in the last two because their ranges exclude \(0\), and in the first two because \(c\ne c'\).  A nonzero integer of absolute
value at most \(\ell_{\mathrm{VT}}\) cannot vanish modulo
\(\ell_{\mathrm{VT}}+1\).
\end{proof}

\ifdefined\TITWRAPPER
\section*{Acknowledgment}
OpenAI ChatGPT and OpenAI Codex were used as assistive tools in exploring
proof strategies and checking algebraic identities and small finite cases in
Sections~II--IX and the appendices, in locating related literature for
Sections~I and X, in preparing the verification software described above,
and in editing English and \LaTeX{} throughout. The authors critically
evaluated all outputs, independently verified the final statements, proofs,
calculations, software results, and citations, and take full responsibility
for the manuscript.
\fi

\ifdefined\TITWRAPPER
  \expandafter\endinput
\fi

\phantomsection\addcontentsline{toc}{section}{References}
\bibliographystyle{plain}
\bibliography{references}

\end{document}